\documentclass[prd,aps,amsfonts,eqsecnum,superscriptaddress,nofootinbib,notitlepage,longbibliography,final]{revtex4-1}

\usepackage[english]{babel}
\usepackage[utf8]{inputenc}

\usepackage[a4paper, left=2cm, right=2cm, top=2cm, bottom=2cm]{geometry}

\usepackage{amsmath,physics,amsfonts}
\usepackage{graphicx,subcaption}
\usepackage[colorlinks=true, allcolors=blue]{hyperref}
\usepackage{float}
\usepackage{hhline}
\usepackage{mathtools}
\mathtoolsset{showonlyrefs}
\usepackage{soul}

\usepackage{pgfplots}
\pgfplotsset{compat=newest}

\usepgfplotslibrary{groupplots}

\definecolor{zero}{RGB}{180, 180, 180}
\definecolor{one}{RGB}{206, 104, 104}
\definecolor{two}{RGB}{231, 195, 146}
\definecolor{three}{RGB}{231, 217, 146}
\definecolor{four}{RGB}{209, 223, 141}
\definecolor{five}{RGB}{116, 185, 116}
\definecolor{six}{RGB}{126, 147, 165}
\definecolor{seven}{RGB}{146, 136, 176}
\definecolor{eight}{RGB}{166, 124, 166}

\usepackage{amsthm}

\newtheorem{theorem}{Theorem}[section]
\newtheorem{corollary}{Corollary}[theorem]
\newtheorem{lemma}[theorem]{Lemma}
\newtheorem{prop}[theorem]{Proposition}

\newtheorem{definition}[theorem]{Definition}
\newtheorem{example}[theorem]{Example}
\theoremstyle{definition}
\newtheorem{remark}[theorem]{Remark}

\DeclareMathOperator{\id}{\mathbf{1}}

\renewcommand\labelenumi{(\roman{enumi})}
\renewcommand\theenumi\labelenumi

\begin{document}

\title{Polynomial Time Quantum Gibbs Sampling for Fermi-Hubbard Model at any Temperature}
\author{Štěpán Šmíd}
\email{s.smid23@imperial.ac.uk}
\affiliation{%
Department of Computing, Imperial College London, United Kingdom
}%
\author{Richard Meister}
\affiliation{%
Department of Computing, Imperial College London, United Kingdom
}%
\author{Mario Berta}
\affiliation{%
Institute for Quantum Information, RWTH Aachen University, Germany
}%
\affiliation{%
Department of Computing, Imperial College London, United Kingdom
}%
\author{Roberto Bondesan}
\affiliation{%
Department of Computing, Imperial College London, United Kingdom
}%

\date{\today}


\begin{abstract}
    Recently, there have been several advancements in quantum algorithms for Gibbs sampling. These algorithms simulate the dynamics generated by an artificial Lindbladian, which is meticulously constructed to obey a detailed-balance condition with the Gibbs state of interest, ensuring it is a stationary point of the evolution, while simultaneously having efficiently implementable time steps. The overall complexity then depends primarily on the mixing time of the Lindbladian, which can vary drastically, but which has been previously bounded in the regime of high enough temperatures [Rouz\'e {\it et al.}~arXiv:2403.12691 and arXiv:2411.04885].

    In this work, we calculate the spectral gap of the Lindbladian for free fermions using third quantisation, and also prove a logarithmic bound on its mixing time by analysing corresponding covariance matrices. Then we prove a constant gap of the perturbed Lindbladian corresponding to interacting fermions up to some maximal coupling strength. This is achieved by using theorems about stability of the gap for lattice fermions. Our methods apply at any constant temperature and independently of the system size.
    The gap then provides an upper bound on the mixing time, and hence on the overall complexity of the quantum algorithm, proving that the purified Gibbs state of weakly interacting (quasi-)local fermionic systems of any dimension can be prepared in $\widetilde{\mathcal{O}} (n^3 \operatorname{polylog}(1/\epsilon))$ time on $\mathcal{O}(n)$ qubits, where $n$ denotes the size of the system and $\epsilon$ the desired accuracy. As an application of Gibbs sampling, we explain how to calculate partition functions for the considered systems.
    We provide exact numerical simulations for small system sizes supporting the theory and also identify different suitable jump operators and filter functions for the sought-after regime of intermediate coupling in the Fermi-Hubbard model.
\end{abstract}

\maketitle


\section{Overview} \label{sec:overview}

\paragraph{Introduction.} Quantum computers promise to have a transformative impact on computing, as quantum algorithms are believed to solve certain computational tasks much faster, i.e., with {\it superpolynomial speed-up} compared to purely classical algorithms. Although it is important to realise that such large quantum advantages are anything but generic, one of the most promising areas of application is the simulation of quantum many-body systems (QMBS) \cite{dalzell2023algorithmsreview}. Here, while quantum proposals on how to resolve the zero temperature ground state physics of QMBS range back to the early age of quantum computing \cite{Kiatev2002classicalquantumcomputation}, there has been recent algorithmic progress on how to efficiently create non-zero temperature quantum Gibbs states via simulated Lindbladian evolution in the quantum circuit model \cite{chen2023quantum,chen2023efficient}. These {\it quantum Gibbs samplers (QGS)} correspond to the non-commutative analogue of the classically highly successful {\it Markov chain Monte Carlo (MCMC)} methods \cite{levin2017MCMT}, and are then also hoped to efficiently provide insights into the non-zero temperature physics of QMBS in regimes that are computationally challenging for classical methods. There is the physical intuition that QGS starting from the breakthrough result \cite{chen2023efficient} will perform well for some average case instances relevant to computational physics and computational chemistry.\footnote{The average case character of this promise would then also not contradict known worst case hardness results such as \cite{Kiatev2002classicalquantumcomputation,Aharonov:2009aa,schuch2009computational,PRXQuantum.3.020322}.} This is contrasted to the previously proposed intricate non-zero temperature quantum methods that are (partially) missing rigorous guarantees \cite{Temme:2011aa,Yung2012metropolis,shtanko2021,moussa2022lowdepthquantummetropolisalgorithm,Rall2023thermalstate,Wocjan:2023aa} (however, see also \cite{jiang2024weakmeasurements}), or, are believed to be computationally expensive on relevant finite-size instance sizes for finite temperature \cite{narayan2017qbshittingtime,vanApeldoorn2020quantumsdpsolvers,vanapeldoorn2019spdsolvers,gilyen2019qsvt,an2023quantumalgorithmlinearnonunitary}.\footnote{We refer to \cite{zhang2023dissipativequantumgibbssampling} for another potentially more near-term approach.} The guiding idea behind the latest algorithmic Lindbladian constructions is to efficiently simulate (fast) thermalisation processes in nature as, e.g., modelled by the Davies generator (cf.~\cite{Mozgunov2020completelypositive,PhysRevB.102.115109} for some recent discussions). In particular, it is possible to bring together {\it algorithmic efficiency} with an exact notion of {\it quantum detailed balanced}, and we refer to the recent QGS frameworks \cite{ding2024efficient,gilyen2024glauber} as well as references therein for an extended discussion.\\

\paragraph{Motivation.} Classical MCMC algorithms are usually termed efficient when they converge to the Gibbs state in time polynomial or even logarithmic in system size (see \cite{levin2017MCMT} and references therein). In contrast, the recently proposed Lindbladian evolution-based QGS are only efficient in the sense that all algorithmic steps are implemented in basically linear time, but the complexity or overall run-time further relies on the so-called {\it mixing time}, which then has to be estimated for each system of interest on a case-by-case basis. The situation is akin to quantum adiabatic algorithms, whose runtime is governed by the Hamiltonian gap along the adiabatic path (see \cite{RevModPhys.90.015002} and references therein).

Starting from classical MCMC it is well understood that rigorous bounds on the convergence time as studied in mathematical physics (see, e.g, \cite{temme2015thermalize,Kastoryano:2016aa,Brandao:2019aa,Bardet:2024aa,kochanowski2024rapid} and references therein) often vastly overestimate the observed convergence times when running the algorithms in practice and reading out physical information, such as relevant observables or correlation functions. However, it is by design challenging to numerically run QGS to estimate practical mixing times\,---\,after all we do not (yet) have reliable large-scale quantum computers and the classical simulation thereof is believed to be hard. As such, analytical bounds on mixing times are the current main tool to start any efficiency analysis.\\

\paragraph{Approach.} For our work, we are looking for QMBS with instances of QGS where
\begin{enumerate}
    \item classical methods are not sufficient to conclusively determine the physics at non-zero temperature
    \item we can derive rigorous and efficient bounds on the mixing time, at least for certain non-trivial parameter regimes
    \item we can subsequently give informed heuristics on how to fine-tune QGS to potentially even expect efficient mixing times for relevant average case parameter regimes beyond analytical worst case guarantees.
\end{enumerate}
As an example, previous general results on efficient quantum mixing times are generically available in the high-temperature limit \cite{rouze2024efficient,rouze2024optimalquantumalgorithmgibbs}, where, however, classical methods can also work \cite{bakshi2024hightemperature}.\footnote{Other problem specific rigorous results on efficient mixing time bounds include random sparse Hamiltonians \cite{ramkumar2024mixingtimequantumgibbs}, random local Hamiltonians \cite{basso2024random}, the toric code \cite{ding2024toriccode}, parent Hamiltonians of shallow quantum circuits \cite{chen2024constanttemp,rajakumar2024contantlocal}, among others.}

In contrast, here we focus on the quantum simulation of {\it fermionic systems} and specifically the {\it Fermi-Hubbard model} on a $D$-dimensional lattice. This model enjoys widespread applications in science, for example, to the Mott metal-insulator transition and to high temperature superconductivity \cite{baertschi2024usecaseslosalamos}, while at the same time remaining challenging for classical methods \cite{Arovas_2022,Qin2022hubbardmodel}. As such, it further serves as a standard benchmark for computational methods \cite{chan2024quantumchemistryclassicalheuristics}, including quantum simulations using ultra-cold atoms \cite{gross2017quantum}.
We note that the sought-after {\it average case efficiency} for typical finite system sizes would not be in contradiction with the known worst-case hardness results for fermionic systems \cite{PRXQuantum.3.020322,schuch2009computational}.\\

\paragraph{Results.} Our main finding (Theorem \ref{thm: main, gap} together with Corollary \ref{cor:main-result}) is for fermionic Hamiltonians with exponentially decaying interactions. Namely, we show that for such systems of size $n$ and at any constant temperature $T>0$, there exists a constant maximal interaction strength such that the corresponding purified Gibbs state can be created up to approximation $\epsilon>0$ in trace distance 
\begin{align}\label{eq:main-result}
\text{in quantum gate complexity $\widetilde{\mathcal{O}}(n^3 \operatorname{polylog}(1/\epsilon))$ and with $\mathcal{O}(n)$ qubits,}
\end{align}
where the $\widetilde{\mathcal{O}}(\cdot)$ notation absorbs subleading polylogarithmic terms. Crucially, the constant on the interaction strength is {\it independent of the system size} and thus we conclude that (weakly) interacting fermionic systems in fixed dimension and at any constant temperature can be efficiently simulated on quantum computers! Although free fermions are efficiently solvable, to the best of our knowledge, there is no provably efficient classical algorithm for the weakly interacting regime.\\
\indent Our results include, in particular, the {\it Fermi-Hubbard model}. Its spinful version on a $D$-dimensional lattice is governed by the Hamiltonian
\begin{align}
    H_\mathrm{FH} = -t \sum_{\langle i, j \rangle, \sigma} \left(a_{i,\sigma}^\dagger a_{j,\sigma} + a_{j,\sigma}^\dagger a_{i,\sigma}\right) + U \sum_i a_{i,\uparrow}^\dagger a_{i,\uparrow} a_{i,\downarrow}^\dagger a_{i,\downarrow},
\end{align}
where $\langle \cdot, \cdot \rangle$ denotes neighbouring sites on the lattice, $\sigma \in \{\uparrow, \downarrow\}$ the spins, and $a_{i,\sigma}^{(\dagger)}$ the fermionic annihilation and creation operators on site $i$ with spin $\sigma$. The weakly interacting limit corresponds to $U/t\lesssim1$ \cite{Arovas_2022,Qin2022hubbardmodel,chan2024quantumchemistryclassicalheuristics}, which can then serve as an analytical starting point for further numerical investigations. 
Furthermore, while the Fermi-Hubbard model is exactly solvable for the $D=1$ case \cite{essler2005one}, we emphasize that our results are equally valid in any dimension. 
We refer to \cite{Arovas_2022} for a recent review of exact and heuristic results for the Hubbard model.
To the best of our knowledge, 
no classical algorithm with performance guarantees exists in the regime we study, leaving open the possibility of exponential quantum advantage.
This is in contrast to the case of high temperature, where polynomial time algorithms are known \cite{bakshi2024hightemperature,Mann_2021}.
However, powerful heuristic methods, such as tensor networks, allow one to study the ground state of the Fermi-Hubbard model up to lattices of size $16\times 16$
\cite{liu2025accuratesimulationhubbardmodel}.
\\

\paragraph{Methods.} We follow the recent work \cite{chen2023efficient} which derives that the mixing time of Lindbladian QGSs can be upper bounded by giving lower bound bounds on the {\it spectral gap of the Lindbladian}. This spectral gap, in turn, is equal to the spectral gap of a corresponding {\it parent Hamiltonian}, which we analyse. Our approach is to first study free fermions in the general QGS framework of \cite{ding2024efficient} and make the particular design choice of {\it Majorana jump operators} and {\it Gaussian filter functions} in the corresponding algorithmic Lindbladian.\footnote{Our considerations for free fermions also work for any choice of filter functions.} In particular, this allows us to {\it exactly} compute the finite spectral gap of the free fermionic parent Hamiltonian (Proposition \ref{prop - spectrum of free fermionic Lindbladian}) via Prosen's third quantisation formalism \cite{Prosen_2008}. On top of this, recognising that the dynamics is restricted to that of Gaussian states, we calculate the covariance matrix of the evolved state, which together with optimal trace norm bounds \cite{bittel2025optimaltracedistanceboundsfreefermionic} allows us to prove a \textit{logarithmic upper bound} on the mixing time for free fermions -- hence proving \textit{rapid mixing} at any temperature for the first time (Proposition \ref{prop - rapid mixing of free fermions}). Second, we make use of Hasting's stability result \cite{hastings2017stabilityfreefermihamiltonians} on the spectral gap of free fermions under perturbation (see also \cite{De_Roeck_2018,koma2020stabilityspectralgaplattice}) in order to quantitatively extend the finite spectral gap in a system size-independent manner to the interacting parent Hamiltonian (Theorem \ref{thm: main, gap}). To lift the locality of interaction from the fermionic Hamiltonian to the Lindbladian's parent Hamiltonian we use Lieb-Robinson bounds and employ matrix analysis methods.\\

\paragraph{Extensions.} Our methods equally apply to the opposite regime $U/t\gg1$ of the Fermi-Hubbard model. Here, we start by exactly solving the {\it atomic limit} $t=0$ case of no hopping and the same choices of jump operators and filter functions (Proposition \ref{prop - spectrum of atomic Lindbladian}), after which we use adapted {\it eigenvalue perturbation techniques} to control finite $t>0$ (Theorem \ref{thm: gap stability atomic Lindbladian}). We again find (Corollary \ref{cor:main-result}) that at any finite temperature there exists a constant (system size-independent) maximal $t$ such that the corresponding Gibbs state can be algorithmically created with the complexities as stated in \eqref{eq:main-result}. We emphasize that the flexibility of our proof techniques naturally lend themselves to future explorations of other QMBS models in various regimes.
We note that provably efficient classical algorithms exist for quantum perturbations of classical Hamiltonian also at low temperatures \cite{Helmuth_2023}. As an application of Gibbs sampling, we adapt \cite[Theorem 8]{rouze2024optimalquantumalgorithmgibbs} to explain how to calculate partition functions for the considered systems, hence providing the means to resolve the physics in thermal equilibrium of the underlying QMBS in an end-to-end fashion.
\\

\paragraph{Numerical simulations.} We perform small-scale exact classical simulations for the weakly interacting spinless and spinful Fermi-Hubbard model in order to trial the hidden asymptotic constants in our analytical result. We find reasonable finite-gap behaviour and confirm in particular the predicted {\it system size-independent scaling}. In general, the {\it temperature dependence} is not favourable in our analytical result, and this also becomes visible in the numerical analysis. However, varying the choice of the jump operators from Majorana to Paulis and the filter functions from Gaussian to Metropolis, we observe a much improved temperature dependence in our simulations. To further push away this from the weakly interacting regimes, where perturbative methods can be applicable, we additionally test intermediate-strength couplings with $2 \lesssim U/t \lesssim 6$ for the spinless $D=1$ case at different temperatures. We see promising behaviour much beyond our analytical bounds, but for this intermediate regime the scaling in system size remains inconclusive from our small-scale numerics. We refer to Figure \ref{fig:Intro - results} for more details; the code for these simulations is available at \cite{Smid_GitHub_GibbsSampling}.\\

The remainder of this paper is structured as follows. We first give a general background on quantum Gibbs samplers in Sections \ref{sec:Quantum Gibbs sampling} -- \ref{sec:detailed-balance}. We then introduce fermionic systems in the third quantisation formalism (Section \ref{sec: third quantisation}) and discuss locality of the corresponding Lindbladian parent Hamiltonians (Section \ref{sec: locality}). Our main analytical results are derived in Section \ref{sec:fermions}, where we compute the relevant spectral gap and rapid mixing time for free fermions (Section \ref{sec:free-fermions}), analyse the stability of the gap under perturbations caused by exponentially decaying interactions (Section \ref{sec:stability}), discuss extensions to the atomic limit (Section \ref{sec:stability atomic}), and finally argue why this leads to algorithmically efficient QGS (Section \ref{sec:efficiency}). Section \ref{sec: partition function} explains how QBS provides a mean for calculating the partition function of the underlying QMBS. We give our numerical results in Section \ref{sec:simulations}, both for the analytically bounded regime (Section \ref{sec:analytically-bounded}) and much beyond (Section \ref{sec:beyond}). Last but not least, we speculate in Section \ref{sec:outlook} on the further use of the presented techniques. Some technical arguments are deferred to Appendices \ref{app:lemmas} -- \ref{appendix:detailed bounds}.

\begin{figure}[H]
    \centering
    \resizebox{0.9\textwidth}{!}{
        \subfloat[The gap $\Delta$ as a function of $U$]{
            \begin{tikzpicture}

\definecolor{zero}{RGB}{180, 180, 180}
\definecolor{one}{RGB}{206, 104, 104}
\definecolor{two}{RGB}{231, 195, 146}
\definecolor{three}{RGB}{231, 217, 146}
\definecolor{four}{RGB}{209, 223, 141}
\definecolor{five}{RGB}{116, 185, 116}
\definecolor{six}{RGB}{126, 147, 165}
\definecolor{seven}{RGB}{146, 136, 176}
\definecolor{eight}{RGB}{166, 124, 166}

\begin{axis}[
    height=.32\columnwidth,
    width=.47\columnwidth,
    xmin=-10.1, xmax=10.1,
    ymin=0, ymax=1.1,
    xlabel={$U$}, ylabel={$\Delta$},
    axis on top,
    legend cell align={left},
    legend style={
        at={(0.5,1.02)},
        anchor=south,
        cells={anchor=west},
        /tikz/every even column/.append style={column sep=.6em},
        /tikz/every odd column/.append style={column sep=.05em}},
    legend columns=4,
    ]
    \addlegendimage{empty legend}
    \addlegendimage{empty legend}
    \addlegendimage{empty legend}
    \addlegendimage{empty legend}
    \addlegendentry{}
    \addlegendentry{}
    \addlegendentry{\hspace{-3.5em}$\boldsymbol{n_\mathrm{sites}}$}
    \addlegendentry{}
    \addplot[mark=none, color=one, thick] table[x=U, y=Delta_nqb5_beta1] {data/usweep.dat};
    \addlegendentry{$5$}
    \addplot[mark=none, color=five, thick] table[x=U, y=Delta_nqb6_beta1] {data/usweep.dat};
    \addlegendentry{$6$}
    \addplot[mark=none, color=six, thick] table[x=U, y=Delta_nqb7_beta1] {data/usweep.dat};
    \addlegendentry{$7$}
    \addplot[mark=none, color=eight, thick] table[x=U, y=Delta_nqb8_beta1] {data/usweep.dat};
    \addlegendentry{$8$}
    
\end{axis}

\end{tikzpicture}
       }
        \subfloat[Slope of decay of $\Delta$ as $n$ grows]{
            \begin{tikzpicture}

\definecolor{zero}{RGB}{180, 180, 180}
\definecolor{one}{RGB}{206, 104, 104}
\definecolor{two}{RGB}{231, 195, 146}
\definecolor{three}{RGB}{231, 217, 146}
\definecolor{four}{RGB}{209, 223, 141}
\definecolor{five}{RGB}{116, 185, 116}
\definecolor{six}{RGB}{126, 147, 165}
\definecolor{seven}{RGB}{146, 136, 176}
\definecolor{eight}{RGB}{166, 124, 166}

\begin{axis}[
    height=.32\columnwidth,
    width=.47\columnwidth,
    xmin=1.8, xmax=11.2,
    ymin=0, ymax=1.3,
    xlabel={$n_\mathrm{qubits}$}, ylabel={$\tilde{d}$},
    axis on top,
    legend cell align={left},
    legend style={at={(0.5,1.02)}, anchor=south},
    legend columns=2,
    ]
    \addplot[mark=*, mark options={fill=white}, mark size=1.5, color=one, thick] table[x=n_qb, y=d_beta1_spinless_pos] {data/gap_initial_slope_1D_only.dat};
    \addlegendentry{$U > 0$, spinless}
    \addplot[mark=*, mark options={solid}, mark size=1.5, color=one, thick, densely dashed] table[x=n_qb, y=d_beta1_spinless_neg] {data/gap_initial_slope_1D_only.dat};
    \addlegendentry{$U < 0$, spinless}
    \addplot[mark=triangle*, mark options={fill=white}, mark size=2, color=five, thick] table[x=n_qb, y=d_beta1_spinful_pos] {data/gap_initial_slope_1D_only.dat};
    \addlegendentry{$U > 0$, spinful}
    \addplot[mark=triangle*, mark options={solid}, mark size=2, color=five, thick, densely dashed] table[x=n_qb, y=d_beta1_spinful_neg] {data/gap_initial_slope_1D_only.dat};
    \addlegendentry{$U < 0$, spinful}
\end{axis}

\end{tikzpicture}
        }}
\caption{\textbf{Numerical results for the Fermi-Hubbard model.} \textbf{(a)} Plotting the gap $\Delta$ of the full Lindbladian $\mathcal{L}^\dagger$ associated with the spinless $D=1$ Fermi-Hubbard thermal state with design choices beyond our analytical results\,---\,when using the Metropolis filter function and single site Pauli jump operators instead\,---\,as a function of the coupling strength $U$. Here we plot different system sizes separately, for the case $\beta = 1$ and $t = 1$, demonstrating a large spectral gap in the regime of intermediate coupling $2 \lesssim U/t \lesssim 6$, which also does not seem to close with growing inverse temperature $\beta$ (cf.~Fig.~\ref{fig:usweep}). We observe complete closing of the gap only when the interaction $U$ exceeds the support of the filter function. \textbf{(b)} Plotting the slope $\tilde d = \mp\left.\frac{\partial \Delta}{\partial U}\right|_{U=0^{\pm}}$ under which the spectral gap $\Delta$ of the full Lindbladian $\mathcal{L}^\dagger$ closes from that of the unperturbed Lindbladian $\mathcal{L}_0^\dagger$ in the analytically bounded regime, as the system size $n$ grows, at $\beta = 1$. As per our main result (Theorem \ref{thm: main, gap}) leading to the complexities as stated in Eq.~\eqref{eq:main-result}, this quantity has to be upper bounded uniformly in $n$. Figure~\ref{fig:spinless-fermi-hubbard-slope} shows data for more sets of parameters exhibiting different types of behaviours.}
    \label{fig:Intro - results}
\end{figure}
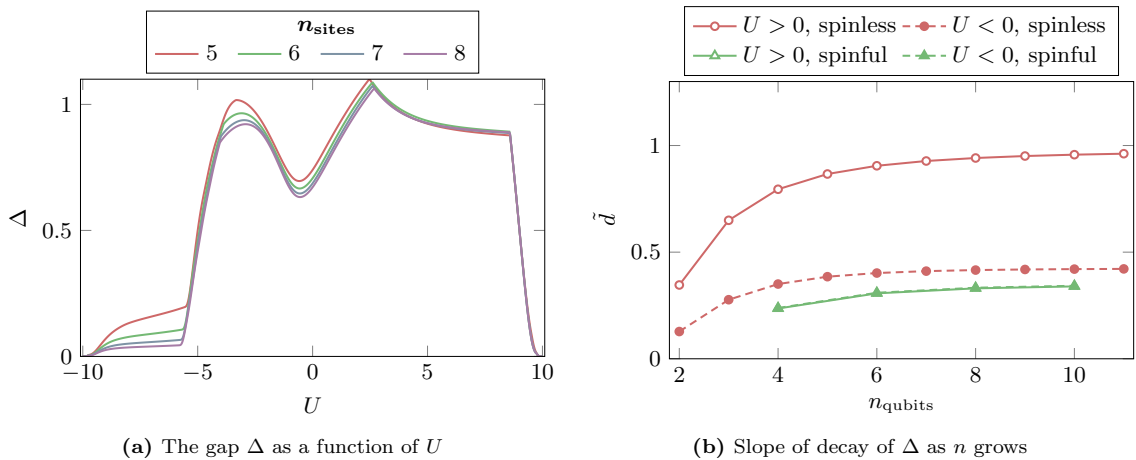


\section{Background}
\label{sec:background}

\subsection{Quantum Gibbs Sampling}
\label{sec:Quantum Gibbs sampling}

We start by reviewing latest literature on quantum Gibbs sampling and framework mostly developed in \cite{chen2023efficient,ding2024efficient}.
Quantum Gibbs sampling is the task of preparing the thermal state $\sigma_\beta = e^{-\beta H}/Z$ for a quantum Hamiltonian $H$.
This task is the quantum analog of classical Gibbs sampling and we briefly introduce that first.
Classical Gibbs sampling aims at
sampling from a classical distribution $\pi(i)=e^{-\beta E(i)}/Z$ where $E(i)$ is the energy of the configuration $i\in\Omega$ of the classical system.
As usual, we assume that $E(i)$ is known explicitly, but computing the partition function $Z$ is in general intractable.
The workhorse of classical Gibbs sampling is Markov chain Monte Carlo.
To sample from $\pi$ we construct the transition operator of a Markov chain $P_{ij}$ which gives the probability of transitioning from $i$ to $j$
such that $
\lim_{t\to\infty}(P^T)^t(\pi_0)=\pi$ for an arbitrary initial distribution $\pi_0$.
$\pi$ is called the stationary distribution of the Markov chain.
A sufficient condition for $\pi$ to be a stationary distribution of $P$ is the classical detailed balance condition, which is defined as $P$ being self-adjoint with respect to the inner product $\langle f,g\rangle_\pi = \sum_{i\in\Omega}\bar{f}_ig_i\pi_i$, where $f,g$ are functions on the configuration space.
Explicitly,
\begin{align}
\langle f,Pg\rangle_\pi   
=
\sum_{i,j\in\Omega}\bar{f}_iP_{ij}g_j\pi_i
=
\langle Pf,g\rangle_\pi 
=
\sum_{i,j\in\Omega}P_{ji}\bar{f}_ig_j\pi_j
\end{align}
which is true for all $f,g$ if 
\begin{align}
    P_{ij}\pi_i=P_{ji}\pi_j
\end{align}
for all $i,j\in\Omega$. By summing this equation over $i$ and using that $\sum_{i\in\Omega}P_{ji}=1$ we indeed have 
$\pi P = \pi$ so that the distribution $\pi$ is an eigenvector of $P$ with eigenvalue $1$.
If $P$ is also aperiodic and irreducible, then the Perron-Frobenius theorem guarantees that
$\pi$ is the unique stationary state of $P$ as desired.
Here, aperiodic means that the greatest common divisor of the number of transitions by which any $i\in\Omega$ can be reached starting from $i$ is $1$ and irreducible means that $P$ has no non-trivial invariant subspaces.\\

In the quantum case, we can proceed analogously. We first define the quantum Markov semigroup $\mathcal{P}_t$ as the semigroup of completely positive, unital maps. While in the classical case we used discrete time dynamics for simplicity, in the quantum case it is more useful to use continuous-time dynamics, so that we can work with the generator of the dynamics called the Lindbladian $\mathcal{L}$: $\mathcal{P}_t = e^{t\mathcal{L}}$.
To perform quantum Gibbs sampling, we construct a quantum Markov semigroup such that $\lim_{t\to\infty}
\mathcal{P}_t^\dagger(\rho_0) = \sigma_\beta$
where $\rho_0$ is an arbitrary initial state and $\Phi^\dagger$ for a superoperator $\Phi$ is the adjoint w.r.t.~the Hilbert-Schmidt inner product: $\langle A,B\rangle=\Tr(A^\dagger B)$.
There are many scalar products that reduce to the classical inner product $\langle \cdot, \cdot\rangle_\pi$ when $\sigma_\beta$ is a diagonal matrix with diagonal $\pi$.
The most useful one for our purposes is the Kubo-Martin-Schwinger (KMS) inner product.
Given a full rank state $\sigma>0$, this is defined for two operators $A,B$ as
\begin{align}
    \langle A,B\rangle_{\sigma}
    =
    \Tr(
    A^\dagger
    \Gamma_\sigma
    (B))
    \,,
\end{align}
where 
\begin{align}
\label{eq:Gamma}
\Gamma_\sigma(A) = 
\sigma^{1/2}
A
\sigma^{1/2}    
\,.
\end{align}

\begin{definition}(Quantum Detailed Balance)
    A Lindbladian $\mathcal{L}$ satisfies the KMS quantum detailed balance (QDB) condition if $\mathcal{L}$ is self-adjoint with respect to the KMS inner product.
\end{definition}

Since $\mathcal{L}[\id]=0$, we have that if $\mathcal{L}$ satisfies QDB,
\begin{align}
    0
    =
    \langle A, \mathcal{L}[\id]\rangle_\sigma
    =
    \langle \mathcal{L}[A], \id\rangle_\sigma
    =
    \langle \mathcal{L}[A], \sigma\rangle
    =
    \langle A, \mathcal{L}^\dagger[\sigma]\rangle
    \,,
\end{align}
for any operator $A$.
This shows that $\mathcal{L}^\dagger[\sigma]=0$
so that $\sigma$ is a stationary state of the dynamics generated by $\mathcal{L}^\dagger$. We can write the QDB condition more explicitly as:
\begin{equation}
    \mathcal{L}
    = 
    \Gamma_\sigma^{-1}
    \circ
    \mathcal{L}^\dagger 
    \circ
    \Gamma_\sigma \,.
\end{equation}
Note that in general $\mathcal{L}$ is a non-Hermitian operator, however the self-adjointness with the KMS inner product guarantees real spectrum.
We can define a Hermitian operator by a similarity transformation as follows:

\begin{definition}(Parent Hamiltonian)
Given a Lindbladian $\mathcal{L}$ satisfying QDB, we define the parent Hamiltonian 
of the state $\sqrt{\sigma}$, also known as the discriminant associated with the Lindbladian, as
\begin{align}
\label{eq:calH}
\mathcal{H} = 
\Gamma_{\sigma}^{-1/2}
\circ
\mathcal{L}^\dagger 
\circ
\Gamma_{\sigma}^{+1/2}
\,.
\end{align}
\end{definition}

Hermiticity of 
 $\mathcal{H}$ follows from the QDB condition:
\begin{align}
\mathcal{H}^\dagger =     
\Gamma_{\sigma}^{+1/2}
\circ
\mathcal{L}
\circ
\Gamma_{\sigma}^{-1/2}
=
\Gamma_{\sigma}^{-1/2}
\circ
\mathcal{L}^\dagger
\circ
\Gamma_{\sigma}^{+1/2}
=
\mathcal{H}
\,.
\end{align}
Also $\mathcal{H}$
has the same spectrum as $\mathcal{L}$ since they are related by a similarity transformation.
In particular 
\begin{align}
\mathcal{H}[\sqrt{\sigma}] =     
\Gamma_{\sigma}^{-1/2}
\circ
\mathcal{L}^\dagger 
[\sigma^{1/4}\sqrt{\sigma}\sigma^{1/4}]
=
0
\end{align}
which shows that indeed $\mathcal{H}$ is the parent Hamiltonian of $\sqrt{\sigma}$.
Here we slightly abuse the notion of parent Hamiltonian, since $\sqrt{\sigma}$ is the highest energy state of $\mathcal{H}$ rather than the lowest one.\\

The Lindbladian can be written in the following form
\begin{equation}
    \mathcal{L}^\dagger [\rho] = -i[G,\rho] + \sum_{a \in \mathcal{A}} \left( L_a \rho L_a ^\dagger - \frac{1}{2} \{ L_a ^\dagger L_a, \rho \} \right) \,.
\end{equation} 
$L_a$ are called the Lindblad operators
and $G=G^\dagger$ is the coherent term.
In terms of these operators the parent Hamiltonian is
\begin{align}
    \mathcal{H}[\rho]
    &=
    -i(\tilde{G}\rho - \rho \tilde{G}^\dagger)
    +
    \sum_{a\in\mathcal{A}}
    \left(
    \tilde{L}_a\rho\tilde{L}_a^\dagger 
    -
    \frac{1}{2}
    (\tilde{M}_a\rho + \rho \tilde{M}^\dagger_a)
    \right)
    \\
    \label{eq:tilde_ops}
    \tilde{G}&=\sigma_\beta^{-1/4}G
    \sigma_\beta^{+1/4}
    \,,\quad 
    \tilde{L}_a=\sigma_\beta^{-1/4}L_a
    \sigma_\beta^{+1/4}
    \,,\quad 
    \tilde{M}_a
    =
    \sigma_\beta^{-1/4}L_a^\dagger L_a
    \sigma_\beta^{+1/4}\,.
\end{align}
Similarly to the classical case, we define irreducibility of the quantum channel $\mathcal{E}$ as the absence of invariant subspaces. More precisely if $\mathcal{E}$ acts on the space of $d\times d$ matrices $\mathcal{M}_d$, $\mathcal{E}$ is irreducible if for a projector $P$, the identity $\mathcal{E}(P\mathcal{M}_d P)\subseteq P\mathcal{M}_dP$ occurs only for the trivial cases of $P=0,\id$ \cite[Thm 6.2]{wolf}.
Then if $e^{t\mathcal{L}^\dagger}$ is irreducible, the kernel of $\mathcal{L}^\dagger$ is one-dimensional and spanned by the full rank density matrix $\sigma_\beta$ \cite[Prop 7.5]{wolf}.
This guarantees uniqueness of the stationary state.
A useful criterion for irreducibility is that the algebra generated by $L_a$'s and $G$ is the whole operator algebra \cite[Cor 7.2]{wolf}.\\

The efficiency of the Lindbladian dynamics to prepare a thermal state is governed by the mixing time.
\begin{definition}
    The mixing time of the Lindbladian $\mathcal{L}^\dagger$ is \begin{equation}
        t_\textup{mix}(\epsilon) = \inf \left\{ t \geq 0 \left|\, \forall \rho: \left\|e^{t\mathcal{L}^\dagger}[\rho] - \sigma_\beta\right\|_{\Tr} \leq \epsilon \right.\right\}\,,
    \end{equation} where $\|A\|_{\Tr}
    =\Tr(\sqrt{A^\dagger A})
    $ denotes the trace norm.
\end{definition}

The Lindbladian time evolution implemented on a quantum computer will prepare the following state:
\begin{definition}
    The purified Gibbs state, also known as the thermofield double state, is
    \begin{equation}
        |\sqrt{\sigma_\beta}\rangle = \frac{1}{\sqrt{\Tr(e^{-\beta H})}} \sum_i e^{-\beta E_i /2} |E_i\rangle |\overline{E_i}\rangle\,,
    \end{equation} where we follow the vectorisation convention $|\psi\rangle\langle \phi| \to |\psi\rangle|\overline \phi\rangle$.
\end{definition}
This state is a vector in the doubled Hilbert space.
The Gibbs state $\sigma_\beta$ can be recovered by computing the reduced density matrix on one of the two copies of the Hilbert space.
However, access to the purified Gibbs state can be more useful, allowing, for example, more efficient estimation of observables \cite{Knill_2007}.


\subsection{Construction of Lindbladians with Quantum Detailed Balance and Their Properties}
\label{sec:detailed-balance}

Next, we review the construction of a Lindbladian that satisfies QDB for $\sigma=\sigma_\beta$.
We follow the construction of \cite{ding2024efficient}\,---\,the main difference from the construction of \cite{chen2023efficient} is that it allows one to use a finite number of Lindblad operators.
The construction is given in terms of a set of self-adjoint operators $\{A_a\}_{a\in \mathcal{A}}$ called jump operators and filter functions 
$\{\hat{f}^a(\nu)\}_{a\in \mathcal{A}}$
obeying
\begin{equation}
    \label{eq:q_fcn}
    \hat{f}^a(\nu) = q^a (\nu) e^{-\beta \nu /4},\quad q^a(-\nu) = \overline{q^a(\nu)} 
    \,.
\end{equation}
The Lindblad operators are then given by \begin{align}
    L_a &= \hat{f} ^a (\operatorname{ad}_H) A^a\\
    &= \sum\limits_{\nu \in B_H} \hat{f}^a (\nu) A^a _\nu\\
    \label{eq:jump_op_int_def}
    &= 
    \int_{-\infty} ^{\infty} f^a (t)  e^{iHt} A^a e^{-iHt}\ \dd t 
\end{align}
where $B_H=\{\nu=E_i-E_j\}$ is the set of Bohr frequencies, 
with $E_i,E_j$ running over the spectrum of $H$, and 
\begin{align}
    A_\nu = \sum_{i,j|E_i-E_j=\nu}
    P_iAP_j
    \,,\quad 
    A=\sum_{\nu\in B_H}A_\nu
    \,,\quad 
    A_\nu^\dagger = A_{-\nu}
    \,,
\end{align}
with $P_i$ the projector onto the eigenspace of eigenvalue $E_i$. Here, $\operatorname{ad}_H X = [H,X]$ represents the adjoint endomorphism. Note that $\operatorname{ad}_H A_\nu = [H,A_\nu]=\nu A_{\nu}$.
$\hat{f}^a(\nu)$
is the Fourier transform of $f^a(t)$.
The coherent term is given by \begin{align}
    G &= -i \tanh \circ \log (\Delta ^{1/4} _{\sigma_\beta}) \left( \frac{1}{2} \sum_{a \in \mathcal{A}} L_a ^\dagger L_a \right)\\
    &= \frac{i}{2} \sum\limits_{a \in \mathcal{A}} \sum\limits_{\nu \in B_H} \tanh\left( \frac{\beta \nu}{4} \right) (L_a ^\dagger L_a)_\nu\\
    &= \sum_{a \in \mathcal{A}} \int_{-\infty}^\infty g(t) e^{iHt} (L_a ^\dagger L_a) e^{-iHt}\ \dd t \label{eqn:Coherent term}
\end{align} 
with
\begin{equation}
    \hat{g} (\nu) = \frac{i}{2} \tanh\left( \frac{\beta \nu}{4} \right) \cdot \kappa(\nu)\,,
\end{equation} where $\Delta_\rho [X] = \rho X \rho^{-1}$ is the modular superoperator, and $\kappa(\nu)$ is a sort of smooth indicator function, obeying $\kappa(\nu) = 1$ on $\nu \in [-2\|H\|,2\|H\|]$, and decaying smoothly and rapidly afterwards, so that it belongs to the class of Gevrey functions as per \cite[Equation (3.17)]{ding2024efficient}.
In \cite{ding2024efficient}, it was proven that the Lindbladian so defined satisfies QDB with the thermal state $\sigma_\beta$.\\

The operators entering the parent Hamiltonian of equation \eqref{eq:tilde_ops} then become
\begin{align}
    \label{eq:tilde_L}
    \tilde{L}_a = e^{+\beta/4 H} \sum_{\nu \in B_H} \hat{f}^a(\nu) A_\nu^a e^{-\beta/4 H} = \sum_{\nu \in B_H} \hat{f}^a(\nu)e^{\beta/4\nu} A_\nu^a = \int_{-\infty}^{+\infty} f^a(t+i\beta/4) e^{iHt}A^ae^{-iHt}\ \dd t\,,
\end{align}
where we used $[H,A_\nu]=\nu A_\nu$, 
and similarly, 
\begin{align}
    \label{eq:tilde_G}
    \tilde{G}
    =
    e^{+\beta/4 H}
    \frac{i}{2} \sum\limits_{a \in \mathcal{A}} \sum\limits_{\nu \in B_H} \tanh\left( \frac{\beta \nu}{4} \right)
    (L_a ^\dagger L_a)_\nu
    e^{-\beta/4 H}
    =
    \sum_{a \in \mathcal{A}} \int_{-\infty}^\infty g(t+i\beta/4) e^{iHt} (L_a ^\dagger L_a) e^{-iHt}\ \dd t
    \,.
\end{align}
Note that $\tilde{L}_a = \tilde{L}_a^\dagger$ as $\overline{q^a(\nu)}=q^a(-\nu)$, $A_\nu^\dagger=A_{-\nu}$ and $B_{H}$ is symmetric under negation, $\tilde{G}^\dagger = \tilde{G}$ is as well. See also Lemma \ref{lemma:contour change} for the derivation of $\tilde{L}, \tilde{G}$ from a simple change of contour argument.
Finally, we note that the parent Hamiltonian is frustration-free, namely $\mathcal{H}=\sum_{a\in\mathcal{A}}\mathcal{H}_a$ where each $\mathcal{H}_a$ annihilates after vectorisation the purified Gibbs state. We defer to Section \ref{sec: locality} a discussion on locality of the parent Hamiltonian.\\

A popular filter function that we will mostly focus on below is the Gaussian one
\begin{align}
\label{eq:Gaussian filter}
    \hat f(\nu) = e^{-(\beta\nu+1)^2/8 + 1/8}\,,\quad 
    f(t) = \frac{1}{2\pi} \int_{-\infty}^{\infty}\hat{f}(\nu)e^{-i\nu t}\ \dd \nu
    =
    \sqrt{\frac{2}{\pi\beta^2}}
    \exp\left(-\tfrac{2}{\beta^2}
    \left(t-i\tfrac{\beta}{4}\right)^2\right)
    \,,
\end{align}
so that $f(t+i\beta/4)$ is positive. Another choice suggested by the authors of Ref.~\cite{ding2024efficient} is the Metropolis-type filter
\begin{equation} \label{eq:metropolis}
    \hat{f}^a(\nu) = \hat{f}(\nu) = q(\nu) \, e^{-\beta\nu / 4} =e^{-\sqrt{1 + \beta^2 \nu^2}} \, w(\nu / S) \, e^{-\beta\nu/4}\,,
\end{equation}
where $q(\nu)$ is supported on $[-S, S]$, and $w(x)$ is a ``bump function'' with support \emph{only} in the interval $x \in [-1, 1]$, for which we use
\begin{equation}
    w(x) =
    \begin{cases}
        e^{-\frac{1}{5(1-x^2)}} &\lvert x \rvert < 1\\
        0 & \lvert x \rvert \geq 1
    \end{cases}.
\end{equation}

Reference \cite[Theorem 34]{ding2024efficient} also proves that this Lindbladian evolution can be simulated on a quantum computer 
up to time $t$ 
with time complexity
\begin{equation}
\label{eq:runtime_quantum}
\tilde{\mathcal{O}}(t (\beta + 1) |\mathcal{A}|^2 \log^{1+s}(1/\epsilon))\,,
\end{equation} 
where now $\epsilon$ is the precision of the channel in the diamond norm, and $s \geq 1$ is the Gevrey order of the filter function $\hat{f}(\nu)$ (which is for example equal to $1$ for the Gaussian filter).
This assumes normalisation of the jump operators of the form $\max_{a \in \mathcal{A}} \|A^a\| \leq 1$, access to their block encodings, access to controlled Hamiltonian simulation, and preparation oracles for the filter function $f(t)$ (where $f^a(t) = f(t)$ is taken to be the same for all $a \in \mathcal{A}$) and coherent function $g(t)$.


\subsection{Fermionic Systems and Third Quantisation}
\label{sec: third quantisation}

We consider a set of fermionic creation and annihilation operators $a_i, a_i^\dagger$, $i\in \{1,\dots, n\}$. They generate the canonical anti-commutation relations algebra, which is defined by 
\begin{align}
    \{a_i, a_j\} = \{a_i^\dagger, a_j^\dagger\} = 0\,,\quad 
    \{a_i, a^\dagger_j\} = \delta_{i,j}\,,
\end{align}
where $\{a,b\}=ab+ba$ is the anticommutator.
Note that we will often use $1$ for the identity operator when the interpretation is obvious from the context.
The space of quantum states is the Fock space, which is spanned by $\ket{x}=(a_1^\dagger)^{x_1}\cdots (a_n^\dagger)^{x_n}\ket{0^n}$
where $x$ is a bit string of length $n$ and $\ket{0^n}$ is the vacuum, $a_i\ket{0^n}=0$ for all $i$.
An inner product is defined so that $\braket{x}{y}=\delta_{x,y}$ and $a^\dagger$ is indeed the adjoint of $a$.
We denote $N_S = \sum_{i\in S} a_i^\dagger a_i$ the number of fermions for a set $S$ of indices.
We also denote
$N_{\text{tot}}=N_{[n]}$, with $[j]=\{1,\dots,j\}$.
A fermionic operator $A$ is called even if $(-1)^{N_{\text{tot}}}A(-1)^{N_{\text{tot}}}=A$
and odd if 
$(-1)^{N_{\text{tot}}}A(-1)^{N_{\text{tot}}}=-A$.
Physical fermionic Hamiltonians are of the form
\begin{align}
    H = \sum_{I\subseteq [n]}
    h_I
\end{align}
where $h_I$ is an \textit{even} polynomial in $a_i, a_i^\dagger$ with $i\in I$, so that it is even.
The Fock space admits similarly an orthogonal decomposition in  even and odd orthogonal subspaces\,---\,a basis state $\ket{x}$ is even if $|x|$ is even, and odd if $|x|$ is odd, where $|x|$ is the Hamming weight of the bit string. 
Mathematically, the decomposition into even and odd sectors gives a $\mathbb{Z}_2$ grading. The tensor product $\ket{x}\otimes \ket{y}$ has grading (or parity) $|x|+|y|\mod 2$ and similarly for operators.
When we use the tensor product symbol for fermionic objects we implicitly assume this $\mathbb{Z}_2$ grading.

The Fock space has dimension $2^n$ and can be identified with the space of $n$ qubits.
Creation and annihilation operators can be represented on the space of qubits via the Jordan-Wigner transformation
\begin{align} \label{eq:jordan-wigner}
    N_i = \frac{1}{2}(\id + Z_i)
    \,,\quad 
    a_i = (-1)^{N_{[i-1]}} \sigma_i^{-}
    \,,\quad 
    a_i^\dagger = (-1)^{N_{[i-1]}} \sigma_i^{+}
\end{align}
where $\sigma^{\pm}
=
(X\pm i Y)/2$.
Other transformations exist, but we will not need them here.
It is also convenient to introduce another basis of the fermionic algebra given by the Majorana operators
\begin{align}
    \omega_{2j-1}=
    a_j+a_j^\dagger \,,\quad 
    \omega_{2j}=
    i(a_j-a_j^\dagger)\,,
\end{align}
for $j\in [n]$. Note that 
the $\omega_i$'s are self-adjoint and satisfy
$\{ \omega_i, \omega_j \}=2\delta_{i,j}$.\\

Next we discuss third quantisation, which is a formalism introduced in \cite{Prosen_2008} that allows one to efficiently solve Lindblad master equations for quadratic fermionic systems.
We start by introducing a Hilbert space structure $B \to |B\rangle$ to the space of operators by defining a canonical basis $\{P_\alpha\}_{\alpha \in \{0,1\}^{2n}}$ with $P_\alpha = \prod_{i=1}^{2n}\omega_i ^{\alpha_i}$. These basis vectors are orthonormal with respect to the inner product $\langle P|Q\rangle = \frac{1}{4^n} \Tr(P^\dagger Q)$. Now define $2n$ annihilation linear maps $c_j$ over this operator space by $c_j|P_\alpha\rangle = \delta_{\alpha_j,1} |\omega_j P_\alpha\rangle$. The action of their Hermitian adjoints, called creation linear maps, is readily found to be $c^\dagger _j |P_\alpha\rangle = \delta_{\alpha_j,0} |\omega_j P_\alpha\rangle$. These maps then obey the canonical anticommutation relations, $\{c_j,c_k\} = 0$ and $\{c_j,c^\dagger_k\} = \delta_{j,k}$, and so they act like canonical fermions. They will be referred to as adjoint Fermi maps, or a-fermions for short.

We can straightforwardly obtain the following relations: 
\begin{align} 
|P_\alpha \omega_j \rangle &= (-1)^{|\alpha|+\alpha_j} |\omega_j P_\alpha\rangle \\
    |\omega_j \omega_k P_\alpha \rangle - |P_\alpha \omega_j \omega_k\rangle &= 2 (c_j c^\dagger _k + c^\dagger _j c_k) |P_\alpha\rangle\\
    |\omega_j P_\alpha\rangle &= (c^\dagger _j + c_j)|P_\alpha\rangle \\
    (-1)^{\alpha_j}|\omega_j P_\alpha\rangle &= (c^\dagger _j - c_j)|P_\alpha\rangle \\
    (-1)^{|\alpha|} |P_\alpha\rangle &= \exp(i\pi N)|P_\alpha\rangle,
\end{align}
where $N = \sum_j c^\dagger _j c_j$ is the number operator. Also note that the Lindbladian, while not necessarily conserving the number of Majorana fermions, conserves their parity; and so we can restrict ourselves to the physical case of even numbers of Majorana fermions, hence recognizing that $\exp(i\pi N) = 1$ on this subspace. These properties hence allow us to rewrite the action of a quadratic fermionic Lindbladian like \begin{equation}
    \mathcal{L}^\dagger |_{+} [P_\alpha] \cong \mathcal{L}^\dagger |_{+} |P_\alpha\rangle
\end{equation} by expressing $\mathcal{L}^\dagger|_{+}$ as a quadratic form in a-fermions. Hence the spectrum of $\mathcal{L}^\dagger|_{+}$ can then be simply studied as that of a quadratic (not necessarily Hermitian) fermionic system.
We shall do this explicitly in Section \ref{sec:free-fermions}.

\begin{example}
    Consider a simple fermionic superoperator $\mathcal{L}[\rho] = a  \cdot \omega_1  \rho  \omega_2 + b \cdot  \omega_1 \omega_2 \rho  + c \cdot \rho \omega_1 \omega_2$. Associating the Hilbert space structure to this space, we can write $\mathcal{L}|\rho\rangle = a  |\omega_1  \rho  \omega_2\rangle + b | \omega_1 \omega_2 \rho\rangle  + c | \rho \omega_1 \omega_2\rangle$. Applying the second rule to the last term, we get $\mathcal{L}|\rho\rangle = a  |\omega_1  \rho  \omega_2\rangle + b | \omega_1 \omega_2 \rho\rangle  + c ( |\omega_1 \omega_2\rho\rangle - 2(c_1 c_2^\dagger + c_1 ^\dagger c_2) |\rho\rangle)$. Now applying the first rule to the first term, we obtain $\mathcal{L}|\rho\rangle = a  (-1)^{|\alpha| + \alpha_2}|\omega_1 \omega_2 \rho \rangle + (b+c) | \omega_1 \omega_2 \rho\rangle - 2c(c_1 c_2^\dagger + c_1 ^\dagger c_2) |\rho\rangle$. Restricting our view to physical states with $|\alpha|$ being even, and applying rules 3 and 4, we finally arrive at $\mathcal{L}|_+|\rho\rangle = a  (c_1^\dagger +c_1)(c_2^\dagger - c_2)|\rho \rangle + (b+c) (c_1^\dagger +c_1)(c_2^\dagger + c_2)|\rho\rangle - 2c(c_1 c_2^\dagger + c_1 ^\dagger c_2) |\rho\rangle$. Hence we see that \begin{align}
        \mathcal{L}|_+ &\cong a  (c_1^\dagger +c_1)(c_2^\dagger - c_2) + (b+c) (c_1^\dagger +c_1)(c_2^\dagger + c_2) - 2c(c_1 c_2^\dagger + c_1 ^\dagger c_2)\\
        &= a  (c_1^\dagger +c_1)(c_2^\dagger - c_2) + b (c_1^\dagger +c_1)(c_2^\dagger + c_2) + c(c_1^\dagger -c_1)(c_2^\dagger - c_2)\,.
    \end{align}
\end{example}

\begin{remark}
    The idea of Section \ref{sec:stability}, where we will prove stability of the gap of this Lindbladian under perturbation, is to view both the unperturbed part and the perturbation of the corresponding parent Hamiltonian in third quantisation. Since the unperturbed part will transform into a free fermionic Hamiltonian, we will be able to use gap stability results for free fermions \cite{hastings2017stabilityfreefermihamiltonians,De_Roeck_2018,koma2020stabilityspectralgaplattice} to show constant gap of the interacting Lindbladian.
\end{remark}


\pagebreak
\subsection{On Locality}
\label{sec: locality}

Reference \cite{ding2024efficient} shows that if $H$ is a geometrically local Hamiltonian, $A_a$ are local
and the filter function is Gaussian, then the Lindblad operators $L_a$ are quasi-local and $G$ is a sum of quasi-local terms.
Here we extend this result and discuss the locality properties of the parent Hamiltonian for fermionic systems and systems with exponentially decaying interactions.
The quasi-locality of the parent Hamiltonian will be an important ingredient in the proofs of gap stability we present below.\\

We consider a lattice $\Lambda$.
For qubit systems, an operator $O$ has support $I$ if it can be written as $O=\id_{\Lambda\setminus I}\otimes A$ for some operator $A$ acting on the space of the qubits at $I$.
This implies that operators that are 
spatially separated\,---\,i.e.~with disjoint supports\,---\,commute.
This definition is not useful for fermionic systems as odd fermionic operators anti-commute even if they are spatially separated.
We say that a fermionic operator $A$ has support $I$ if $A$ is a polynomial in $a_i, a_i^\dagger$ with $i\in I$.
We call a fermionic operator local if its support is a geometrically local region of the lattice, and
we call a fermionic operator quasi-local if it can be approximated by a local operator with an exponentially decaying error.
The following result is a generalisation of \cite[Prop.~20]{ding2024efficient} for Hamiltonians with exponentially decaying interactions.

\begin{prop}
\label{prop: locality}
Consider a Hamiltonian $H$ with interactions that decay at least exponentially, and local jump operators $A^a$ with Gaussian filter functions.
Then the parent Hamiltonian 
\eqref{eq:calH} is a sum of quasi-local terms.
\end{prop}

\begin{proof}
We define local approximations of $\tilde{L}, \tilde{G}$, and $L$ 
from \eqref{eq:tilde_L} and \eqref{eq:tilde_G}
by 
\begin{align}
    \tilde{L}_a^{(r)}
    &=
    \int_{-\infty} ^{\infty} f^a (t+i\beta/4)  e^{iH_{B_r(a)}t} A^a e^{-iH_{B_r(a)}t}\ \dd t 
    \,,\\
    \tilde{G}^{(r)}
    &=
    \sum_{a\in \mathcal{A}} \tilde G_a ^{(r)} =
    \sum_{a\in \mathcal{A}}
    \int_{-\infty} ^{\infty} g (t+i\beta/4)  e^{iH_{B_r(a)}t} (L^{(r)\dagger}_a L^{(r)}_a) e^{-iH_{B_r(a)}t}\ \dd t 
    \,,\\
    L_a^{(r)}
    &=
    \int_{-\infty} ^{\infty} f^a (t)  e^{iH_{B_r(a)}t} A^a e^{-iH_{B_r(a)}t}\ \dd t 
    \,,
\end{align}
where $B_r(a)$ is a ball
of radius $r$ around the support of $A^a$ and $H_\Omega=\sum_{I\,|\,I\cap\Omega\neq\emptyset}h_I$ is the truncated Hamiltonian to region $\Omega$.

Here we shall use a weaker version of the Lieb-Robinson bound than the one for local systems \cite[Lemma 5]{Haah_2021} used in \cite[Prop.~20]{ding2024efficient}, which also holds for exponentially decaying Hamiltonian interactions, and tells us that \begin{equation}\label{eqn:LR}
    \left\| e^{iHt} A^a e^{-iHt} - e^{iH_{B_r(a)}t} A^a  e^{-iH_{B_r(a)}t}\right\| \leq \|A^a\| \min\left\{ 2, J e^{-\mu r} (e^{\mu v |t|}-1)\right\}
\end{equation}
for some constants $J$, $v$, and $\mu$. From here, we shall assume $\|A^a\| \leq 1$. Using the Gaussian filter \eqref{eq:Gaussian filter}, it follows that \begin{align}
    \| L_a - L_a^{(r)}\| \leq \int_{-\infty}^\infty |f(t)|  J e^{-\mu r} (e^{\mu v |t|}-1)\ \dd t = C e^{-\mu r}\,,
\end{align} and similarly that \begin{align}
    \| \tilde L_a - \tilde L_a^{(r)}\| \leq \int_{-\infty}^\infty |f(t+i\beta/4)|  J e^{-\mu r} (e^{\mu v |t|}-1)\ \dd t = \tilde C e^{-\mu r}\,,
\end{align} as the integrals over $|f(t)| e^{c|t|}$ and $|f(t+i\beta/4)| e^{c|t|}$ converge.

\enlargethispage{0.2cm}
Regarding the coherent term, consider the function $\hat g(\nu) = \frac{i}{2} \tanh\left( \frac{\beta \nu}{4} \right)$ without the presence of the bump function. Its representation in the time domain is then $g(t) = \frac{1}{\beta} \frac{1}{\sinh(2\pi t /\beta)}$, which decays exponentially as $|t| \to \infty$ but has a singularity at $t=0$. But the coherent term in the parent Hamiltonian then depends on $g(t+i\beta/4) = -\frac{i}{\beta} \frac{1}{\cosh(2\pi t /\beta)}$, which is no longer singular at $t= 0$. Hence we can observe that \begin{align}
    \|\tilde G_a - \tilde G_a ^{(r)}\| &\leq \int_{-\infty}^\infty |g(t+i\beta/4)| \left\| e^{iHt} L_a ^\dagger L_a e^{-iHt} - e^{iH_{B_r(a)}t} L_a^{(r)\dagger} L_a^{(r)}  e^{-iH_{B_r(a)}t}  \right\| \dd t \\
    &\leq 2 \int_{-\infty}^\infty |g(t+i\beta/4)| \left\| e^{iHt} L_a e^{-iHt} - e^{iH_{B_r(a)}t} L_a^{(r)}  e^{-iH_{B_r(a)}t}  \right\| \dd t \\
    &\leq 2 \int_{-\infty}^\infty \int_{-\infty}^\infty |g(t+i\beta/4)| \cdot |f(s)| \left\| e^{iH(t+s)} A^a e^{-iH(t+s)} - e^{iH_{B_r(a)}(t+s)} A^a  e^{-iH_{B_r(a)}(t+s)}  \right\| \dd s \dd t \\
    &\leq 2 \int_{-\infty}^\infty \int_{-\infty}^\infty |g(t+i\beta/4)| \cdot |f(s)| \min\left\{ 2, J e^{-\mu r} (e^{\mu v |t+s|}-1) \right\} \dd s \dd t \\
    &\leq 2 \int_{-\infty}^\infty |g(t+i\beta/4)| \min\left\{ 2, \tilde J e^{-\mu r} (c e^{\mu v |t|}-1) \right\} \dd t\,,
\end{align} where the last inequality follows from splitting $|t+s| \leq |t| + |s|$ and carrying out the integral over $s$, which converges since $f(s)$ is Gaussian. Now this minimum changes at $|t| = t^* = \frac{1}{\mu v} \log \left( (2 e^{\mu r}/ \tilde J  + 1)/c\right)$, which we can lower bound by $t^* \geq \frac{1}{\mu v} \log(2/(c\tilde J)) + r/v = \tilde c + r/v$. Hence we can continue the upper bound like \begin{align}
    \|\tilde G_a - \tilde G_a ^{(r)}\| &\leq 2 \int_{|t| > \tilde c + r/v} |g(t+i\beta/4)| \cdot 2\ \dd t + 2 \int_{|t| < \tilde c + r/v} |g(t+i\beta/4)| \cdot \tilde J e^{-\mu r} (c e^{\mu v |t|}-1)\ \dd t\\
    &= \frac{8}{\beta} \int_{\tilde c + r/v}^\infty \frac{1}{\cosh(2\pi t/\beta)}\ \dd t + \frac{4\tilde J}{\beta} e^{-\mu r} \int_0^{\tilde c + r/v}  \frac{1}{\cosh(2\pi t/\beta)} (c e^{\mu v t} -1)\ \dd t\\
    &\leq \frac{16}{\beta} \int_{\tilde c + r/v}^\infty e^{-2\pi t/\beta}\ \dd t + \frac{8\tilde J}{\beta} e^{-\mu r} \int_0^{\tilde c + r/v}  e^{-2\pi t/\beta} (c e^{\mu v t} -1)\ \dd t\\
    &= \frac{8}{\pi} e^{-2\pi(\tilde c + r/v)/\beta} + \frac{4\tilde J}{\pi} e^{-\mu r} \left(e^{-2\pi(\tilde c + r/v)/\beta} -1\right) - \frac{8 c \tilde J}{2\pi - \beta\mu v} \left(e^{-2\pi r /(v\beta) + \tilde c (\mu v - 2\pi /\beta)}-e^{-\mu r}\right)\,,
\end{align} which is indeed exponentially decaying in $r$, proving that $\tilde G_a$ is quasi-local, and hence that $\tilde G$ is a sum of quasi-local terms. Altogether this shows that the parent Hamiltonian $\mathcal H$ is a sum of quasi-local terms.
\end{proof}

One can observe that, since the transformation to the parent Hamiltonian improves the decay of the filter functions and gets rid of the singularity of the coherent function $g(t)$ appearing in the defining integrals, the parent Hamiltonian is actually a better behaved and a more natural object than the Lindbladian itself; even though it contains terms of the form $\sigma_\beta ^{-1/4} O \sigma_\beta ^{1/4}$, which are generally non-local and can grow to infinite size even for a finite $\beta$ \cite{Avdoshkin_2020}, as we do not have Lieb-Robinson bounds for imaginary/Euclidean time evolution.\\

Note that the Lieb-Robinson bound, which follows from the bound on the commutator with the Hamiltonian terms, is true in our fermionic setting independently of whether 
$A_a$  is even or odd in the number of fermions, since the constituent Hamiltonian terms are always even, so that \eqref{eqn:LR} still holds.
We refer to \cite{nachtergaele2018lieb} for more on Lieb-Robinson bounds and locality for fermions.\\

We conclude this section with a remark on the runtime of the quantum algorithm that simulates the Lindbladian dynamics. Note that if we take local jump operators we have $|\mathcal{A}| = \Omega(n)$. This is due to
quasi-locality of $L_a$'s and $G$, and the irreducibility criterion for uniqueness of the stationary state $\sigma_\beta$ discussed in section \ref{sec:Quantum Gibbs sampling} that requires the $L_a$'s to span the whole operator algebra.
This implies that the runtime of Eq.~\eqref{eq:runtime_quantum} is lower bounded by $\Omega(n^2)$ even before considering the mixing time.


\subsection{Fermi-Hubbard Model}

As mentioned in the overview (Section~\ref{sec:overview}), this work is concerned with the applicability of a particular QGS to fermionic systems, and specifically to the Fermi-Hubbard model. In its original form, it consists of fermions on a $D$-dimensional lattice and is governed by the Hamiltonian
\begin{equation}
    \label{eq:Hamiltonian FH}
    H_\mathrm{FH} = -t \sum_{\langle i, j \rangle, \sigma} \left(a_{i,\sigma}^\dagger a_{j,\sigma} + a_{j,\sigma}^\dagger a_{i,\sigma}\right) + U \sum_i a_{i,\uparrow}^\dagger a_{i,\uparrow} a_{i,\downarrow}^\dagger a_{i,\downarrow},
\end{equation}
where $\langle \cdot, \cdot \rangle$ means neighbouring sites on the lattice, $\sigma \in \{\uparrow, \downarrow\}$, and $a_{i,\sigma}^{(\dagger)}$ are the usual fermionic annihilation (creation) operators on site $i$ with spin $\sigma$. The model parameters $t$ and $U$ are usually positive, though we will also consider $U < 0$ in some instances (the attractive Fermi-Hubbard model).

There is also a \emph{spinless} (sometimes called \emph{polarised}) version of this model, which removes the spin from the particles and replaces the on-site interaction with that of nearest neighbours. Its Hamiltonian is therefore
\begin{equation}
    \label{eq:Hamiltonian pFH}
    H_\mathrm{pFH} = -t \sum_{\langle i, j \rangle} \left(a_{i}^\dagger a_{j} + a_{j}^\dagger a_{i}\right) + U \sum_{\langle i, j \rangle} a_{i}^\dagger a_{i} a_{j}^\dagger a_{j}.
\end{equation}
Being less computationally demanding on classical hardware (for the same number of sites) but still exhibiting interesting behaviour, the spinless Fermi-Hubbard model is a good candidate for numerical finite-size study of the QGS considered in this work, and in Section~\ref{sec:simulations} we present results for both the spinful and spinless Fermi-Hubbard models.


\section{Analytical Results on Interacting fermions}\label{sec:fermions}

This chapter will provide the bulk of the proof for gapness of the Lindbladian $\mathcal L ^\dagger$ corresponding to weakly interacting fermionic systems, and hence for the efficiency of the quantum Gibbs state preparation. In Section \ref{sec:free-fermions}, we will explicitly calculate the gap of the Lindbladian for free fermions and express it using the third quantisation as a quadratic fermionic system. Section \ref{sec:stability} will then bound the perturbation of the Lindbladian for the interacting fermionic case, and explain how we can use the stability of free fermions to lower bound the gap. 
In Section \ref{sec:stability atomic}, we discuss the atomic limit where inter site interactions are set to zero and show that the Lindbladian gap persists also for perturbations around this limit.
Finally, Section \ref{sec:efficiency} discusses how these results on the gap translate to mixing time and algorithmic complexity of Gibbs state preparation.


\subsection{Spectrum of the Lindbladian for Free Fermions}\label{sec:free-fermions}

\begin{lemma}\label{lemma: G vanishing}
    For a free fermionic system, given by $H_0 = \boldsymbol{\omega}^T \cdot h \cdot \boldsymbol{\omega} = \sum_{i,j} \omega_i h_{ij} \omega_j$ with $h$ Hermitian and anti-symmetric, by taking the set of jump operators to be $\mathbf{A} = M \cdot \boldsymbol{\omega}$, where $M$ is a unitary matrix, and the filter functions $\hat f^a$ to be real and equal, the coherent term $G$ vanishes.
\end{lemma}

\begin{proof}
    Note that on the space $\mathcal{S} = \operatorname{span}\{\omega_a\}$, we have that $\operatorname{ad}_{H_0} =_{\mathcal S} -4h$. Hence the time-evolved jump operators are given by \begin{equation}
        \mathbf{A}(t) = e^{iH_0 t} \mathbf{A} e^{-iH_0 t} = e^{iH_0 t} M \cdot \boldsymbol{\omega} e^{-iH_0 t} = M \cdot e^{-4iht} \cdot  \boldsymbol{\omega}\,,
    \end{equation} and the Lindblad operators are then just \begin{equation}
        \mathbf{L} = M \cdot \int_{-\infty}^\infty f(t) e^{-4iht}\ \dd t \cdot \boldsymbol{\omega} = M \cdot \hat{f} (-4h) \cdot \boldsymbol{\omega}\,.
    \end{equation} Hence we get that \begin{align}
        \sum_{a\in \mathcal{A}} L^\dagger _a L_a = \mathbf{L}^\dagger \cdot \mathbf{L} 
        &= \boldsymbol{\omega}^T \cdot \hat{f} (-4h) \cdot M^\dagger \cdot M \cdot \hat{f} (-4h) \cdot \boldsymbol{\omega} \\
        &= \boldsymbol{\omega}^T \cdot [\hat{f} (-4h)]^2 \cdot \boldsymbol{\omega} \\
        &= \boldsymbol{\omega}^T \cdot \left( \frac{[\hat{f} (-4h)]^2 - [\hat{f} (4h)]^2}{2} + \operatorname{diag}\left([\hat{f} (-4h)]^2 \right)\right) \cdot \boldsymbol{\omega} \\
        &= \boldsymbol{\omega}^T \cdot q(4h)^2 \cdot \sinh(2\beta h) \cdot \boldsymbol{\omega} + \Tr\left( \hat{f}(-4h)^2 \right)\,,
    \end{align} where we have split up the matrix $[\hat{f} (-4h)]^2$ into its anti-symmetric, diagonal, and a symmetric hollow part.
    Now observe that \begin{equation}
        [(H_0)_n,\boldsymbol{\omega}^T \cdot A \cdot \boldsymbol{\omega}] = \boldsymbol{\omega}^T \cdot 4^n [(h)_n,A] \cdot \boldsymbol{\omega} 
    \end{equation} for any anti-symmetric matrix $A$, where $[(X)_n,Y]$ denotes the $n$-th iterated commutator. Hence by using the Campbell identity, we obtain \begin{align}
        \sum_{a \in \mathcal{A}} e^{iH_0 t} L^\dagger _a L_a e^{-iH_0 t} &= \boldsymbol{\omega}^T \cdot e^{4ith} \cdot q(4h)^2 \cdot \sinh(2\beta h) \cdot e^{-4ith} \cdot \boldsymbol{\omega} + \Tr\left( \hat{f}(-4h)^2 \right)\\
        &= \boldsymbol{\omega}^T \cdot q(4h)^2 \cdot \sinh(2\beta h) \cdot \boldsymbol{\omega} + \Tr\left( \hat{f}(-4h)^2 \right)\,,
    \end{align} independent of $t$, proving that $\sum_{a\in \mathcal{A}} L^\dagger _a L_a$ is an integral of motion under $H_0$, which means that \begin{equation}G = \int_{-\infty}^\infty g(t) \cdot \sum_{a \in \mathcal{A}} e^{iH_0 t} L^\dagger _a L_a e^{-iH_0 t}\ \dd t = \int_{-\infty}^\infty g(t) \cdot \sum_{a \in \mathcal{A}} L^\dagger _a L_a\ \dd t \propto \hat{g}(0) = 0\,,\end{equation}
    meaning that the coherent term vanishes.
\end{proof}

\begin{prop}\label{prop - spectrum of free fermionic Lindbladian}
    The Lindbladian $\mathcal{L}^\dagger_0$ corresponding to the free fermionic Hamiltonian $H_0$ with the set of jump operators $\{\omega_a\}_{a=1}^{2n}$ and equal real filter functions $\hat f^a(\nu) = \hat f(\nu) = q (\nu) e^{-\beta \nu /4}$ has spectral gap\footnote{Here by spectral gap we mean the gap between the highest and second highest eigenvalue of the Lindbladian, i.e. the one that bounds the mixing time; though this will turn out to be the same gap as between the lowest and second lowest eigenvalue.} given by \begin{equation}
        \Delta_0 = 2 \cdot \min_i q(4\epsilon_i) ^2 \cosh(2\beta \epsilon_i),
    \end{equation} where $\epsilon_i \in \operatorname{spec}(h)$ are the eigenvalues of the single particle Hamiltonian $h$.
\end{prop}

\begin{proof}
    For future convenience, let's consider the similarity transformation $\mathcal{H}_0[\rho] = \sigma_{\beta} ^{-1/4}\cdot \mathcal{L}_0^\dagger [\sigma_{\beta} ^{1/4} \cdot\rho\cdot \sigma_{\beta} ^{1/4}]\cdot \sigma_{\beta} ^{-1/4}$ into the parent Hamiltonian, which is Hermitian due to the QDB condition, i.e. self-adjoint w.r.t.~the Hilbert-Schmidt inner product. As this is a similarity transformation, the spectrum of this superoperator will be the same as of $\mathcal{L} ^\dagger _0$. We can calculate that \begin{equation}
        \sigma_\beta ^{-1/4} L_a \sigma_\beta ^{1/4} = \hat{f}(-4h)_a \cdot e^{-\beta h} \cdot \boldsymbol{\omega},
    \end{equation}
    and the QDB condition also ensures $\sigma_\beta ^{-1/4} L_a \sigma_\beta ^{1/4} = \sigma_\beta ^{1/4} L^\dagger _a \sigma_\beta ^{-1/4}$. Using the calculation from Lemma \ref{lemma: G vanishing}, we can also straightforwardly evaluate \begin{align}
        \sum\limits_{a \in \mathcal A} \sigma_\beta ^{-1/4} L^\dagger _a L_a \sigma_\beta ^{1/4} &= \boldsymbol{\omega}^T \cdot q(4h)^2 \cdot \sinh(2\beta h) \cdot \boldsymbol{\omega} + \Tr\left( \hat{f}(-4h)^2 \right) \\
        &=  \sum\limits_{a \in \mathcal A} \sigma_\beta ^{1/4} L^\dagger _a L_a \sigma_\beta ^{-1/4}\,,
    \end{align} and so the parent Hamiltonian simplifies to \begin{align}
        \mathcal{H}_0[\rho] = &\sum\limits_{a \in \mathcal A}  \boldsymbol{\omega}^T \cdot q(4h)^\dagger _a \cdot \rho \cdot q(4h)_a \cdot \boldsymbol{\omega}- \frac{1}{2} \boldsymbol{\omega}^T \cdot q(4h)^2 \cdot \sinh(2\beta h)\cdot \boldsymbol{\omega}  \cdot  \rho - \rho \cdot \boldsymbol{\omega}^T \cdot \frac{1}{2} q(4h)^2 \cdot \sinh(2\beta h)\cdot \boldsymbol{\omega}\\ &- \Tr\left(q(4h)^2\cdot \cosh(2\beta h)\right) \cdot \rho \,.
    \end{align}

    Now following Prosen's third quantisation \cite{Prosen_2008}, which we reviewed in Section \ref{sec: third quantisation}, we obtain the equivalent form \begin{align}
        \mathcal{H}_0 &\cong  - \mathbf{c}^\dagger \cdot S \cdot \mathbf{c} + \mathbf{c} \cdot S \cdot \mathbf{c}^\dagger + \mathbf{c}^\dagger \cdot A \cdot \mathbf{c}^\dagger + \mathbf{c} \cdot A \cdot \mathbf{c}-\Tr\left(\sqrt{S^2 + A^2}\right)\,,
    \end{align} where we have restricted the Hilbert space to that of physical states with even numbers of Majorana fermions; and $S = q(4h)^2$, $A = q(4h)^2 \sinh(2\beta h)$, and $\{c^\dagger _i,c _i\}_{i = 1}^{2n}$ is a set of $2n$ canonical fermionic creation and annihilation operators. This is just a quadratic a-fermionic system with dynamical matrix $D = \left(\begin{matrix}
        -S & A \\ A & S
    \end{matrix}\right)$. Since both $S$ and $A$ are just functions of $h$, they are simultaneously diagonalisable with the eigenbasis of $h$, and hence $D$ is also easily diagonalisable, with eigenvalues \begin{equation}\lambda^{\pm}_i = \pm q(4\epsilon_i)^2 \cosh(2\beta \epsilon_i)\,,\end{equation} where $\epsilon_i \in \operatorname{spec}(h)$. Finally, the complete spectrum of $\mathcal{H}_0$, which is the same as that of $\mathcal{L}^\dagger_0$, is then \begin{equation}
        \operatorname{spec}(\mathcal{L}^\dagger_0) = \left\{ \sum_{i=1}^{2n}(-1 + (-1)^{x_i}) \cdot  q(4\epsilon_i)^2 \cosh(2\beta \epsilon_i) \right\}_{x \in \{0,1\}^{2n}}\,,
    \end{equation}
    and the corresponding spectral gap is \begin{equation}
        \Delta_0 = 2 \cdot \min_i q(4\epsilon_i)^2 \cosh(2\beta \epsilon_i)\,.
    \end{equation}
    This argument also assures that the Gibbs state is the unique fixed point of the dynamics generated by $\mathcal{L} ^\dagger _0$.
\end{proof}

\begin{prop}\label{prop - rapid mixing of free fermions}
     For free fermionic Hamiltonians, which have a bounded single particle Hamiltonian\,---\,meaning $\| h\| \leq \mathcal{O}(1)$\,---\, when taking the initial state to be specifically the maximally mixed state $\rho = \frac{I}{2^n}$, the Lindbladian $\mathcal{L}^\dagger _0$ mixes rapidly, i.e. in logarithmic time, with an upper bound \begin{equation}
         t_\textup{mix} \leq \frac{1}{2\Delta_0} \log\left( \frac{\tanh(2\beta\|h\|)}{2} \cdot \frac{n}{\epsilon}\right) = \frac{1}{4 \min_i q(4\epsilon_i)^2 \cosh(2\beta \epsilon_i)} \log\left( \frac{\tanh(2\beta\|h\|)}{2} \cdot \frac{n}{\epsilon}\right)\,.
     \end{equation}
\end{prop}

\begin{proof}
    First, we need to recognise that when we start with a Gaussian state $\rho = \frac{I}{2^n}$ and evolve it with a quadratic Lindbladian, we will stay within the subspace of Gaussian states. These can be uniquely characterised by their covariance matrices $\Gamma_{ij}=\frac{i}{2}\Tr([\omega_i,\omega_j]\rho)$. Denoting $\Gamma(t)$ the covariance matrix of $\rho(t) = e^{t\mathcal{L}^\dagger}[\frac{I}{2^n}]$, we can follow \cite{Barthel_2022} to obtain its equation of motion generated by our Lindbladian as \begin{equation}
        \frac{\dd}{\dd t}\Gamma(t) = -2q(4h)^2\cosh(2\beta h) \cdot \Gamma(t) - \Gamma(t) \cdot 2q(4h)^2\cosh(2\beta h) + 2iq(4h)^2\sinh(2\beta h)\,.
    \end{equation} Note that the initial covariance matrix is simply $\Gamma(0) = 0$, and since this commutes with $h$, and all the terms of the equation also commute with $h$, we can expect that $\Gamma(t)$ commutes with $h$ for any $t$, and hence we can straightforwardly solve this equation with \begin{equation}
        \Gamma(t) = \frac{i}{2}\tanh(2\beta h) \cdot \left( 1 - e^{-4q(4h)^2\cosh(2\beta h)t}\right)\,.
    \end{equation} We can also check that the covariance matrix of the Gibbs state $\sigma_\beta$ is $\frac{i}{2}\tanh(2\beta h ) = \Gamma(\infty)$, and so the evolution indeed converges to the Gibbs state.

    Finally, we can use optimal trace norm bounds obtained in \cite{bittel2025optimaltracedistanceboundsfreefermionic}, which tell us that \begin{align}
        \left\|e^{t\mathcal{L}^\dagger}\left[\frac{I}{2^n}\right] - \sigma_\beta\right\|_{\Tr} &\leq \frac{1}{2}\|\Gamma(t)-\Gamma_{\sigma_\beta}\|_{\Tr}\\
        &= \frac{1}{2}\left\| \frac{i}{2}\tanh(2\beta h) \cdot  e^{-4q(4h)^2\cosh(2\beta h)t}\right\|_{\Tr}\\
        &= \sum_{j|\epsilon_j \in \operatorname{spec}(h) } \frac{1}{4}|\tanh(2\beta\epsilon_j)| \cdot  e^{-4q(4\epsilon_j)^2\cosh(2\beta \epsilon_j)t}\\
        &\leq \frac{n}{2} \tanh(2\beta\|h\|)\cdot  e^{-4\min_j q(4 \epsilon_j)^2\cosh(2\beta \epsilon_j)t}\\
        &\overset{\text{set }}{\leq \epsilon}\,.
    \end{align}
    This final inequality can be then solved for $t$ like $
        t \geq \frac{1}{4 \min_i q(4\epsilon_i)^2 \cosh(2\beta \epsilon_i)} \log\left( \frac{\tanh(2\beta\|h\|)}{2} \cdot \frac{n}{\epsilon}\right)\,,
   $ and hence we deduce that \begin{equation}
        t_\textup{mix} \leq \frac{1}{4 \min_i q(4\epsilon_i)^2 \cosh(2\beta \epsilon_i)} \log\left( \frac{\tanh(2\beta\|h\|)}{2} \cdot \frac{n}{\epsilon}\right)\,.
    \end{equation}
\end{proof}

\begin{corollary}
    For free fermionic Hamiltonians, which have a bounded single particle Hamiltonian\,---\,meaning $\| h\| \leq \mathcal{O}(1)$\,---\, the Lindbladian $\mathcal{L}^\dagger _0$ has a constant spectral gap $\Delta_0$ and is efficiently simulable.
\end{corollary}

\begin{proof}
    Using for example the Gaussian filter function $\hat f(\nu) = e^{-(\beta\nu+1)^2 /8 + 1/8}$, which is efficiently implementable, the gap simplifies to \begin{equation}\label{eqn: free gap, Gaussian}
        \Delta_0 = 2 \cdot e^{-4\beta^2 \|h\|^2} \cosh(2\beta \|h\|),
    \end{equation} a monotonically decreasing function w.r.t.~$\|h\|$, which is hence bounded below when  $\| h\| \leq \mathcal{O}(1)$. Such a condition is assured when considering free fermionic Hamiltonians with hopping rates decaying at least polynomially, as then $\|h\|_\infty \leq \mathcal{O}(1)$. Here the induced infinity norm means the maximal absolute row sum of the matrix. The mixing time is then $t_\textup{mix} = \mathcal{O}(\log(n/\epsilon))$ and the total time complexity of the algorithm will be $\tilde{\mathcal{O}}(n^2 e^{4 \beta^2 \|h\|^2} \operatorname{polylog}(1/\epsilon))$, where $\epsilon$ is the required precision from the Gibbs state in the trace norm. The details about complexity will be discussed later in Section \ref{sec:efficiency}.
\end{proof}


\subsection{Stability of the Free Fermionic Gap Under Perturbations}\label{sec:stability}

In this section, we shall consider the Lindbladian $\mathcal{L}^\dagger _0$ corresponding to a quasi-local free fermionic Hamiltonian $H_0 = \sum_{i,j} \omega_i h_{ij} \omega_j$, and the Lindbladian $\mathcal{L}^\dagger$ corresponding to the perturbed quasi-local fermionic Hamiltonian $H = H_0 + \lambda V$. We will denote their (Hermitian) parent Hamiltonians, obtained via similarity transformations, by $\mathcal{H}_0$ and $\mathcal{H}$ respectively; and the perturbation of the parent Hamiltonians by $\mathcal{V} = \mathcal{H} - \mathcal{H}_0 $. We shall prove that $\mathcal{L}^\dagger$ remains gapped for perturbations with strength $|\lambda| \leq \lambda_\text{max}$ for some constant $\lambda_\text{max}$ by using theorems about stability of the gap for lattice fermions proved in \cite{hastings2017stabilityfreefermihamiltonians} and refined in \cite{De_Roeck_2018,koma2020stabilityspectralgaplattice}.
   
\begin{definition}[Definition 1 of \cite{hastings2017stabilityfreefermihamiltonians}]\label{def:1 Hastings}
    An operator $W$ is said to have $(K,\mu)$-decay if it can be decomposed as \begin{equation}
        W = \sum_{r \geq 1} \sum_{C \in \mathcal{C}(r)} W_C\,,
    \end{equation} where $\mathcal{C}(r)$ denotes the set of cubes with side length $r$, and $W_C$ are operators supported only on cubes $C$ such that \begin{equation}
        \max_{C \in \mathcal{C}(r)} \|W_C\| \leq K e^{-\mu r}
    \end{equation} with positive constants $K$ and $\mu$.
    \label{def: Hastings locality}
\end{definition}

\begin{definition}[Definition 2 of \cite{hastings2017stabilityfreefermihamiltonians}]
\label{def: Hastings locality for free fermions}
    An operator $B = \sum_{i,j \in \Lambda} \omega_i B_{ij} \omega_j$ is said to have a $[J,\nu]$-decay if \begin{equation}
        |B_{ij}| \leq J e^{-\nu \operatorname{dist}(i,j)}
    \end{equation} with positive constants $J$ and $\nu$, where $\operatorname{dist}(i,j)$ is the distance on $\Lambda$ in Manhattan metric.
\end{definition}

\begin{theorem}[Corollary 1 of \cite{hastings2017stabilityfreefermihamiltonians}]\label{thm: Hastings stability}
    If the Hermitian operator $\mathcal{H}_\textup{free} = \sum_{i,j \in \Lambda} \omega_i \mathfrak{H}_{ij} \omega_j$ has $[J,\nu]$-decay and a spectral gap $\Delta_0$, and the Hermitian operator $\mathcal{H}_\textup{int}$ has $(K,\mu)$-decay, then there exist positive constants $K_\textup{max}$ and $s$ independent of the system size, such that whenever $K \leq K_\textup{max}$, the gap of $\mathcal{H}_\textup{free} + \mathcal{H}_\textup{int}$ is lower bounded by $\Delta_0 - s K$.
\end{theorem}

\begin{lemma}\label{lemma: decay of free part}
    Assume that $H_0$ has $[J_0,\nu_0]$-decay, and that we are using the Gaussian filter function with Majorana jump operators. Then the parent Hamiltonian $\mathcal{H} _0$  of the Lindbladian $\mathcal{L}^\dagger _0$ corresponding to the free fermionic system simplifies to a free fermionic Hamiltonian with $[J,\nu]-$decay using the third quantisation.
\end{lemma}

\begin{proof}
    Here we are considering the transformed operator $\mathcal{H}_0[\rho] = \sigma_{\beta,0} ^{-1/4} \cdot\mathcal{L}_0^\dagger [\sigma_{\beta,0} ^{1/4} \cdot\rho\cdot \sigma_{\beta,0} ^{1/4}] \cdot \sigma_{\beta,0} ^{-1/4}$, which is Hermitian due to the QDB condition.
    The simplification to free fermions was shown in Proposition \ref{prop - spectrum of free fermionic Lindbladian}, from which we may further define the new Majorana modes to match the Definition \ref{def: Hastings locality for free fermions}.
    The quasi-locality of the parent Hamiltonians for systems with exponentially decaying interactions was shown in Proposition \ref{prop: locality}.
    
    Alternatively, we can study the locality of $\mathcal{H}_0$ directly by considering the decay of elements of the matrices $S = q(4h)^2$ and $A = q(4h)^2 \sinh(2\beta h)$, which give its a-fermionic description, where $q(\nu) = e^{-\beta ^2 \nu^2/8}$, i.e. $q(4h)^2 = e^{-4\beta^2 h^2}$. Note that $h^2$ is a positive semi-definite, real, Hermitian matrix; which has bounded eigenvalues, as $\|h\|_\infty \leq \mathcal{O}(1)$. Hence we can use \cite[Theorem 3.1]{schweitzer2021decayboundsbernsteinfunctions} to say that \begin{equation}|S_{ij}| = \left|[q(4h)^2]_{ij}\right| \leq \exp(-\mathcal{O}(\operatorname{dist}(i,j)))\,,\end{equation} where $\operatorname{dist}(i,j)$ represents the distance on the adjacency graph of the matrix $h^2$. For short-range Hamiltonians, this then shows explicitly that $S$ is quasi-local. $A$ follows similarly as $A = \frac{e^{1/4}}{2} (e^{-4(\beta h - 1/4)^2} - e^{-4(\beta h + 1/4)^2})$, and we can apply the same theorem to these two parts separately. This shows directly that $\mathcal{H}_0$ is in the a-fermionic picture quasi-local for $(k,l)-$local Hamiltonians $H_0$.
\end{proof} 

\begin{lemma}\label{lemma: decay of interaction part}
    Assume further that the interaction term $V$ consists only of terms with even number of Majorana fermions, and that the interactions decay at least exponentially as described in Section \ref{sec: locality}.
    Then the perturbation of the parent Hamiltonian, $\mathcal{V} = \mathcal{H} - \mathcal{H}_0$, has $(K,\mu)$-decay for some constants $K$ and $\mu$, where $K$ is upper bounded by $c |\lambda|$, with $c$ being a constant independent of the system size.
\end{lemma}

\begin{proof}
    Note that $\mathcal{V}$ is Hermitian, as both $\mathcal{H}$ and $\mathcal{H}_0$ are due to their respective detailed-balance conditions. Its explicit form together with detailed calculations for this proof are in Appendix \ref{appendix:detailed bounds}.

    To prove the $(c|\lambda|,\mu)$-decay of $\mathcal{V}$, we shall start by considering $\| \sigma_{\beta} ^{-1/4} L_a \sigma_{\beta} ^{1/4}  - \sigma_{\beta,0} ^{-1/4} L^0_a \sigma_{\beta,0} ^{1/4} \|$ and show it is upper bounded by $c_1|\lambda|$:
    \begin{align}
        \| \sigma_{\beta} ^{-1/4} L_a \sigma_{\beta} ^{1/4}  - \sigma_{\beta,0} ^{-1/4} L^0_a \sigma_{\beta,0} ^{1/4} \| 
        &\leq   \int_{-\infty}^\infty |f^a(t)| \left\|  e^{H(\beta/4 + it)} A^a e^{-H(\beta/4 + it)} - e^{H_0(\beta/4 + it)} A^a e^{H_0(\beta/4 + it)}  \right\|\ \dd t
    \end{align}
    Now we can use Lemma \ref{lemma:bound on perturbed evolution} to say that \begin{align}
        \left\|  e^{H(\beta/4 + it)} A^a e^{-H(\beta/4 + it)} - e^{H_0(\beta/4 + it)} A^a e^{H_0(\beta/4 + it)}  \right\|
    &\leq  |\lambda| |\beta/4 + it| e^{4|\beta/4 + it| \cdot \|h\|_\infty} \max_i \|[V,\omega_i]\|\\
    &\leq c_2 |\lambda| \cdot |\beta/4 + it|\cdot e^{c_3|\beta/4 + it|}\,,\label{eqn:bound on L - L}
    \end{align} which is independent of the system size when we assume that $V$ contains only terms with even numbers of Majorana fermions, as is required for physical Hamiltonians, and has exponentially decaying interactions. Here the second to last inequality follows from submultiplicativity of $\ell_\infty$ norm (which is given by the maximal absolute row sum of the matrix). Finally, we get that \begin{align}
        \| \sigma_{\beta} ^{-1/4} L_a \sigma_{\beta} ^{1/4}  - \sigma_{\beta,0} ^{-1/4} L^0_a \sigma_{\beta,0} ^{1/4} \| &\leq   \int_{-\infty}^\infty |f^a(t)| c_2 |\lambda| \cdot |\beta/4 + it|\cdot e^{c_3|\beta/4 + it|} \dd t = c_1 |\lambda|
    \end{align} for some constant $c_1$ independent of the system size. Hence we can bound the dissipative parts of the parent Hamiltonian like \begin{align}
        \left\| \sigma_{\beta} ^{-1/4} L_a \sigma_{\beta} ^{1/4} \otimes\overline{ \sigma_{\beta} ^{-1/4} L_a \sigma_{\beta} ^{1/4}} - \sigma_{\beta,0} ^{-1/4} L^0_a \sigma_{\beta,0} ^{1/4} \otimes \overline{\sigma_{\beta,0} ^{-1/4} L_a ^{0} \sigma_{\beta,0} ^{1/4}}\right\| &\leq 2c_1 |\lambda|\,,\\
        \left\| \sigma_{\beta} ^{-1/4} L_a ^\dagger L_a \sigma_{\beta} ^{1/4} \otimes I  - \sigma_{\beta,0} ^{-1/4} L_a ^{0\dagger} L^0_a \sigma_{\beta,0} ^{1/4} \otimes I  \right\| &\leq 2c_1 |\lambda|\,,\\
        \left\| I\otimes \overline{\sigma_{\beta} ^{1/4} L_a ^\dagger L_a \sigma_{\beta} ^{-1/4}} - I \otimes \overline{\sigma_{\beta,0} ^{1/4} L_a ^{0\dagger} L^0_a \sigma_{\beta,0} ^{-1/4}} \right\| &\leq 2c_1 |\lambda|\,.
    \end{align}

  Now looking at the coherent term, we shall split it up into quasi-local contributions $G =\sum_a G_a$ with $G_a = \int_{-\infty}^\infty g(t) e^{iHt} (L_a ^\dagger L_a) e^{-iHt}\ \dd t$. Then we similarly need to bound $\| \sigma_{\beta} ^{-1/4} G_a \sigma_{\beta} ^{1/4} -  \sigma_{\beta,0} ^{-1/4} G^0_a \sigma_{\beta,0} ^{1/4}\|$:
    \begin{align}
        \| \sigma_{\beta} ^{-1/4} &G_a \sigma_{\beta} ^{1/4} - \sigma_{\beta,0} ^{-1/4} G^0_a \sigma_{\beta,0} ^{1/4}\| \\
        &\leq \int_{-\infty}^\infty |g(t)| \left(\left\| L^\dagger_a L_a - L^{0\dagger}_a L^0 _a \right\| + \left\| e^{H(\beta/4 + it)} L^{0\dagger}_a L^0_a e^{-H(\beta/4 + it)} -  e^{H_0(\beta/4 + it)} L^{0\dagger}_a L^0_a e^{-H_0(\beta/4 + it)}  \right\|\right) \dd t\,.
    \end{align}
    Here we can again use Lemma \ref{lemma:bound on perturbed evolution} to bound
    \begin{multline} 
        \left\| e^{H(\beta/4 + it)} L^{0\dagger}_a L^0_a e^{-H(\beta/4 + it)} - e^{H_0(\beta/4 + it)} L^{0\dagger}_a L^0_a e^{-H_0(\beta/4 + it)}  \right\| \\ 
        \leq |\lambda| |\beta/4 + it| \max_{s\in [0,1]}
    \left\|\left[V, e^{s(\beta/4 + it) H_0} L^{0\dagger}_a L^0_a e^{-s(\beta/4 + it) H_0}\right]\right\|\,.
    \end{multline}
    Using the exact solution $L^0_a = \sum_i \hat{f}(-4h)_{ai} \omega_i$, we can upper bound this further like \begin{align}
    \max_{s\in [0,1]}\|[V, e^{s(\beta/4 + it) H_0} L^{0\dagger}_a L^0_a e^{-s(\beta/4 + it) H_0}]\| 
    &\leq 2  \| \hat{f}(-4h) \|_\infty ^2 \cdot e^{\beta \|h\|_\infty} \cdot w_h(t) \cdot \max_k \left\|\left[V,\omega_k\right]\right\|\\
    &\leq 2c_2 e^{c_3 \beta/4}  \cdot w_h(t) \cdot \| \hat{f}(-4h) \|_\infty ^2\,,
    \end{align} where $w_h(t)$ is system-size-independent function growing subexponentially in $t$ (as discussed in Appendix \ref{appendix:detailed bounds}).
     Note that we have $\|\hat f(-4h)\|_\infty \leq e^{2\beta^2 \|h\|_\infty ^2 + \beta \|h\|_\infty}$ due to submultiplicativity of the $\ell_\infty$ norm, and so $\|h\|_\infty = \mathcal{O}(1)$ ensures that $\|\hat f(-4h)\|_\infty = \mathcal{O}(1)$. Observe that the previous argument for bounding the conjugated expression $\left\| \sigma_{\beta} ^{-1/4} L_a ^\dagger L_a \sigma_{\beta} ^{1/4} - \sigma_{\beta,0} ^{-1/4} L_a ^{0\dagger} L^0_a \sigma_{\beta,0} ^{1/4} \right\|$ also shows that $\left\|  L_a ^\dagger L_a - L_a ^{0\dagger} L^0_a \right\| \leq c_4 |\lambda|$.
     Finally, this means that \begin{align}
         \| \sigma_{\beta} ^{-1/4} G_a \sigma_{\beta} ^{1/4} -  \sigma_{\beta,0} ^{-1/4} G^0_a \sigma_{\beta,0} ^{1/4}\| &\leq \int_{-\infty}^\infty |g(t)| \cdot \left(  c_4|\lambda| +  |\lambda| |\beta/4 + it| c_5 w_h(t) \right)\ \dd t = c_6 |\lambda|,
    \end{align} where the convergence is ensured by the decay bounds of $g(t)$ obtained in \cite[Lemma 30]{ding2024efficient}.
    
    This proves that the strength of the perturbation of the parent Hamiltonian (in the vectorised picture) is upper bounded by a constant multiple of the strength of the perturbation of the system's Hamiltonian, uniformly in system size, i.e. that $\mathcal{V}_a$, where $\mathcal{V} = \sum_{a \in \mathcal{A}} \mathcal{V}_a$, is upper bounded like $\|\mathcal{V}_a\| \leq c|\lambda|$. To match the formulation of Definition \ref{def: Hastings locality}, we need to express $\mathcal{V}_a$ as a telescoping sum like $\mathcal{V}_a = \mathcal{V}_a ^{(0)} + \sum_{r = 1}^\infty \mathcal{V}_a^{(r)} - \mathcal{V}_a^{(r-1)}$, where $\mathcal{V}_a^{(r)}$ is a truncation of $\mathcal{V}_a$ to the ball $B_r(a)$ of radius $r$ centred at $a$. This truncation then amounts to replacing $H$ by a truncated version $H_{B_r(a)}$ in all the time-evolved formulae of the involved operators. The argument for bounding $\mathcal{V}_a$ directly translates to a bound on $ \mathcal{V}_a^{(r)}$ and hence on $\varepsilon_a^{(r)} = \mathcal{V}_a^{(r)} - \mathcal{V}_a^{(r-1)}$. The quasi-locality of the parent Hamiltonians $\mathcal{H}$ and $\mathcal{H} _0$, and hence that of $\mathcal{V}$, was shown in Proposition \ref{prop: locality}, and stems from the Lieb-Robinson argument in \cite[Proposition 20]{ding2024efficient} proving quasi-locality at any temperature. As these properties are independent, they show together that $\|\varepsilon_a^{(r)}\| \leq c|\lambda| e^{-\mu r}$, and so $\mathcal{V} = \sum_{a \in \mathcal{A}} \sum_{r \geq 0} \varepsilon_a^{(r)}$ has $(c|\lambda|,\mu)$-decay (where $\varepsilon_a^{(0)} \equiv \mathcal{V}_a^{(0)}$).
\end{proof}
  
In Lemma \ref{lemma:general bound decay}, we also present a slightly weaker notion of this result, with the strength bounded by $|\lambda|^\alpha$ for an arbitrary constant $\alpha < 1$ for small enough $|\lambda|$, which works for general Hamiltonians.


\begin{theorem}\label{thm: main, gap}
    Under the assumptions of Lemmas \ref{lemma: decay of free part} and \ref{lemma: decay of interaction part},
    at any inverse temperature $\beta$, there exist positive constants $\lambda_\textup{max}$ and $d$, such that the Lindbladian $\mathcal{L}^\dagger$ corresponding to the perturbed fermionic Hamiltonian $H = H_0 + \lambda V$ has a spectral gap $\Delta$ lower bounded by $\Delta_0 - d| \lambda|$ for any $|\lambda| \leq \lambda_\textup{max}$, where the unperturbed gap $\Delta_0$ is specified in \eqref{eqn: free gap, Gaussian}; independent of system size.
\end{theorem}

\begin{proof}
    We want to bound the gap of $\mathcal{L}^\dagger$, and we wish to use the results about stability of gaps of free fermionic systems under perturbation. Hence we will consider the similarity transformation that will take $\mathcal{L}^\dagger$ to $\mathcal{H} = \mathcal{H}_0 + \mathcal{V}$, where $\mathcal{H}_0$ is a Hermitian free fermionic Hamiltonian. As we have already shown, $\mathcal{H}_0$ has a gap $\Delta_0$ and $[J,\nu]$-decay, while $\mathcal{V}$ has $(K,\mu)$-decay, and hence by the stability Theorem \ref{thm: Hastings stability} there exist constants $K_\text{max}$ and $d_1$, s.t. for all $K \leq K_\text{max}$, the gap of the perturbed parent Hamiltonian, and hence the gap of the Lindbladian, is lower bounded like $\Delta \geq \Delta_0 - d_1 K$. But we also know that $K \leq c |\lambda|$, and hence there exists $\lambda_\text{max} = \frac{K_\text{max}}{c}$ such that whenever $|\lambda| \leq \lambda_\text{max}$, we also have $K \leq K_\text{max}$ and $\Delta \geq \Delta_0 - d_1 c |\lambda|$.
\end{proof}

While other filter functions (potentially in combination with other jump operators) might work significantly better in practice (see Section \ref{sec:simulations}), here we required superexponential decay of the filter function $f(t)$ in the time domain, ensuring the locality of the parent Hamiltonians for systems with exponentially decaying correlations, and the convergence of the integrals appearing in the particular bounds of the strength of the Lindbladian perturbation we use here\,---\,this lead us to use the Gaussian filter.


\subsection{Stability of the Lindbladian Gap Under Perturbations of the Atomic Limit}
\label{sec:stability atomic}

In this section, we investigate the so-called atomic limit\,---\,where interactions among different sites are absent\,---\,and its perturbations. The atomic limit of the spinful Fermi-Hubbard Hamiltonian corresponds to setting $t=0$ in the Hamiltonian of Eq.~\eqref{eq:Hamiltonian FH}:
\begin{align}    
    H_{\mathrm{atomic}} = 
    U\sum_{i=1}^n N_{i,\uparrow} N_{i,\downarrow}
    \,.
\end{align}
The following discussion can be easily generalised to any Hamiltonians that are separable in the lattice sites but for simplicity of exposition we will discuss here only
the Fermi-Hubbard model at $t=0$.
$H_{\mathrm{atomic}}$ is trivially solvable and its eigenstates
are given by electrons localized at the lattice sites.
We will now show that if we choose local fermionic jump operators, the Lindbladian and parent Hamiltonian associated to $H_{\mathrm{atomic}}$
are also separable and 
we can compute exactly their spectrum.

\begin{prop}\label{prop - spectrum of atomic Lindbladian}
    The Lindbladian $\mathcal{L}^\dagger_{\mathrm{atomic}}$ corresponding to the Hamiltonian $H_{\mathrm{atomic}}$ with the set of jump operators $\{\omega_a\}_{a=1}^{2n}$ and Gaussian filter function is gapped for any $\beta\ge 0$ and $U\in\mathbb{R}$.
\end{prop}

\begin{proof}
For each $i\in \{1,\dots,n\}$, $\alpha\in\{\uparrow,\downarrow\}$ we define the Lindblad operators and the operators
$\tilde{L}$
associated to creation and annihilation operators:
\begin{align}
    L_{i,\alpha,-}
    &=
    \int_{-\infty}^{+\infty}
    f(t)
    e^{itH}
    a_{i,\alpha}
    e^{-itH}
    =
    \int_{-\infty}^{+\infty}
    f(t)
    e^{-it UN_{i,\bar{\alpha}}}
    a_{i,\alpha}
    =
    \hat{f}(-UN_{i,\bar{\alpha}})
    a_{i,\alpha}
    \\
    L_{i,\alpha,+}
    &=
    \int_{-\infty}^{+\infty}
    f(t)
    e^{itH}
    a^\dagger_{i,\alpha}
    e^{-itH}
    =
    \hat{f}(+UN_{i,\bar{\alpha}})
    a^\dagger_{i,\alpha}
    \,,
    \\
    \tilde{L}_{i,\alpha,-}
    &=
    \int_{-\infty}^{+\infty}
    f(t)
    e^{(\beta/4+it)H}
    a_{i,\alpha}
    e^{-(\beta/4+it)H}
    =
    q_-(N_{i,\bar{\alpha}})a_{i,\alpha}
    \,,\\ 
    \tilde{L}_{i,\alpha,+}
    &=
    \int_{-\infty}^{+\infty}
    f(t)
    e^{(\beta/4+it)H}
    a_{i,\alpha}^\dagger
    e^{-(\beta/4+it)H}
    =
    q_+(N_{i,\bar{\alpha}})a^\dagger_{i,\alpha}
    \,,
\end{align}
where $\bar{\alpha}$ is the opposite direction of the spin $\alpha$, we used $[N,a]=-a, [N,a^\dagger]=a^\dagger$ to compute the time evolution of the oscillators, and defined 
\begin{align}
    q_{\pm}(x)\equiv q(\pm U x)
    \,.
\end{align}
Recall that the function $q$ is related to the filter function as in \eqref{eq:q_fcn} and $\overline{q_+(\nu)}=q_-(\nu)$.
Note that 
\begin{align}
\tilde{L}^\dagger_{i,\alpha,+}
=
\overline{q_+(N_{i,\bar{\alpha}})}a_{i,\alpha}
=\tilde{L}_{i,\alpha,-}
\,.
\end{align}
Now we take the self-adjoint Majorana operators as jump operators:
\begin{align}
    \omega_{i,\alpha,0}
    =
    \frac{1}{\sqrt{2}}
    (a_{i,\alpha}+a_{i,\alpha}^\dagger)
    \,,\quad 
    \omega_{i,\alpha,1}
    =
    \frac{-i}{\sqrt{2}}
    (a_{i,\alpha}-a_{i,\alpha}^\dagger)
    \,,
\end{align}
and the corresponding self-adjoint operators:
\begin{align}
    \tilde{L}_{i,\alpha,0}=\frac{1}{\sqrt{2}}
    (\tilde{L}_{i,\alpha,-}+\tilde{L}_{i,\alpha,+})
    =
    \tilde{L}_{i,\alpha,0}^\dagger
    \,,\quad
    \tilde{L}_{i,\alpha,1}=
    \frac{-i}{\sqrt{2}}
    (\tilde{L}_{i,\alpha,-}-\tilde{L}_{i,\alpha,+})
    =
    \tilde{L}_{i,\alpha,1}^\dagger
\,.
\end{align}
For the operators $L_a^\dagger L_a$, we have
\begin{align}
    L_{i,\alpha,-}^\dagger L_{i,\alpha,-}
    &=
    |\hat{f}(-UN_{i,\bar{\alpha}})|^2 N_{i,\alpha}
    \,,\quad
    L_{i,\alpha,+}^\dagger L_{i,\alpha,+}
    =
    |\hat{f}(+UN_{i,\bar{\alpha}})|^2 (1-N_{i,\alpha})
    \,,\quad 
    L_{i,\alpha,+}^\dagger L_{i,\alpha,-}
    =
    L_{i,\alpha,-}^\dagger L_{i,\alpha,+}
    =
    0\,.
\end{align}
Note that $[H, L_{i,\alpha,z}^\dagger L_{i,\alpha,z}]=0$ for all indices so that only the $\nu=0$ component of $(L_{i,\alpha,z}^\dagger L_{i,\alpha,z})_\nu$ is non-zero and thus $G=0$. Further, denoted
\begin{align}
    \label{eq:F atomic}
    F(N_{i,\alpha}, N_{i,\bar{\alpha}})
    :=
    \frac{1}{2}
    (L_{i,\alpha,+}^\dagger L_{i,\alpha,+} + 
    L_{i,\alpha,-}^\dagger L_{i,\alpha,-})
    =
    \frac{1}{2}
    (
    |\hat{f}(+UN_{i,\bar{\alpha}})|^2 (1-N_{i,\alpha})
    +
    |\hat{f}(-UN_{i,\bar{\alpha}})|^2 N_{i,\alpha}
    )
    \,,
\end{align}
we have
\begin{align}
    \tilde{M}_{i,\alpha,0}
    &=
    L_{i,\alpha,0}^\dagger L_{i,\alpha,0}
    =
    \frac{1}{2}
    (L_{i,\alpha,+}^\dagger + L_{i,\alpha,-}^\dagger)
    (L_{i,\alpha,+} + L_{i,\alpha,-})
    =
    F(N_{i,\alpha}, N_{i,\bar{\alpha}})
    \\
    \tilde{M}_{i,\alpha,1}
    &=
    L_{i,\alpha,1}^\dagger L_{i,\alpha,1}
    =
    \frac{1}{2}
    (L_{i,\alpha,+}^\dagger - L_{i,\alpha,-}^\dagger)
    (L_{i,\alpha,+} - L_{i,\alpha,-})
    =
    F(N_{i,\alpha}, N_{i,\bar{\alpha}})
    \,.
\end{align}

Then the Lindbladian and parent Hamiltonian are
\begin{align}
    \mathcal{L}&=
    \sum_{i=1}^n\sum_{\alpha\in\{\uparrow,\downarrow\}}
    \sum_{z=0}^1
    \mathcal{L}_{i,\alpha,z}
    \,,
\quad 
    \mathcal{L}_{i,\alpha,z}(\rho)
    =
    L_{i,\alpha,z}\rho L_{i,\alpha,z}^\dagger
    -
    \frac{1}{2}\{ F(N_{i,\alpha}, N_{i,\bar{\alpha}}), \rho \}
    \\
    \mathcal{H}&=
    \sum_{i=1}^n\sum_{\alpha\in\{\uparrow,\downarrow\}}
    \sum_{z=0}^1
    \mathcal{H}_{i,\alpha,z}
    \,,
    \quad
    \mathcal{H}_{i,\alpha,z}(\rho)
    =
    \widetilde{L}_{i,\alpha,z}\rho \widetilde{L}_{i,\alpha,z}
    -
    \frac{1}{2}\{ F(N_{i,\alpha}, N_{i,\bar{\alpha}}), \rho \}
    \,.
\end{align}
They are separable and therefore
exactly solvable. 
The eigenstates 
of $\mathcal{H}$ are of the form 
\begin{align}
    \rho = \bigotimes_{i=1}^n 
    \rho_i \,,
\end{align}
where $\rho_i$ is a solution to the reduced eigenproblem:
\begin{align}
\label{eq:eigenval_single_site_t0}
    \sum_{z=0}^1
    (\tilde{L}_{\uparrow,z}\sigma \tilde{L}_{\uparrow,z}
    +
    \tilde{L}_{\downarrow,z}\sigma \tilde{L}_{\downarrow,z}
    )
    -
    \{ 
    F(N_{\uparrow}, N_{\downarrow})
    +
    F(N_{\downarrow}, N_{\uparrow})
    , \sigma \}
    =
    e \sigma\,,
\end{align}
where we suppress the $i$ index for notational simplicity.
To diagonalise this, we use a Jordan-Wigner transformation where we identify $\uparrow\,\equiv 1$, $\downarrow\,\equiv 2$:
\begin{align}
    N_\alpha = \frac{1}{2}(\id + Z_\alpha)
    \,,\quad 
    a_1 = \sigma_1^-
    \,,\quad 
    a_1^\dagger = \sigma_1^+
    \,,\quad 
    a_2 = Z_1 \sigma_2^-
    \,,\quad 
    a_2^\dagger = Z_1 \sigma_1^+
    \,,
\end{align}
so that
\begin{align}
    \widetilde{L}_{\uparrow,0}
    &=
    \frac{1}{\sqrt{2}}
    (
    q_-(N_{2})a_{1}
    +
    q_+(N_{2})a^\dagger_{1})
    =
    \frac{1}{\sqrt{2}}
    (
    q_-(N_{2})\sigma^-_{1}
    +
    q_+(N_{2})\sigma^{+}_{1})
    \\
    \widetilde{L}_{\downarrow,0}
    &=
    \frac{1}{\sqrt{2}}
    (
    q_-(N_{1})a_{2}
    +
    q_+(N_{1})a^\dagger_{2})
    =
    \frac{1}{\sqrt{2}}
    (
    q_-(N_{1})Z_1\sigma^-_{2}
    +
    q_+(N_{1})Z_1\sigma^{+}_{2})
    \\
    \widetilde{L}_{\uparrow,1}
    &=
    \frac{-i}{\sqrt{2}}
    (
    q_-(N_{2})a_{1}
    -
    q_+(N_{2})a^\dagger_{1})
    =
    \frac{-i}{\sqrt{2}}
    (
    q_-(N_{2})\sigma^{-}_{1}
    -
    q_+(N_{2})\sigma^+_{1})
    \\
    \widetilde{L}_{\downarrow,1}
    &=
    \frac{-i}{\sqrt{2}}
    (
    q_-(N_{1})a_{2}
    -
    q_+(N_{1})a^\dagger_{2})
    =
    \frac{-i}{\sqrt{2}}
    (
    q_-(N_{1})Z_1\sigma^{-}_{2}
    -
    q_+(N_{1})Z_1\sigma^+_{2})
    \,.
\end{align}
The stationary state of the Lindbladian 
is $\sigma_\beta$ and its eigenvalue is $0$.
The gap then corresponds to the highest non-zero eigenvalue of 
\eqref{eq:eigenval_single_site_t0}\,---\,recall that the spectrum is non-positive.
It can be easily computed by diagonalising numerically a $16\times 16$ matrix corresponding to the vectorised reduced Hamiltonian.
We plot the result in Figure \ref{fig:gap_spinful_t0}.
We see that the gap is non-zero and goes to $0$ as $|U|\to\infty$.
This result holds for any system size.

\begin{figure}[h]
    \centering
    \begin{tikzpicture}

\definecolor{zero}{RGB}{180, 180, 180}
\definecolor{one}{RGB}{206, 104, 104}
\definecolor{two}{RGB}{231, 195, 146}
\definecolor{three}{RGB}{231, 217, 146}
\definecolor{four}{RGB}{209, 223, 141}
\definecolor{five}{RGB}{116, 185, 116}
\definecolor{six}{RGB}{126, 147, 165}
\definecolor{seven}{RGB}{146, 136, 176}
\definecolor{eight}{RGB}{166, 124, 166}

\begin{axis}[
    height=.43\columnwidth,
    width=.55\columnwidth,
    xmin=-5.1, xmax=5.1,
    ymin=0,
    xlabel={$U$}, ylabel={$\Delta$},
    axis on top,
    ]

    \addplot[mark=none, color=one, thick] table[x=U, y=Delta_nqb4_beta1] {data/usweep_t0.dat};

\end{axis}

\end{tikzpicture}
    \caption{Gap of the Lindbladian for the spinful Fermi-Hubbard model at $t=0$.
    We set $\beta=1$, fix the Gaussian filter and vary $U$.
    }
    \label{fig:gap_spinful_t0}
\end{figure}
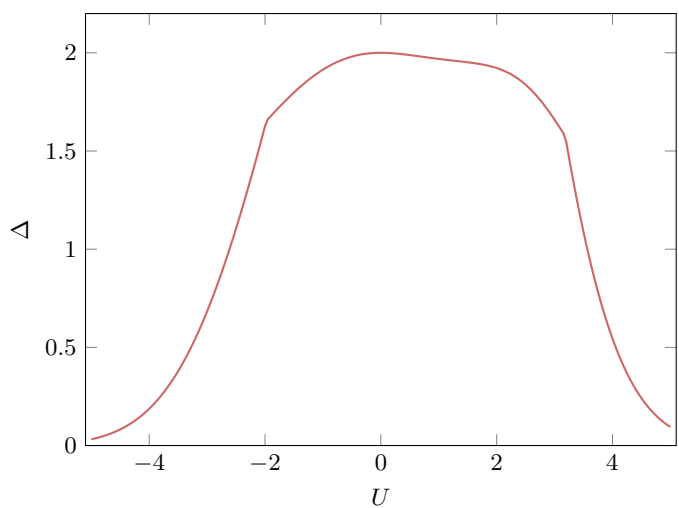

\end{proof}

The gap, however, closes for zero temperature as the next proposition shows.

\begin{prop}
    The gap of the Lindbladian $\mathcal{L}^\dagger_{\mathrm{atomic}}$ corresponding to the atomic Hamiltonian $H_{\mathrm{atomic}}$ with the set of jump operators $\{\omega_a\}_{a=1}^{2n}$ and Gaussian filter function closes as $\beta\to\infty$.
\end{prop}
\begin{proof}
We shall show that $\id$ becomes degenerate with $\sigma_\beta$ as $\beta\to\infty$ for any system size.
We shall use here results  and notations from the proof of Prop.~\ref{prop - spectrum of atomic Lindbladian}.
Since $G=0$, we have
\begin{align}
\mathcal{L}^\dagger_{\mathrm{atomic}}(\id)
    =    \sum_{i=1}^n
    \sum_{\alpha\in\{\uparrow,\downarrow\}}
    \sum_{z=0}^1 
    [L_{i,\alpha,z},L_{i,\alpha,z}^\dagger]
\end{align}
We know from equation  \eqref{eq:F atomic}
that
\begin{align}
    L_{i,\alpha,+}^\dagger L_{i,\alpha,+} + 
    L_{i,\alpha,-}^\dagger L_{i,\alpha,-}
    =
    |\hat{f}(+UN_{i,\bar{\alpha}})|^2 (1-N_{i,\alpha})
    +
    |\hat{f}(-UN_{i,\bar{\alpha}})|^2 N_{i,\alpha}    
\end{align}
and we also have
\begin{align}
    L_{i,\alpha,+}
    L_{i,\alpha,+}^\dagger 
    + 
    L_{i,\alpha,-}  
    L_{i,\alpha,-}^\dagger
    =
    |\hat{f}(+UN_{i,\bar{\alpha}})|^2 N_{i,\alpha}
    +
    |\hat{f}(-UN_{i,\bar{\alpha}})|^2 (1-N_{i,\alpha})   
\end{align}
These two expressions cancel so that $\mathcal{L}^\dagger_{\mathrm{atomic}}(\id)=0$ if $\hat{f}$ is even.
From \eqref{eq:Gaussian filter} we see that for large $\beta$ the Gaussian filter is indeed even, which proves the proposition.
\end{proof}

Next we discuss the robustness of the gap of $\mathcal{L}^\dagger_{\mathrm{atomic}}$
under weak interactions among sites.
This follows by realising that the associated parent Hamiltonian is geometrically local and frustration-free\,---\,that is, it is a sum of local terms and the ground state is an eigenstate of each of them\,---\,and that the perturbation to the parent Hamiltonian is quasi-local.
We can then use results about the stability of the gap for frustration-free Hamiltonians, such as \cite{Bravyi_2011,
Michalakis_2013,hastings2017stabilityfreefermihamiltonians} and also the older works on quantum perturbations of classical systems \cite{datta1996low,borgs2000low}, to prove the stability of the gap of the perturbed Lindbladian.
In particular, we can prove the following.

\begin{theorem}\label{thm: gap stability atomic Lindbladian}
Consider the interacting fermionic Hamiltonian $H=H_{\mathrm{atomic}}+\lambda V$ and assume that $V$ has interactions that decay at least exponentially.
Then at any inverse temperature $\beta$, there exist positive constants $\lambda_\textup{max}$ and $d$, such that the Lindbladian $\mathcal{L}^\dagger$ corresponding to the perturbed Hamiltonian $H$ has a spectral gap $\Delta$ lower bounded by $\Delta_\mathrm{atomic} - d| \lambda|^\alpha$ for any $|\lambda| \leq \lambda_\textup{max}$ and arbitrary positive constant $\alpha<1$, where $\Delta_\mathrm{atomic}$ is 
the gap of the Lindbladian $\mathcal{L}^\dagger_{\mathrm{atomic}}$ 
of Prop.~\ref{prop - spectrum of atomic Lindbladian}. 
\end{theorem}

\begin{proof}
Let us denote by $e_{\ell}\le \cdots \le e_1 <e_0\equiv 0$ the set of eigenvalues of the single site parent Hamiltonian in the atomic limit, equation \eqref{eq:eigenval_single_site_t0}.
Then $|e_1|>0$ is the Lindbladian gap in the atomic limit.
Now we consider the Hamiltonian
\begin{align}
    \mathcal{H}_0
    =
    \sum_{i=1}^n
    \sum_{a=0}^{\ell}
    E_a
    (\mathcal{P}_{e_a})_i
    \,,\quad 
    E_a = -\frac{e_a}{e_1}\,,
\end{align}
where $(\mathcal{P}_e)_i$ is the projector onto the eigenspace of $e$ acting at site $i$.
$\mathcal{H}_0$ is a rescaled and negated version of the parent Hamiltonian at $t=0$ such that 
$\mathcal{H}_0\ge 0$. It has a non-degenerate ground state with eigenvalue $0$ given by the Fermi-Hubbard thermal state at $t=0$ and gap $1$.
Now we denote by $\tilde{\mathcal{H}}(\lambda)$ minus the parent Hamiltonian associated to $H$, 
so that the spectrum is positive and the thermal state is the ground state. 
Then $\tilde{\mathcal{H}}(0)=e_1
    \mathcal{H}_0$ and we define
\begin{align}
    \mathcal{V} = 
    \tilde{\mathcal{H}}(\lambda)
    -
    \tilde{\mathcal{H}}(0)
    =
    \tilde{\mathcal{H}}(\lambda)
    -
    e_1
    \mathcal{H}_0
    \,.
\end{align}
Next we will prove a gap for $\mathcal{H}=\mathcal{H}_0+\mathcal{V}$ using the following result.

\begin{lemma}[Theorem 4 of \cite{hastings2017stabilityfreefermihamiltonians}]\label{thm:hastings_4}
There exist constant $J_0, c_1$ depending only on 
$\tilde{J},\tilde{\mu}, \mu, D$ such that the following holds for all $J\le J_0$.
Let $H_0$ have
$(\tilde{J}, \tilde{\mu})$ decay and let $V$ have $(J, \mu)$ decay, according to the notion of decay in Definition \ref{def:1 Hastings}.
Assume  $H_0 \ge H_{\text{proj}}$ for some $H_{\text{proj}}$ which is a sum of commuting projectors and which obeys the topological quantum order conditions of \cite{Bravyi_2011}. Assume
also $P H_0 = 0$ where $P$ is the projector onto the ground state of $H_{\text{proj}}$.
Then the spectral gap of $H_0 + V$ 
is at least
$1-c_1J-\delta$,
for some $\delta$ bounded by $J$ times a quantity decaying faster than any power of $L$.
\end{lemma}

We are going to use this result with $H_0$ identified with $\mathcal{H}_0$.
This means that $\tilde{\mu}=\infty$.
Then we define a Hamiltonian $\mathcal{H}_\text{proj}$ which we identify with $H_{\text{proj}}$ by:
\begin{align}
    \mathcal{H}_\text{proj}
    =
    \sum_{i=1}^n
    Q_i
    \,,\quad 
    Q_i=
    1-P_i
    \,,\quad 
    P_i \equiv (\mathcal{P}_0)_i
    \,.
\end{align}
Let us denote 
$P_A = \prod_{i\in A}P_i$ for a set of sites $A$, 
and the projector onto the ground state of 
$\mathcal{H}_\text{proj}$ by
\begin{align}
    P = \prod_{i=1}^n 
    (1-Q_i)
    =
    \prod_{i=1}^n 
    P_i
    \,.
\end{align}
Note that the ground state of $\mathcal{H}_\text{proj}$ is non-degenerate and $P_A$ is a rank one projector for any set of sites $A$.
We have the following properties:
\begin{enumerate}
    \item $\mathcal{H}_0-\mathcal{H}_\text{proj}\ge 0$. Indeed $P\mathcal{H}_0=P\mathcal{H}_\text{proj}=0$ so that
    $\mathcal{H}_0,\mathcal{H}_\text{proj}$ have the same ground state.
    Also they have the same gap $1$ and the other eigenvalues of $\mathcal{H}_0$ are greater or equal to those of $\mathcal{H}_\text{proj}$ since $[\mathcal{H}_0, \mathcal{H}_\text{proj}] = 0$ and $E_a \ge 1$ for $a>1$.
    \item $\mathcal{H}_\text{proj}$ is the sum of local commuting projectors.
    \item $\mathcal{H}_\text{proj}$ satisfies the topological quantum order conditions \cite{Bravyi_2011}
    \begin{enumerate}
        \item TQO-1: if $O_A$ is supported on $A$ then 
        \begin{align}
            PO_AP = 
            P_AO_AP_A P_{A^\perp}
            =\Tr(P_AO_A)P=cP
        \end{align}
        where $A^\perp = \Lambda \setminus A$. This happens for any $A$, so $L^* = L$, the system size.
        \item TQO-2: If $PO_A=P_AO_A P_{A^\perp}=0$, then $P_AO_A=0$, and so $P_BO_A=0$ if $B$ includes $A$.
    \end{enumerate}
\end{enumerate}

Finally we identify $V$ with $\mathcal{V}$ and 
the last hypothesis of Lemma \ref{thm:hastings_4} to verify the $(J,\mu)$ decay of $\mathcal{V}$.
This follows from Lemma \ref{lemma:general bound decay} which shows that $J\le c|\lambda|^\alpha$ for any $0<\alpha<1$.
Lemma \ref{thm:hastings_4} then implies that
$\mathcal{H}$ is gapped for all $|\lambda|^\alpha\le J_0/c$
and that the gap is $1-\mathcal{O}(|\lambda|^\alpha)$.
This in turns implies a gap 
$|e_1|-\mathcal{O}(|\lambda|^\alpha)$
for 
$\tilde{\mathcal{H}}(\lambda)
=
e_1(\mathcal{H}_0+e_1^{-1}\mathcal{V})$ 
and thus for the Lindbladian $\mathcal{L}^\dagger$
 for all 
$|\lambda|\le \lambda_{\text{max}}$. Here $\lambda_{\text{max}}=(|e_1|J_0/c)^{1/\alpha}$ since the strength of the perturbation in 
$\tilde{\mathcal{H}}(\lambda)/e_1$ 
is $J\le c |e_1^{-1}| |\lambda|^\alpha$.
\end{proof}

Note that Theorem \ref{thm: gap stability atomic Lindbladian} in particular implies a gap for the Lindbladian associated with 
the spinful Fermi-Hubbard model for weak interactions around the atomic limit by identifying $\lambda$ with $t$.
Note also that similar conclusions about perturbations of the atomic limit can be drawn with Pauli jump operators that we discuss in Section \ref{sec:beyond}. This is because the atomic Lindbladian and parent Hamiltonian remain separable also in that case.\\

We end this section with the remark that the spinless Fermi-Hubbard model at $t=0$ is not separable, see Eq.~\eqref{eq:Hamiltonian pFH}. Thus it does not reduce to an atomic limit and the results of this section do not apply to perturbations of the $t=0$ limit in the spinless case. However, commutativity of the Hamiltonian implies that the Lindbladian and the parent Hamiltonians for the $t=0$ spinless Fermi-Hubbard model are strictly local. We leave investigations of this model for future work.


\subsection{Efficient Quantum Gibbs Sampler}\label{sec:efficiency}

In this section, we shall finally explain how a constant lower bound on the spectral gap of the Lindbladian translates to the bound on the mixing time and hence the overall algorithmic complexity of the Gibbs state preparation, proving its efficiency.

\begin{corollary}
    The mixing time $t_\textup{mix}$ of the Lindbladian $\mathcal{L}^\dagger$ can be then bounded like \begin{equation}
    t_\textup{mix} \leq  \frac{\log(\frac{2}{\epsilon}\| \sigma_\beta^{-1/2}\|)}{\Delta} = \frac{\mathcal{O}(\beta \| H\| + \log(1/\epsilon))}{\Delta} = \mathcal{O}(n + \log(1/\epsilon)).
\end{equation}
\end{corollary}

\begin{proof}
    As per \cite[Proposition E.4]{chen2023efficient}, we can bound \begin{equation}
        \left\| e^{\mathcal{L}^\dagger t} [\rho_1 - \rho_2] \right\|_{\Tr} \leq e^{-\Delta(\mathcal{L}^\dagger) t} \left\| \sigma^{-1/2}_\beta \right\| \|\rho_1 -\rho_2\|_{\Tr}
    \end{equation} using the Hölder's inequality; hence by taking $\rho_2 = \sigma_\beta$ to be the fixed point of the evolution, we get that \begin{equation}
        \left\| e^{\mathcal{L}^\dagger t} [\rho_1] - \sigma_\beta \right\|_{\Tr} \leq e^{-\Delta(\mathcal{L}^\dagger) t} \left\| \sigma^{-1/2}_\beta \right\| \|\rho_1 -\sigma_\beta\|_{\Tr} \leq 2 e^{-\Delta(\mathcal{L}^\dagger) t} \left\| \sigma^{-1/2}_\beta \right\| \overset{\text{set}}{\leq} \epsilon\,.
    \end{equation} The last inequality is then guaranteed whenever $t  \geq \frac{\log \left(\frac{2}{\epsilon}\left\| \sigma^{-1/2}_\beta \right\| \right)}{\Delta(\mathcal{L}^\dagger)}$, from which we can deduce \begin{equation}
          t_\text{mix}  \leq \frac{\log \left(\frac{2}{\epsilon}\left\| \sigma^{-1/2}_\beta \right\| \right)}{\Delta(\mathcal{L}^\dagger)}\,.
    \end{equation} The rest of the bound follows from the spectral gap being lower bounded by a constant and $\|H\| = \mathcal{O}(n)$.
\end{proof}

\begin{corollary}\label{cor:main-result}
    The purified Gibbs state can be prepared on a quantum computer at any constant temperature via Hamiltonian simulation of the parent Hamiltonian in \begin{equation}\widetilde{\mathcal{O}}(n^3 \operatorname{polylog}(1/\epsilon))\end{equation} time complexity using $\mathcal{O}(n)$ qubits, where $\epsilon$ is the desired precision in trace norm and $\widetilde{\mathcal{O}}$ notation absorbs subdominant polylogarithmic terms.
\end{corollary}

\begin{proof}
    This follows from equation \eqref{eq:runtime_quantum} by using 
    the upper bound on the mixing time, and the fact that $|\mathcal{A}| = \mathcal{O}(n)$ using the Majorana jump operators.
\end{proof}


\subsection{Calculating Partition Functions}\label{sec: partition function}

As a possible application of the efficient Gibbs state preparation, we adapt the strategy from \cite{rouze2024optimalquantumalgorithmgibbs} for calculating partition functions $Z_\beta (\lambda_i) = \Tr(e^{-\beta H(\lambda_i)})$ to the case of interacting fermionic systems, where we have denoted $H(\lambda_i) = H_0 + \lambda_i V$. We remark that since we are assuming $\beta$ to be constant (although arbitrarily large), the method in \cite{rouze2024optimalquantumalgorithmgibbs} based on cooling the partition function from infinite temperature is directly applicable, as their restriction to high temperatures stems only from the Gibbs state preparation. However, since we can consider $\beta$ to be large and the coupling strength $\lambda$ to be small, it might be more efficient in practice to consider systematically increasing the coupling strength rather than decreasing the temperature. Note that we can calculate the non-interacting partition function explicitly as \begin{equation}
    Z_\beta (0) = \prod_{i=1}^n 2 \cosh(2\beta \epsilon_i)\,,
\end{equation} where the product is taken only over one $\epsilon_i \in \operatorname{spec}(h)$ from each symplectic pair $\pm \epsilon_i$. By measuring the observable $e^{\beta H(\lambda_i)} e^{-\beta H(\lambda_{i+1})}$ in the state $\sigma_\beta (\lambda_i) = \frac{e^{-\beta H(\lambda_i)}}{Z_\beta(\lambda_i)}$, we would obtain the ratio $\frac{Z_\beta(\lambda_{i+1})}{Z_\beta (\lambda_i)}$. Preparing this observable and the Gibbs state will require access to block encodings of $H_0$ and $V$, from which we get a block encoding for $H(\lambda_i)$ via LCU, and hence block encoding for the observable and the Hamiltonian simulation via QSVT. By choosing a schedule $0 = t_1 \leq t_2 \leq \dots \leq t_{l-1} \leq t_l = |\lambda|$ and denoting $\lambda_i = t_i \frac{\lambda}{|\lambda|} $, we can calculate $Z_\beta(\lambda)$ as a telescoping product \begin{equation}
    Z_\beta(\lambda) = Z_\beta(0) \prod_{i=1}^{l-1} \frac{Z_\beta(\lambda_{i+1})}{Z_\beta (\lambda_i)} = Z_\beta(0) \prod_{i=1}^{l-1} \Tr(e^{\beta H(\lambda_i)} e^{-\beta H(\lambda_{i+1})} \cdot \sigma_\beta(\lambda_i) )\,.
\end{equation}

Hence we can show the following adaptation of \cite[Theorem 8]{rouze2024optimalquantumalgorithmgibbs}:
\begin{corollary}
    For quasi-local interacting fermionic Hamiltonians $H(\lambda) = H_0 + \lambda V$, at any inverse temperature $\beta$, there exists a positive constant $\lambda_\textup{max}$ such that we can calculate an estimate to the partition function $Z_\beta(\lambda)$ up to a relative error $\epsilon$ with success probability at least $3/4$ for any $|\lambda| \leq \lambda_\textup{max}$ in time complexity $\widetilde{\mathcal{O}}(n^5 \epsilon^{-2})$.
\end{corollary}

We refer to \cite[Appendix C]{rouze2024optimalquantumalgorithmgibbs} for the details of these calculations, the gist of which lies in choosing the schedule such that $t_{i+1}-t_i = \Theta(n^{-1})$, and so $l = \Theta(n)$ as $\lambda = \Theta(1)$. Then we would prepare the Gibbs states $\sigma_\beta(\lambda_i)$ and measure the expectation values of the observables $e^{\beta H(\lambda_i)} e^{-\beta H(\lambda_{i+1})}$ for each $i \in [l-1]$ at least $\Theta(n\epsilon^{-2})$ times. Calculating the estimate for each $\frac{Z_\beta(\lambda_{i+1})}{Z_\beta (\lambda_i)}$ as the average over these measurements, and evaluating the estimate $\hat{Z}_\beta(\lambda)$ to the partition function using the telescoping product would hence ensure \begin{equation}
    \mathbb{P}\left((1-\epsilon)Z_\beta(\lambda) \leq \hat{Z}_\beta(\lambda) \leq  (1+\epsilon)Z_\beta(\lambda)\right) \geq 3/4\,.
\end{equation}


\section{Numerical simulations}\label{sec:simulations}

\subsection{Analytically Bounded Regime}
\label{sec:analytically-bounded}

In this section, we investigate the Fermi-Hubbard model in the parameter range where Theorem~\ref{thm: main, gap} holds for the gap of the Lindbladian, i.e. where $U$ is sufficiently small. Without loss of generality, we set $t = 1$ for all calculations where nothing to the contrary is explicitly stated.

Even though, as mentioned before, the one-dimensional Fermi-Hubbard model can be solved exactly using a Bethe-ansatz~\cite{essler2005one}, we limit our numerical analysis to this 1D setting. The main motivation for this is that we want to show finite-size scaling properties, which would demand too many classical resources for higher-dimensional settings.

Therefore the relevant Hamiltonian is that of a non-periodic fermionic chain, which in the spinless case has the free free fermionic part
\begin{equation}\label{eqn: free 1D Hamiltonian}
    H_0 = -t \sum_{i=1}^{n-1} (a_i ^ \dagger a_{i+1} + a_{i+1}^\dagger a_i) \eqqcolon  \sum_{i,j} \omega_i h_{ij} \omega_j\,.
\end{equation}
The spectrum of its single-particle Hamiltonian can be solved explicitly:
\begin{equation}
    \operatorname{spec}(h) = 2\times \left\{ \frac{t}{2}\cos\left( \pi \cdot \frac{k}{n+1} \right) \right\}_{k=1}^{n}\,.
\end{equation}
Recall that the main Theorem~\ref{thm: main, gap} holds for Gaussian filter functions and Majorana jump operators. Choosing $\hat f(\nu) = e^{-(\beta\nu+1)^2/8 + 1/8}$ as the filter, we can use Eq.~\eqref{eqn: free gap, Gaussian} to determine the gap of the Lindbladian as
\begin{align}
    \Delta_0 &= 2 e^{-4\beta ^2 \lVert h\rVert^2} \cosh(2\beta \lVert h\rVert)\nonumber\\
    &= 2 e^{-\beta ^2 t^2 \cos\left( \frac{\pi}{n+1} \right)^2} \cosh(\beta t\cos\left(  \frac{\pi}{n+1} \right))\label{eq:gauss-delta0}\\ 
    &\geq 2 e^{-\beta ^2 t^2 } \cosh(\beta t) \eqqcolon \underline{\Delta}_0\nonumber
\end{align}

Figure~\ref{fig:spinless-fermi-hubbard} shows the gap $\Delta$ of the Lindbladian for the spinless Fermi-Hubbard model at $\beta = 3$ across various system sizes up to 11 sites. The analytical result from Eq.~\eqref{eq:gauss-delta0}, drawn in dashed grey (\ref{line:spinless_analytical}), matches the markers of the numerical simulations at $U = 0$ (\ref{line:spinless_U0}). As $U$ increases, $\Delta$ continuously deviates from this analytical result; initially shrinking the gap across the whole observed range, but for larger $U$ only decreasing the gap for small system sizes, while seemingly saturating earlier and at larger values as the system size increases. This can be seen as a pointer that the actual asymptotic behaviour of this system might be even better than the analytical results suggest.

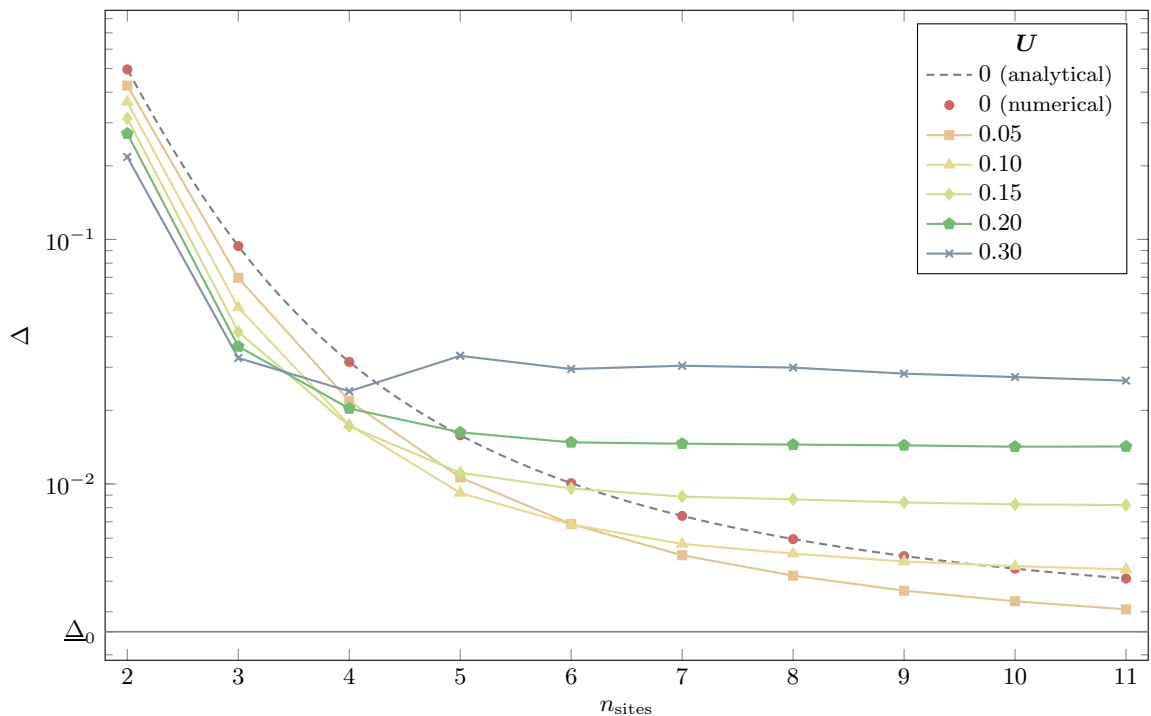
\begin{figure}[htb]
    \centering
    \begin{tikzpicture}

\definecolor{zero}{RGB}{180, 180, 180}
\definecolor{one}{RGB}{206, 104, 104}
\definecolor{two}{RGB}{231, 195, 146}
\definecolor{three}{RGB}{231, 217, 146}
\definecolor{four}{RGB}{209, 223, 141}
\definecolor{five}{RGB}{116, 185, 116}
\definecolor{six}{RGB}{126, 147, 165}
\definecolor{seven}{RGB}{146, 136, 176}
\definecolor{eight}{RGB}{166, 124, 166}

\begin{axis}[
    height=.6\columnwidth,
    width=.9\columnwidth,
    xmin=1.8, xmax=11.2,
    ymin=0.0019,
    unbounded coords=jump,
    ymode=log,
    xlabel={$n_\mathrm{sites}$}, ylabel={$\Delta$},
    axis on top,
    xtick={2, 3, ..., 11},
    extra y ticks={0.002484896389},
    extra y tick labels={$\underline{\Delta}_0$},
    extra y tick style={grid=major,grid style={semithick,color=gray}},
    legend cell align=left,
    ]
    \addlegendimage{empty legend}
    \addlegendentry{\hspace{1.5em}$\boldsymbol{U}$}
    \addplot[mark=none, color=gray, thick, densely dashed, domain=2:11,samples=150] {2 * exp(-3^2 * 1^2 * cos(deg(pi/(x + 1)))^2) * cosh(3 * 1 * cos(deg(pi / (x + 1))))};
    \addlegendentry{$0$ \footnotesize{(analytical)}} \label{line:spinless_analytical}
    \addplot[mark=*, only marks, mark size=1.5, color=one, thick] table[x=n_qb, y=Delta_beta3_U0_spinless] {data/gap_vs_nqb.dat};
    \addlegendentry{$0$ \footnotesize{(numerical)}} \label{line:spinless_U0}
    \addplot[mark=square*, mark size=1.5, color=two, thick] table[x=n_qb, y=Delta_beta3_U0.05_spinless] {data/gap_vs_nqb.dat};
    \addlegendentry{$0.05$}
    \addplot[mark=triangle*, mark size=2, color=three, thick] table[x=n_qb, y=Delta_beta3_U0.1_spinless] {data/gap_vs_nqb.dat};
    \addlegendentry{$0.10$}
    \addplot[mark=diamond*, mark size=2, color=four, thick] table[x=n_qb, y=Delta_beta3_U0.15_spinless] {data/gap_vs_nqb.dat};
    \addlegendentry{$0.15$}
    \addplot[mark=pentagon*, mark size=2, color=five, thick] table[x=n_qb, y=Delta_beta3_U0.2_spinless] {data/gap_vs_nqb.dat};
    \addlegendentry{$0.20$}
    \addplot[mark=x, mark size=2, color=six, thick] table[x=n_qb, y=Delta_beta3_U0.3_spinless] {data/gap_vs_nqb.dat};
    \addlegendentry{$0.30$}
\end{axis}

\end{tikzpicture}
    \caption{Lindbladian gap $\Delta$ of the spinless Fermi-Hubbard model at inverse temperature $\beta=3$ for small interaction strengths $U$. The dashed grey line and red dots are analytically and numerically, respectively, derived results for $U=0$, with the horizontal line at $\underline{\Delta}_0 \coloneqq 2e^{-\beta^2 t^2} \cosh(\beta t)$ being the lower bound for this line for $n_\mathrm{sites}\rightarrow\infty$.}
    \label{fig:spinless-fermi-hubbard}
\end{figure}

For weak interactions, we can examine the analytical results more closely. For this particular fermionic system, where the strength of the perturbation is $U$, Theorem~\ref{thm: main, gap} takes the form $\Delta \geq \Delta_0 - d \lvert U \rvert$ as long as $\lvert U \rvert \leq U_\mathrm{max}$ with $\Delta$ being the gap of the perturbed system, $\Delta_0$ the gap of the unperturbed system, some constant $d$, and a critical perturbation $U_\mathrm{max}$. It is important to stress that this inequality holds \emph{independent of the system size}.

While it is difficult to numerically determine $U_\mathrm{max}$, we can get some idea about what $d$ might be by looking at the derivative\footnote{The gap $\Delta$ is not differentiable at $U=0$, so we take the derivative in the positive direction $0^+$ (repulsive Fermi-Hubbard model) and the negative direction $0^-$ (attractive FH-model) separately.} of $\Delta$ at $U = 0$. Figure~\ref{fig:spinless-fermi-hubbard-slope} shows a few such instances for different parameters of temperature, spinfulness, and sign of $U$. Because, as stated above, for a given $\beta$ the constant $d$ is independent of the system size, the theorem states that each line is bounded \emph{from above} (lower is better). Indeed, for the considered temperatures the lines seem either decrease or tend to (almost) saturate even at the small system sizes shown in the plot.

\begin{figure}[htbp]
    \centering
    \begin{tikzpicture}

\definecolor{zero}{RGB}{180, 180, 180}
\definecolor{one}{RGB}{206, 104, 104}
\definecolor{two}{RGB}{231, 195, 146}
\definecolor{three}{RGB}{231, 217, 146}
\definecolor{four}{RGB}{209, 223, 141}
\definecolor{five}{RGB}{116, 185, 116}
\definecolor{six}{RGB}{126, 147, 165}
\definecolor{seven}{RGB}{146, 136, 176}
\definecolor{eight}{RGB}{166, 124, 166}

\begin{axis}[
    height=.6\columnwidth,
    width=.9\columnwidth,
    xmin=1.8, xmax=11.2,
    ymin=-.05, ymax=1.5,
    xlabel={$n_\mathrm{qubits}$}, ylabel={$\tilde{d}^\pm$},
    axis on top,
    legend cell align={left},
    legend style={at={(0.5,1.02)}, anchor=south},
    legend columns=2
    ]
    \addplot[mark=*, mark options={fill=white}, mark size=1.5, color=one, thick] table[x=n_qb, y=d_1d_beta1_spinless_pos] {data/gap_initial_slope.dat};
    \addlegendentry{$\beta = 1$, $U > 0$, spinless} \label{line:slope_demo}
    \addplot[mark=square*, mark options={fill=white}, mark size=1.5, color=two, thick] table[x=n_qb, y=d_1d_beta3_spinless_pos] {data/gap_initial_slope.dat};
    \addlegendentry{$\beta = 3$, $U > 0$, spinless}
    \addplot[mark=*, mark options={solid}, mark size=1.5, color=one, thick, densely dashed] table[x=n_qb, y=d_1d_beta1_spinless_neg] {data/gap_initial_slope.dat};
    \addlegendentry{$\beta = 1$, $U < 0$, spinless}
    \addplot[mark=square*, mark options={solid}, mark size=1.5, color=two, thick, densely dashed] table[x=n_qb, y=d_1d_beta3_spinless_neg] {data/gap_initial_slope.dat};
    \addlegendentry{$\beta = 3$, $U < 0$, spinless}
    \addplot[mark=triangle*, mark options={fill=white}, mark size=2, color=five, thick] table[x=n_qb, y=d_1d_beta1_spinful_pos] {data/gap_initial_slope.dat};
    \addlegendentry{$\beta = 1$, $U > 0$, spinful}
    \addplot[mark=diamond*, mark options={fill=white}, mark size=2, color=six, thick] table[x=n_qb, y=d_1d_beta3_spinful_pos] {data/gap_initial_slope.dat};
    \addlegendentry{$\beta = 3$, $U > 0$, spinful}
    \addplot[mark=triangle*, mark options={solid}, mark size=2, color=five, thick, densely dashed] table[x=n_qb, y=d_1d_beta1_spinful_neg] {data/gap_initial_slope.dat};
    \addlegendentry{$\beta = 1$, $U < 0$, spinful}
    \addplot[mark=diamond*, mark options={solid}, mark size=2, color=six, thick, densely dashed] table[x=n_qb, y=d_1d_beta3_spinful_neg] {data/gap_initial_slope.dat};
    \addlegendentry{$\beta = 3$, $U < 0$, spinful}
\end{axis}

\begin{axis}[
    width=.3\columnwidth,
    height=.2\columnwidth,
    at={(.56\columnwidth,.37\columnwidth)},
    xmin=0, xmax=0.03,
    ymax=0,
    axis background/.style={fill=white},
    axis on top,
    scaled ticks=false,
    tick label style={/pgf/number format/fixed},
    xtick={0, 0.03},
    ytick={0},
    xlabel=$U$,
    ylabel=$\delta\Delta(U)$,
    xlabel style={yshift=1.21em},
    ylabel style={yshift=-1.21em},
    ]
    \addplot[mark=none, color=black!19, thin] table[x=U, y=d_nqb2_beta1_spinless] {data/delta_gap_vs_u.dat};
    \addplot[mark=none, color=black!28, thin] table[x=U, y=d_nqb3_beta1_spinless] {data/delta_gap_vs_u.dat};
    \addplot[mark=none, color=black!37, thin] table[x=U, y=d_nqb4_beta1_spinless] {data/delta_gap_vs_u.dat};
    \addplot[mark=none, color=black!46, thin] table[x=U, y=d_nqb5_beta1_spinless] {data/delta_gap_vs_u.dat};
    \addplot[mark=none, color=black!55, thin] table[x=U, y=d_nqb6_beta1_spinless] {data/delta_gap_vs_u.dat};
    \addplot[mark=none, color=black!64, thin] table[x=U, y=d_nqb7_beta1_spinless] {data/delta_gap_vs_u.dat};
    \addplot[mark=none, color=black!73, thin] table[x=U, y=d_nqb8_beta1_spinless] {data/delta_gap_vs_u.dat};
    \addplot[mark=none, color=black!82, thin] table[x=U, y=d_nqb9_beta1_spinless] {data/delta_gap_vs_u.dat};
    \addplot[mark=none, color=black!91, thin] table[x=U, y=d_nqb10_beta1_spinless] {data/delta_gap_vs_u.dat};
    \addplot[mark=none, color=black!100, thin] table[x=U, y=d_nqb11_beta1_spinless] {data/delta_gap_vs_u.dat};
\end{axis}

\end{tikzpicture}
    \caption{Numerically evaluated slope of the gap $\tilde{d}^+ = -\frac{\partial\Delta}{\partial  U}\big\rvert_{U=0^+}$ at positive $U$ for the repulsive and $\tilde{d}^- = \frac{\partial\Delta}{\partial  U}\big\rvert_{U=0^-}$ at negative $U$ for the attractive Fermi-Hubbard model at different inverse temperatures $\beta$ depending on the system size. Note that for the model with spin, the $x$-axis is the number of total sites when counting different spins separately, which matches the number of qubits required to represent the system. \textbf{Inset:} The deviation of the gap $\Delta(U)$ from the gap at $U = 0$, i.e. $\delta \Delta(U) \coloneqq \Delta(U) - \Delta(0)$ for the spinless Fermi-Hubbard model at $\beta = 1$. Darker colours correspond to more sites. The apparent clustering of lines towards some ``slope bound'' as they become darker is equivalent to the corresponding line in the main plot (\ref{line:slope_demo}) approaching some upper bound.}
    \label{fig:spinless-fermi-hubbard-slope}
\end{figure}
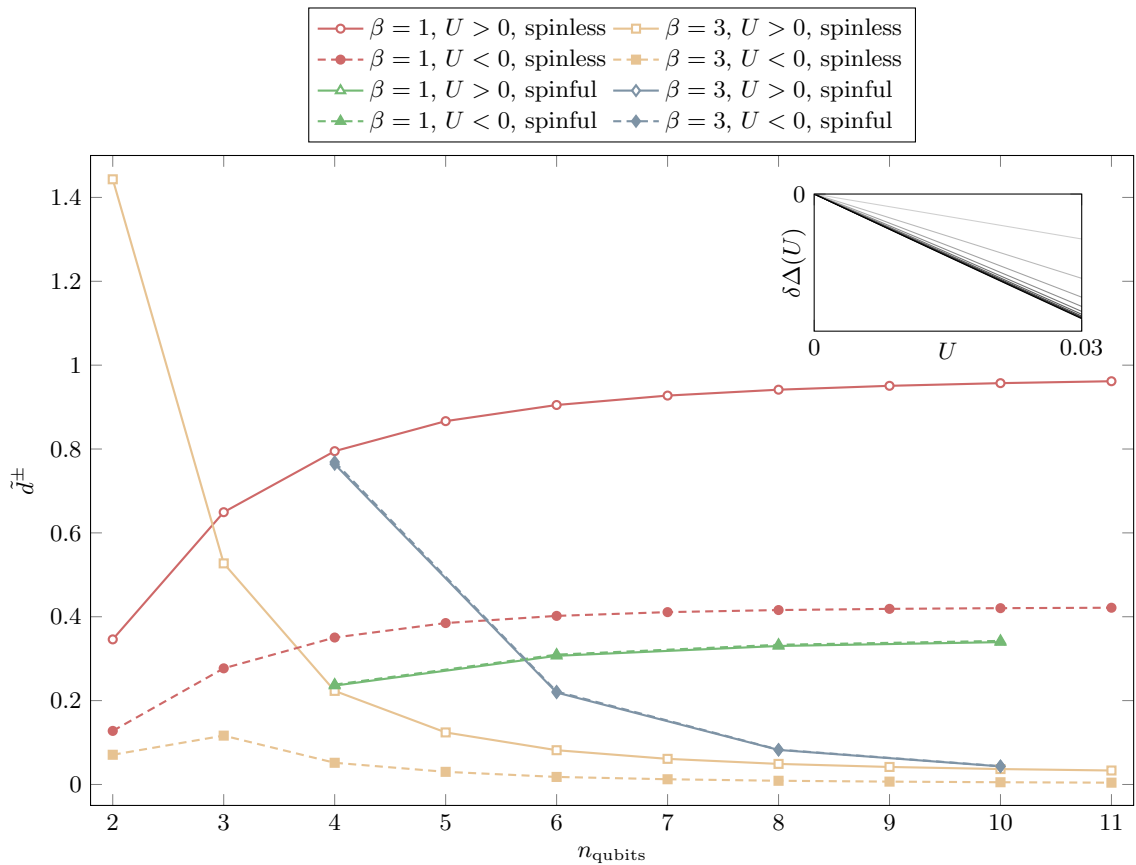


\subsection{Beyond Analytically Bounded Regime}
\label{sec:beyond}

So far\,---\,in the analytical as well as the numerical considerations\,---\,we have only used Gaussian filter functions and Majorana jump operators, as well as small perturbations. In this subsection, we extend our numerical results into regimes where the analytical guarantees of Theorem~\ref{thm: main, gap} may no longer hold. This will give some indication whether it seems reasonable that Gibbs states of the Fermi-Hubbard model can be prepared efficiently using the discussed algorithm. These small-scale insights might also inspire heuristics for which hyperparameters of the algorithm (filter functions, jump operators) can perform well even for larger system sizes. As before, we set $t = 1$ everywhere unless otherwise stated.\\

\paragraph{Strong coupling regime}
First, we maintain the Gaussian filter function and Majorana jump operators, but increase the interaction strength to much higher levels than before. Figure~\ref{fig:strong_interaction} shows the gap $\Delta$ of the Lindbladian for several values of $U$. It seems that with this setup\,---\,at least for the 1D case and the specific temperatures shown\,---\,increasing the interaction strength tends to \emph{shrink} the gap for high temperatures, but \emph{grow} it for lower temperatures.

Overall, it seems that the strong coupling does not change the characteristics of Lindbladian gap for the worse. Even in the cases where the coupling shrinks the gap, the (limited) asymptotic behaviour look very similar, saturating at a comparable rate as the (provably bounded) $U=0$ case. In many other cases, the saturation of the bound seemingly arrives much earlier, and at a larger $\Delta$ than in the non-interacting case. Of course, these small-scale results must be treated with caution and in no way guarantee that the gap will stay open for arbitrarily large systems. But they nonetheless give some confidence that the Gibbs state preparation may also be efficient for larger systems with the investigated setup.

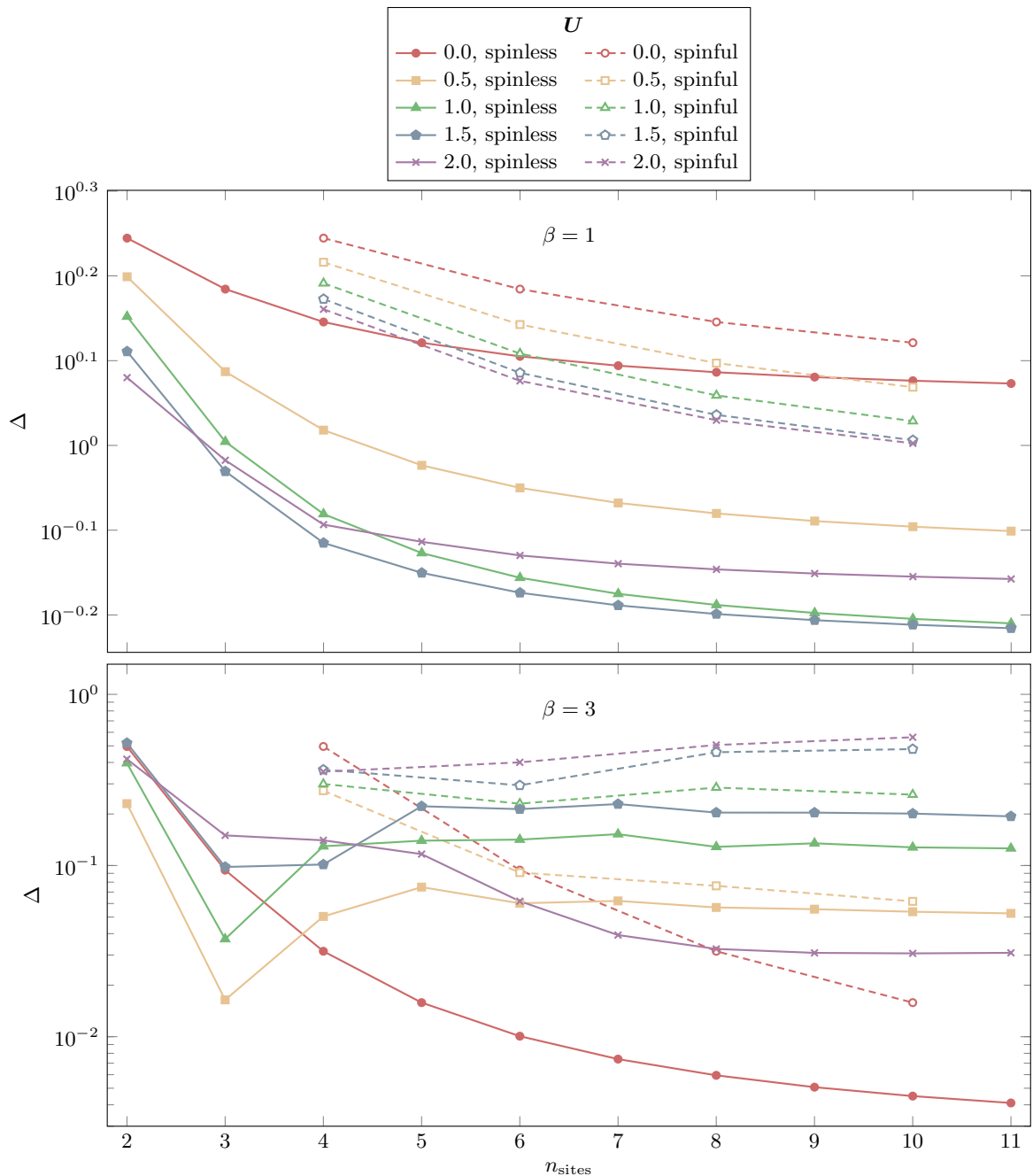
\begin{figure}[tbp]
    \centering
    \begin{tikzpicture}

\definecolor{zero}{RGB}{180, 180, 180}
\definecolor{one}{RGB}{206, 104, 104}
\definecolor{two}{RGB}{231, 195, 146}
\definecolor{three}{RGB}{231, 217, 146}
\definecolor{four}{RGB}{209, 223, 141}
\definecolor{five}{RGB}{116, 185, 116}
\definecolor{six}{RGB}{126, 147, 165}
\definecolor{seven}{RGB}{146, 136, 176}
\definecolor{eight}{RGB}{166, 124, 166}

\begin{groupplot}[
    height=.5\columnwidth,
    width=.9\columnwidth,
    xmin=1.8, xmax=11.2,
    ymode=log,
    xlabel={$n_\mathrm{sites}$}, ylabel={$\Delta$},
    axis on top,
    xtick={2, 3, ..., 11},
    legend cell align={left},
    legend style={
        at={(0.5,1.02)},
        anchor=south,
        cells={anchor=west},
        /tikz/every even column/.append style={column sep=1em},
        /tikz/every odd column/.append style={column sep=.1em}},
    legend columns=2,
    group style={
        group size=1 by 2,
        vertical sep=5pt,
      },
    ]
    \nextgroupplot[ymin=0.57, ymax=2, xlabel=\empty, xticklabels={}]
    \node at (rel axis cs:.5,.9) [anchor=center] {$\beta=1$};
    \addlegendimage{empty legend}
    \addlegendimage{empty legend}
    \addlegendentry{}
    \addlegendentry{\hspace{-5.1em}$\boldsymbol{U}$}
    \addplot[mark=*, mark size=1.5, color=one, thick] table[x=n_qb, y=Delta_beta1_U0_spinless] {data/gap_vs_nqb.dat};
    \addlegendentry{$0.0$, spinless}
    \addplot[mark=*, mark size=1.5, mark options={fill=white,solid}, color=one, thick, densely dashed] table[x=n_qb, y=Delta_beta1_U0_spinful] {data/gap_vs_nqb.dat};
    \addlegendentry{$0.0$, spinful}

    \addplot[mark=square*, mark size=1.5, color=two, thick] table[x=n_qb, y=Delta_beta1_U0.5_spinless] {data/gap_vs_nqb.dat};
    \addlegendentry{$0.5$, spinless}
    \addplot[mark=square*, mark size=1.5, mark options={fill=white,solid}, color=two, thick, densely dashed] table[x=n_qb, y=Delta_beta1_U0.5_spinful] {data/gap_vs_nqb.dat};
    \addlegendentry{$0.5$, spinful}

    \addplot[mark=triangle*, mark size=2, color=five, thick] table[x=n_qb, y=Delta_beta1_U1_spinless] {data/gap_vs_nqb.dat};
    \addlegendentry{$1.0$, spinless}
    \addplot[mark=triangle*, mark size=2, mark options={fill=white,solid}, color=five, thick, densely dashed] table[x=n_qb, y=Delta_beta1_U1_spinful] {data/gap_vs_nqb.dat};
    \addlegendentry{$1.0$, spinful}

    \addplot[mark=pentagon*, mark size=2, color=six, thick] table[x=n_qb, y=Delta_beta1_U1.5_spinless] {data/gap_vs_nqb.dat};
    \addlegendentry{$1.5$, spinless}
    \addplot[mark=pentagon*, mark size=2, mark options={fill=white,solid}, color=six, thick, densely dashed] table[x=n_qb, y=Delta_beta1_U1.5_spinful] {data/gap_vs_nqb.dat};
    \addlegendentry{$1.5$, spinful}

    \addplot[mark=x, mark size=2, color=eight, thick] table[x=n_qb, y=Delta_beta1_U2_spinless] {data/gap_vs_nqb.dat};
    \addlegendentry{$2.0$, spinless}
        \addplot[mark=x, mark size=2, mark options={fill=white,solid}, color=eight, thick, densely dashed] table[x=n_qb, y=Delta_beta1_U2_spinful] {data/gap_vs_nqb.dat};
    \addlegendentry{$2.0$, spinful}

    \nextgroupplot[ymin=0.003, ymax=1.5]
    \node at (rel axis cs:.5,.9) [anchor=center] {$\beta=3$};
    \addplot[mark=*, mark size=1.5, color=one, thick] table[x=n_qb, y=Delta_beta3_U0_spinless] {data/gap_vs_nqb.dat};
    \addplot[mark=*, mark size=1.5, mark options={fill=white,solid}, color=one, thick, densely dashed] table[x=n_qb, y=Delta_beta3_U0_spinful] {data/gap_vs_nqb.dat};

    \addplot[mark=square*, mark size=1.5, color=two, thick] table[x=n_qb, y=Delta_beta3_U0.5_spinless] {data/gap_vs_nqb.dat};
    \addplot[mark=square*, mark size=1.5, mark options={fill=white,solid}, color=two, thick, densely dashed] table[x=n_qb, y=Delta_beta3_U0.5_spinful] {data/gap_vs_nqb.dat};

    \addplot[mark=triangle*, mark size=2, color=five, thick] table[x=n_qb, y=Delta_beta3_U1_spinless] {data/gap_vs_nqb.dat};
    \addplot[mark=triangle*, mark size=2, mark options={fill=white,solid}, color=five, thick, densely dashed] table[x=n_qb, y=Delta_beta3_U1_spinful] {data/gap_vs_nqb.dat};

    \addplot[mark=pentagon*, mark size=2, color=six, thick] table[x=n_qb, y=Delta_beta3_U1.5_spinless] {data/gap_vs_nqb.dat};
    \addplot[mark=pentagon*, mark size=2, mark options={fill=white,solid}, color=six, thick, densely dashed] table[x=n_qb, y=Delta_beta3_U1.5_spinful] {data/gap_vs_nqb.dat};
    \addplot[mark=x, mark size=2, color=eight, thick] table[x=n_qb, y=Delta_beta3_U2_spinless] {data/gap_vs_nqb.dat};
        \addplot[mark=x, mark size=2, mark options={fill=white,solid}, color=eight, thick, densely dashed] table[x=n_qb, y=Delta_beta3_U2_spinful] {data/gap_vs_nqb.dat};
\end{groupplot}

\end{tikzpicture}
    \caption{Gap $\Delta$ of the Lindbladian for the spinless and spinful Fermi-Hubbard model in the regime of strong interactions. The filter function remains Gaussian, the jump operators are single-site Majorana operators. In the spinful model, spin-up and spin-down are counted as separate sites, resulting in only even-numbered $n_\mathrm{sites}$.}
    \label{fig:strong_interaction}
\end{figure}

\paragraph{Metropolis filter function}
One problem that comes with using a Gaussian filter function is apparent in Eqs.~\eqref{eqn: free gap, Gaussian} and \eqref{eq:gauss-delta0}, which is the dependence of the (unperturbed) Lindbladian gap $\Delta_0$ on the inverse temperature $\beta$ as $\Delta_0 \sim e^{-\beta^2}$. This means that while for high temperatures the gap might be substantial, it closes relatively rapidly upon cooling of the model. The authors of Reference \cite{ding2024efficient} also note that using a Gaussian filter function can cause inefficiencies because the size of its support shrinks with $\beta^{-1}$, limiting the available transitions between energy eigenstates.

It therefore seems reasonable to investigate the behaviour of the gap $\Delta$ when using a different kind of filter function. We opt for a Metropolis-type filter mentioned in Eq.~\eqref{eq:metropolis}, which is
\begin{equation}
    \hat{f}^a(\nu) = e^{-\sqrt{1 + \beta^2 \nu^2}} \, w(\nu / S) \, e^{-\beta\nu/4}
\end{equation}
where we now use $S = 10$.

The gap of the unperturbed Lindbladian is then 
\begin{align}
        \Delta_0 &= 2 \cdot \min_i q(4\epsilon_i)^2 \cosh(2\beta \epsilon_i)\nonumber\\
        &= 2 \cdot \min_i e^{-2\sqrt{1 + 16 \beta^2 \epsilon_i^2}}  w(4\epsilon_i / S)^2 \cosh(2\beta \epsilon_i)\nonumber\\
         &= 2 e^{-2\sqrt{1 + 16 \beta^2 \|h\|^2}}  w(4\|h\| / S)^2 \cosh(2\beta \|h\|)\,,
\end{align}
which does not close as long as $S > 4\|h\|$ and $\|h\| = \mathcal{O}(1)$. In the case of the 1D spinless chain \eqref{eqn: free 1D Hamiltonian}, we have that $\|h\| = \frac{t}{2}\cos\left(\frac{\pi}{n+1}\right)$, and so the size of the gap
\begin{align}\label{eq:metropolis-unperturbed-gap}
        \Delta_0 \geq 2 e^{-2\sqrt{1 + 4 \beta^2 t^2}}  w(2t / S)^2 \cosh(\beta t)
\end{align}
scales better with $\beta$ than the Gaussian filter function.

Figure~\ref{fig:metropolis} shows the gap in the same setup as Fig.~\ref{fig:strong_interaction}, but using a Metropolis-type filter instead of the Gaussian filter function. Qualitatively, Fig.~\ref{fig:metropolis} has similar features to Fig.~\ref{fig:strong_interaction}, and much of the comments from above still hold. However, notice the relative difference of $\Delta$ between the $\beta=1$ and $\beta=3$ plots. With a Gaussian filter function, already at $\beta=3$, the magnitude of the gap shrinks quite considerably versus the $\beta=1$ case. The Metropolis-type filter causes similar behaviour, but to a much lesser degree, as could be expected from the analysis of $\Delta_0$.

\begin{figure}[tbp]
    \centering
    \begin{tikzpicture}

\definecolor{zero}{RGB}{180, 180, 180}
\definecolor{one}{RGB}{206, 104, 104}
\definecolor{two}{RGB}{231, 195, 146}
\definecolor{three}{RGB}{231, 217, 146}
\definecolor{four}{RGB}{209, 223, 141}
\definecolor{five}{RGB}{116, 185, 116}
\definecolor{six}{RGB}{126, 147, 165}
\definecolor{seven}{RGB}{146, 136, 176}
\definecolor{eight}{RGB}{166, 124, 166}

\begin{groupplot}[
    height=.57\columnwidth,
    width=.9\columnwidth,
    xmin=1.8, xmax=11.2,
    ymode=log,
    xlabel={$n_\mathrm{sites}$}, ylabel={$\Delta$},
    axis on top,
    xtick={2, 3, ..., 11},
    legend cell align={left},
    legend style={
        at={(0.5,1.02)},
        anchor=south,
        cells={anchor=west},
        /tikz/every even column/.append style={column sep=1em},
        /tikz/every odd column/.append style={column sep=.1em}},
    legend columns=2,
    group style={
        group size=1 by 2,
        vertical sep=5pt,
      },
    ]
    \nextgroupplot[
      ymin=0.645, ymax=0.788,
      xlabel=\empty,
      xticklabels={}
    ]
    \node at (rel axis cs:.5,.9) [anchor=center] {$\beta=1$};
    \addlegendimage{empty legend}
    \addlegendimage{empty legend}
    \addlegendentry{}
    \addlegendentry{\hspace{-5.1em}$\boldsymbol{U}$}
    \addplot[mark=*, mark size=1.5, color=one, thick] table[x=n_qb, y=Delta_beta1_U0_spinless] {data/gap_vs_nqb_metropolis.dat};
    \addlegendentry{$0.0$, spinless}
    \addplot[mark=*, mark size=1.5, mark options={fill=white,solid}, color=one, thick, densely dashed] table[x=n_qb, y=Delta_beta1_U0_spinful] {data/gap_vs_nqb_metropolis.dat};
    \addlegendentry{$0.0$, spinful}

    \addplot[mark=square*, mark size=1.5, color=two, thick] table[x=n_qb, y=Delta_beta1_U0.5_spinless] {data/gap_vs_nqb_metropolis.dat};
    \addlegendentry{$0.5$, spinless}
    \addplot[mark=square*, mark size=1.5, mark options={fill=white,solid}, color=two, thick, densely dashed] table[x=n_qb, y=Delta_beta1_U0.5_spinful] {data/gap_vs_nqb_metropolis.dat};
    \addlegendentry{$0.5$, spinful}

    \addplot[mark=triangle*, mark size=2, color=five, thick] table[x=n_qb, y=Delta_beta1_U1_spinless] {data/gap_vs_nqb_metropolis.dat};
    \addlegendentry{$1.0$, spinless}
    \addplot[mark=triangle*, mark size=2, mark options={fill=white,solid}, color=five, thick, densely dashed] table[x=n_qb, y=Delta_beta1_U1_spinful] {data/gap_vs_nqb_metropolis.dat};
    \addlegendentry{$1.0$, spinful}

    \addplot[mark=pentagon*, mark size=2, color=six, thick] table[x=n_qb, y=Delta_beta1_U1.5_spinless] {data/gap_vs_nqb_metropolis.dat};
    \addlegendentry{$1.5$, spinless}
    \addplot[mark=pentagon*, mark size=2, mark options={fill=white,solid}, color=six, thick, densely dashed] table[x=n_qb, y=Delta_beta1_U1.5_spinful] {data/gap_vs_nqb_metropolis.dat};
    \addlegendentry{$1.5$, spinful}

    \addplot[mark=x, mark size=2, color=eight, thick] table[x=n_qb, y=Delta_beta1_U2_spinless] {data/gap_vs_nqb_metropolis.dat};
    \addlegendentry{$2.0$, spinless}
        \addplot[mark=x, mark size=2, mark options={fill=white,solid}, color=eight, thick, densely dashed] table[x=n_qb, y=Delta_beta1_U2_spinful] {data/gap_vs_nqb_metropolis.dat};
    \addlegendentry{$2.0$, spinful}

    \nextgroupplot[
      ymin=0.62, ymax=0.86,
    ]
    \node at (rel axis cs:.5,.9) [anchor=center] {$\beta=3$};
    \addplot[mark=*, mark size=1.5, color=one, thick] table[x=n_qb, y=Delta_beta3_U0_spinless] {data/gap_vs_nqb_metropolis.dat};
    \addplot[mark=*, mark size=1.5, mark options={fill=white,solid}, color=one, thick, densely dashed] table[x=n_qb, y=Delta_beta3_U0_spinful] {data/gap_vs_nqb_metropolis.dat};

    \addplot[mark=square*, mark size=1.5, color=two, thick] table[x=n_qb, y=Delta_beta3_U0.5_spinless] {data/gap_vs_nqb_metropolis.dat};
    \addplot[mark=square*, mark size=1.5, mark options={fill=white,solid}, color=two, thick, densely dashed] table[x=n_qb, y=Delta_beta3_U0.5_spinful] {data/gap_vs_nqb_metropolis.dat};

    \addplot[mark=triangle*, mark size=2, color=five, thick] table[x=n_qb, y=Delta_beta3_U1_spinless] {data/gap_vs_nqb_metropolis.dat};
    \addplot[mark=triangle*, mark size=2, mark options={fill=white,solid}, color=five, thick, densely dashed] table[x=n_qb, y=Delta_beta3_U1_spinful] {data/gap_vs_nqb_metropolis.dat};

    \addplot[mark=pentagon*, mark size=2, color=six, thick] table[x=n_qb, y=Delta_beta3_U1.5_spinless] {data/gap_vs_nqb_metropolis.dat};
    \addplot[mark=pentagon*, mark size=2, mark options={fill=white,solid}, color=six, thick, densely dashed] table[x=n_qb, y=Delta_beta3_U1.5_spinful] {data/gap_vs_nqb_metropolis.dat};
    \addplot[mark=x, mark size=2, color=eight, thick] table[x=n_qb, y=Delta_beta3_U2_spinless] {data/gap_vs_nqb_metropolis.dat};
        \addplot[mark=x, mark size=2, mark options={fill=white,solid}, color=eight, thick, densely dashed] table[x=n_qb, y=Delta_beta3_U2_spinful] {data/gap_vs_nqb_metropolis.dat};
\end{groupplot}

\end{tikzpicture}
    \caption{Gap $\Delta$ of the Lindbladian for the spinless and spinful Fermi-Hubbard model at different interaction strengths $U$. In the spinful model, spin-up and spin-down are counted as separate sites, resulting in only even-numbered $n_\mathrm{sites}$. The setup is identical to that in Figure \ref{fig:strong_interaction}, but the filter function is now of Metropolis-type, see Eq.~\eqref{eq:metropolis}.}
    \label{fig:metropolis}
\end{figure}
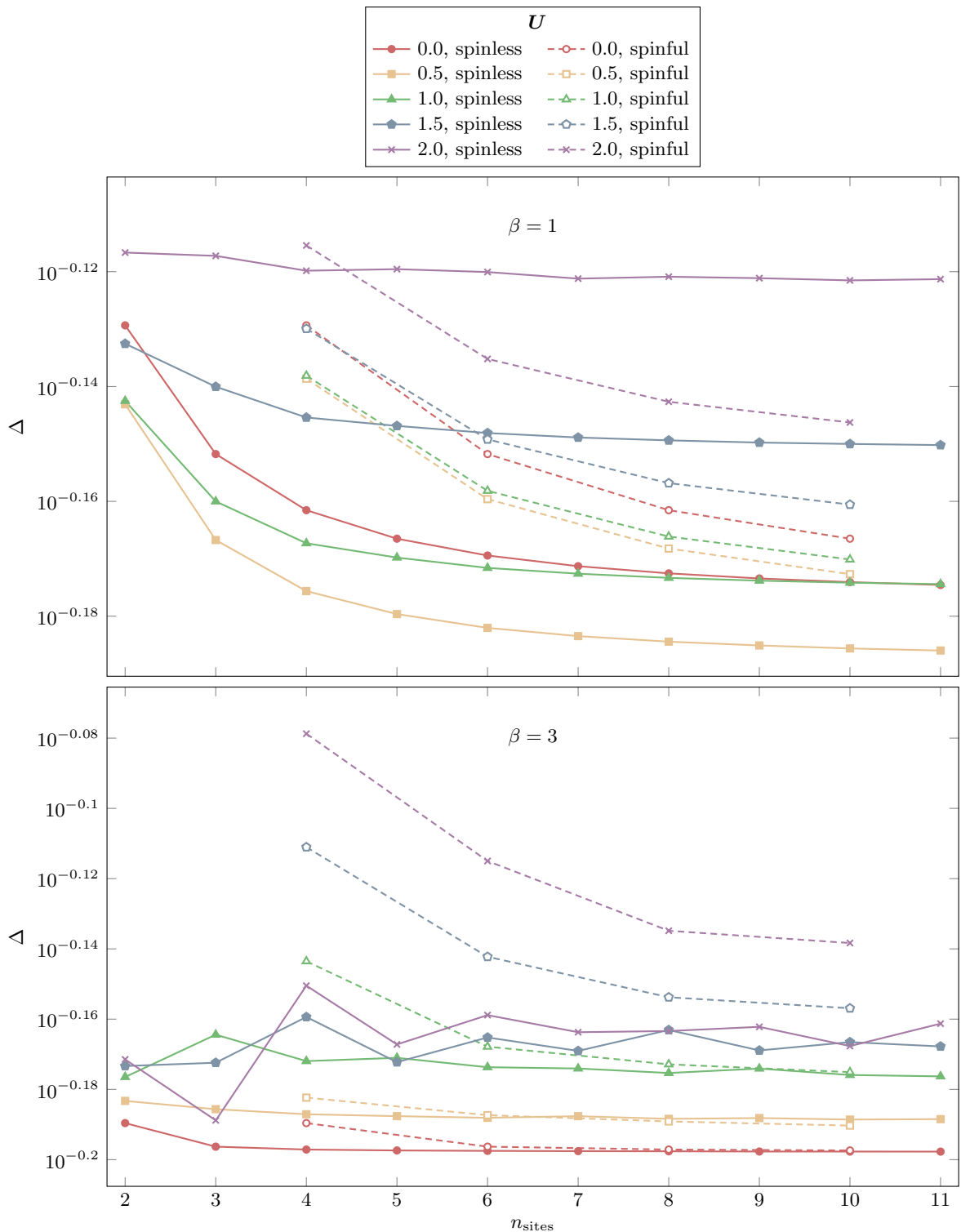

\paragraph{Pauli jump operators and Metropolis filter function}
The last modification we look at is to use different jump operators. So far, we described all operators in terms of fermionic creation and annihilation operators, and the results are therefore independent of how the mapping to qubits is performed. In this part, however, we will first map the fermionic system onto qubit operators using the Jordan-Wigner (JW) transformation~(see Eq.~\eqref{eq:jordan-wigner}). Performing the JW transformation manually for the spinless Fermi-Hubbard model, its Hamiltonian becomes
\begin{align}
    \label{eq:hubbard_1d}
    H &= -t \sum_{i=1}^{n-1} (c_i ^ \dagger c_{i+1} + c_{i+1}^\dagger c_i) + U \sum_{i=1}^{n-1} N_i N_{i+1}\nonumber\\
    &\equiv \sum_{i=1}^{n-1} -\frac{t}{2}(X_i X_{i+1} + Y_i Y_{i+1}) + \frac{U}{4} (Z_i Z_{i+1} + Z_{i+1} + Z_i + I).
\end{align}
In this transformed system, we then use single-site Pauli operators as the jump operators in the Gibbs state preparation algorithm. The filter function remains of Metropolis-type as in the previous paragraph. Figure~\ref{fig:pauli_jump} shows the dependence of the Lindbladian gap on the system size. While this setup is not supported by any of the theorems presented earlier in this work, it yields some interesting behaviour that is worth commenting on. We thus continue with some observations and speculative remarks regarding these results:

\begin{itemize}
\item \emph{Size convergence of spinless model.} At the higher temperature of $\beta = 1$, the spinless variant of the model yields very smooth curves that\,---\,as in previous setups\,---\,show a gap that shrinks when increasing the system size. This behaviour, though, seems to approach some saturation relatively quickly, comparable to the setup with Majorana jump operators. Importantly, a key distinction is that while for the setup with Majorana operators we can prove that the gap for $U=0$ is lower bounded for system sizes $n_\mathrm{sites}\rightarrow\infty$, the same cannot be said for the setup with Pauli operators. However, the results in Fig. \ref{fig:pauli_jump} give some confidence that the asymptotic behaviour might be similar. The lower temperature $\beta = 3$ for the spinless case shows qualitatively similar behaviour, even though the data is not quite as smooth.

\item \emph{Size convergence of spinful model.} When including spin, the data for both considered temperatures become much flatter, staying almost constant at $\beta = 1$ and seemingly fluctuating around a constant for $\beta = 3$. This could hint at even better scaling behaviour than the spinless case, even though four data points might not be enough to draw any strong conclusions.

\item \emph{Temperature stability.} A quite significant difference when using Pauli operators is the temperature dependence of the gap. Recall that with Majorana operators, at $U=0$ a Gaussian filter function yields $\Delta_0 \sim e^{-\beta^2}$ (Eq.~\eqref{eq:gauss-delta0}), and a Metropolis-type filter gives $\Delta_0 \sim e^{-\beta}$ (Eq.~\eqref{eq:metropolis-unperturbed-gap}). Conversely, our numerics suggest that using Pauli jump operators, the temperature dependence of the gap size is quite strongly suppressed. To illustrate this point further, Fig.~\ref{fig:usweep} shows the size of the Lindbladian gap depending on the interaction strength $U$ at different temperatures $\beta = 1, 5, 25$. At high temperatures, the dependence of $\Delta$ on $U$ is quite smooth, and the system size has relatively little influence on it. As the system is cooled down, the overall shape remains roughly the same, but much more structure with rapid oscillations emerges. Crucially, however, the magnitude of the gap seems quite unaffected by the temperature.

\item \emph{Dependence of $\Delta$ on $U$.} Figure~\ref{fig:usweep} also gives some confidence in regards to the stability of the gap across a wide range of interaction strengths $-5 \lesssim U \lesssim 10$, complementing the analytical result that holds in the region around $U\approx 0$. Notice that there is a sharp drop of the gap $\Delta$ to $0$ as $|U|$ approaches the limit of the support $S$ of the filter function (recall that $S=10$ for these simulations). However, to facilitate large $U$, increasing the size of the support only incurs an overhead of the algorithm that is polylogarithmic in $S$, as per \cite[Theorem 34]{ding2024efficient}.
\end{itemize}
\enlargethispage{0.2cm}

\begin{figure}[H]
    \centering
    \begin{tikzpicture}

\definecolor{zero}{RGB}{180, 180, 180}
\definecolor{one}{RGB}{206, 104, 104}
\definecolor{two}{RGB}{231, 195, 146}
\definecolor{three}{RGB}{231, 217, 146}
\definecolor{four}{RGB}{209, 223, 141}
\definecolor{five}{RGB}{116, 185, 116}
\definecolor{six}{RGB}{126, 147, 165}
\definecolor{seven}{RGB}{146, 136, 176}
\definecolor{eight}{RGB}{166, 124, 166}

\begin{groupplot}[
    height=.6\columnwidth,
    width=.9\columnwidth,
    xmin=1.8, xmax=11.2,
    ymode=log,
    xlabel={$n_\mathrm{sites}$}, ylabel={$\Delta$},
    axis on top,
    xtick={2, 3, ..., 11},
    legend cell align={left},
    legend style={
        at={(0.5,1.02)},
        anchor=south,
        cells={anchor=west},
        /tikz/every even column/.append style={column sep=1em},
        /tikz/every odd column/.append style={column sep=.1em}},
    legend columns=2,
    group style={
        group size=1 by 2,
        vertical sep=5pt,
      },
    ]
    \nextgroupplot[
      xlabel=\empty,
      xticklabels={}
    ]
    \node at (rel axis cs:.5,.9) [anchor=center] {$\beta=1$};
    \addlegendimage{empty legend}
    \addlegendimage{empty legend}
    \addlegendentry{}
    \addlegendentry{\hspace{-5.1em}$\boldsymbol{U}$}
    \addplot[mark=*, mark size=1.5, color=one, thick] table[x=n_qb, y=Delta_beta1_U0_spinless] {data/gap_vs_nqb_paulis_metropolis.dat};
    \addlegendentry{$0.0$, spinless}
    \addplot[mark=*, mark size=1.5, mark options={fill=white,solid}, color=one, thick, densely dashed] table[x=n_qb, y=Delta_beta1_U0_spinful] {data/gap_vs_nqb_paulis_metropolis.dat};
    \addlegendentry{$0.0$, spinful}

    \addplot[mark=square*, mark size=1.5, color=two, thick] table[x=n_qb, y=Delta_beta1_U0.5_spinless] {data/gap_vs_nqb_paulis_metropolis.dat};
    \addlegendentry{$0.5$, spinless}
    \addplot[mark=square*, mark size=1.5, mark options={fill=white,solid}, color=two, thick, densely dashed] table[x=n_qb, y=Delta_beta1_U0.5_spinful] {data/gap_vs_nqb_paulis_metropolis.dat};
    \addlegendentry{$0.5$, spinful}

    \addplot[mark=triangle*, mark size=2, color=five, thick] table[x=n_qb, y=Delta_beta1_U1_spinless] {data/gap_vs_nqb_paulis_metropolis.dat};
    \addlegendentry{$1.0$, spinless}
    \addplot[mark=triangle*, mark size=2, mark options={fill=white,solid}, color=five, thick, densely dashed] table[x=n_qb, y=Delta_beta1_U1_spinful] {data/gap_vs_nqb_paulis_metropolis.dat};
    \addlegendentry{$1.0$, spinful}

    \addplot[mark=pentagon*, mark size=2, color=six, thick] table[x=n_qb, y=Delta_beta1_U1.5_spinless] {data/gap_vs_nqb_paulis_metropolis.dat};
    \addlegendentry{$1.5$, spinless}
    \addplot[mark=pentagon*, mark size=2, mark options={fill=white,solid}, color=six, thick, densely dashed] table[x=n_qb, y=Delta_beta1_U1.5_spinful] {data/gap_vs_nqb_paulis_metropolis.dat};
    \addlegendentry{$1.5$, spinful}

    \addplot[mark=x, mark size=2, color=eight, thick] table[x=n_qb, y=Delta_beta1_U2_spinless] {data/gap_vs_nqb_paulis_metropolis.dat};
    \addlegendentry{$2.0$, spinless}
        \addplot[mark=x, mark size=2, mark options={fill=white,solid}, color=eight, thick, densely dashed] table[x=n_qb, y=Delta_beta1_U2_spinful] {data/gap_vs_nqb_paulis_metropolis.dat};
    \addlegendentry{$2.0$, spinful}

    \nextgroupplot[
    ]
    \node at (rel axis cs:.5,.9) [anchor=center] {$\beta=3$};
    \addplot[mark=*, mark size=1.5, color=one, thick] table[x=n_qb, y=Delta_beta3_U0_spinless] {data/gap_vs_nqb_paulis_metropolis.dat};
    \addplot[mark=*, mark size=1.5, mark options={fill=white,solid}, color=one, thick, densely dashed] table[x=n_qb, y=Delta_beta3_U0_spinful] {data/gap_vs_nqb_paulis_metropolis.dat};

    \addplot[mark=square*, mark size=1.5, color=two, thick] table[x=n_qb, y=Delta_beta3_U0.5_spinless] {data/gap_vs_nqb_paulis_metropolis.dat};
    \addplot[mark=square*, mark size=1.5, mark options={fill=white,solid}, color=two, thick, densely dashed] table[x=n_qb, y=Delta_beta3_U0.5_spinful] {data/gap_vs_nqb_paulis_metropolis.dat};

    \addplot[mark=triangle*, mark size=2, color=five, thick] table[x=n_qb, y=Delta_beta3_U1_spinless] {data/gap_vs_nqb_paulis_metropolis.dat};
    \addplot[mark=triangle*, mark size=2, mark options={fill=white,solid}, color=five, thick, densely dashed] table[x=n_qb, y=Delta_beta3_U1_spinful] {data/gap_vs_nqb_paulis_metropolis.dat};

    \addplot[mark=pentagon*, mark size=2, color=six, thick] table[x=n_qb, y=Delta_beta3_U1.5_spinless] {data/gap_vs_nqb_paulis_metropolis.dat};
    \addplot[mark=pentagon*, mark size=2, mark options={fill=white,solid}, color=six, thick, densely dashed] table[x=n_qb, y=Delta_beta3_U1.5_spinful] {data/gap_vs_nqb_paulis_metropolis.dat};
    \addplot[mark=x, mark size=2, color=eight, thick] table[x=n_qb, y=Delta_beta3_U2_spinless] {data/gap_vs_nqb_paulis_metropolis.dat};
        \addplot[mark=x, mark size=2, mark options={fill=white,solid}, color=eight, thick, densely dashed] table[x=n_qb, y=Delta_beta3_U2_spinful] {data/gap_vs_nqb_paulis_metropolis.dat};
\end{groupplot}

\end{tikzpicture}
    \caption{Gap $\Delta$ of the Lindbladian for the spinless and spinful Fermi-Hubbard model at different interaction strengths $U$. In the spinful model, spin-up and spin-down are counted as separate sites, resulting in only even-numbered $n_\mathrm{sites}$. The setup is identical to that in \ref{fig:metropolis}, but the jump operators are now single-site Pauli operators.}
    \label{fig:pauli_jump}
\end{figure}
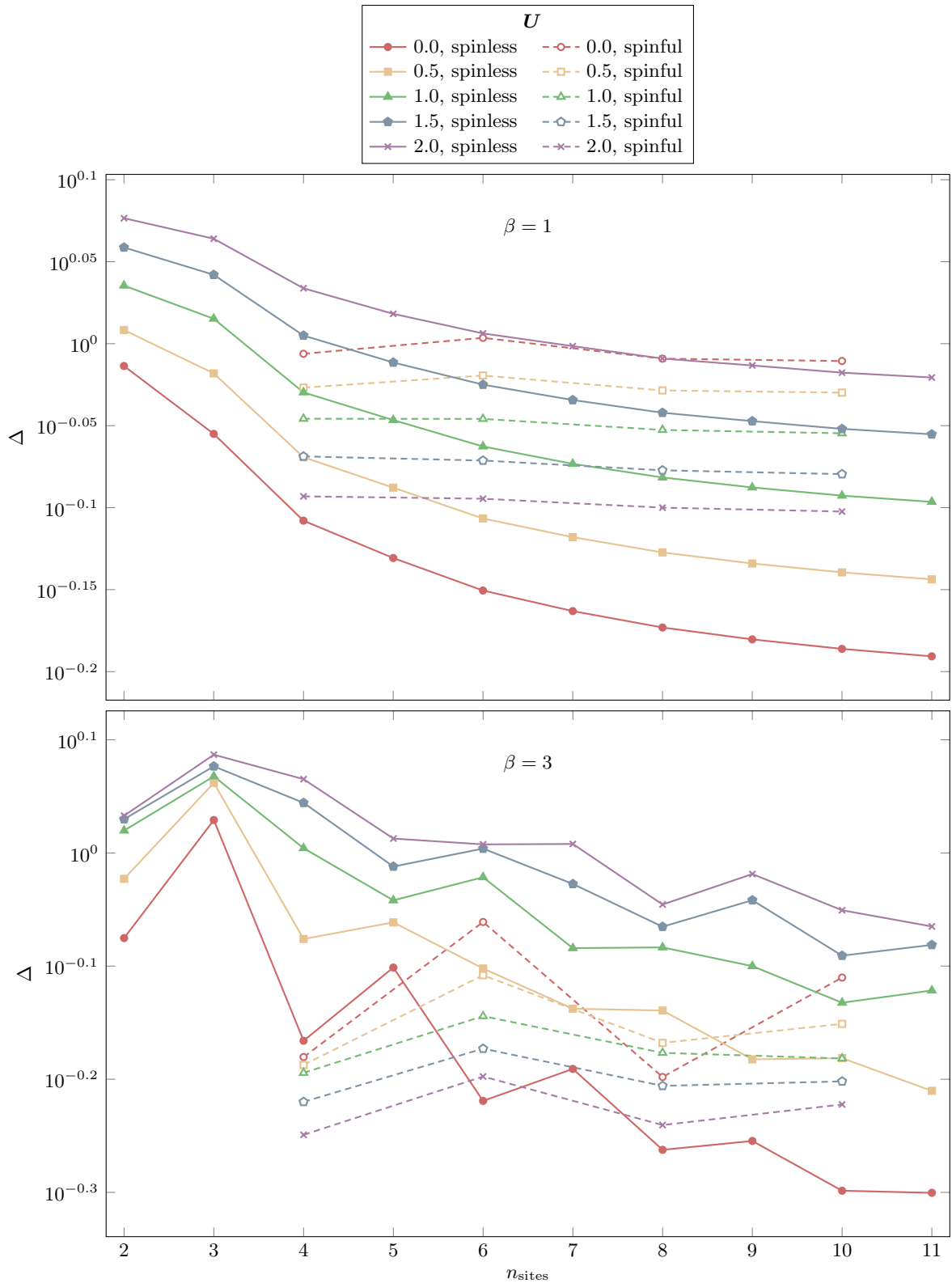

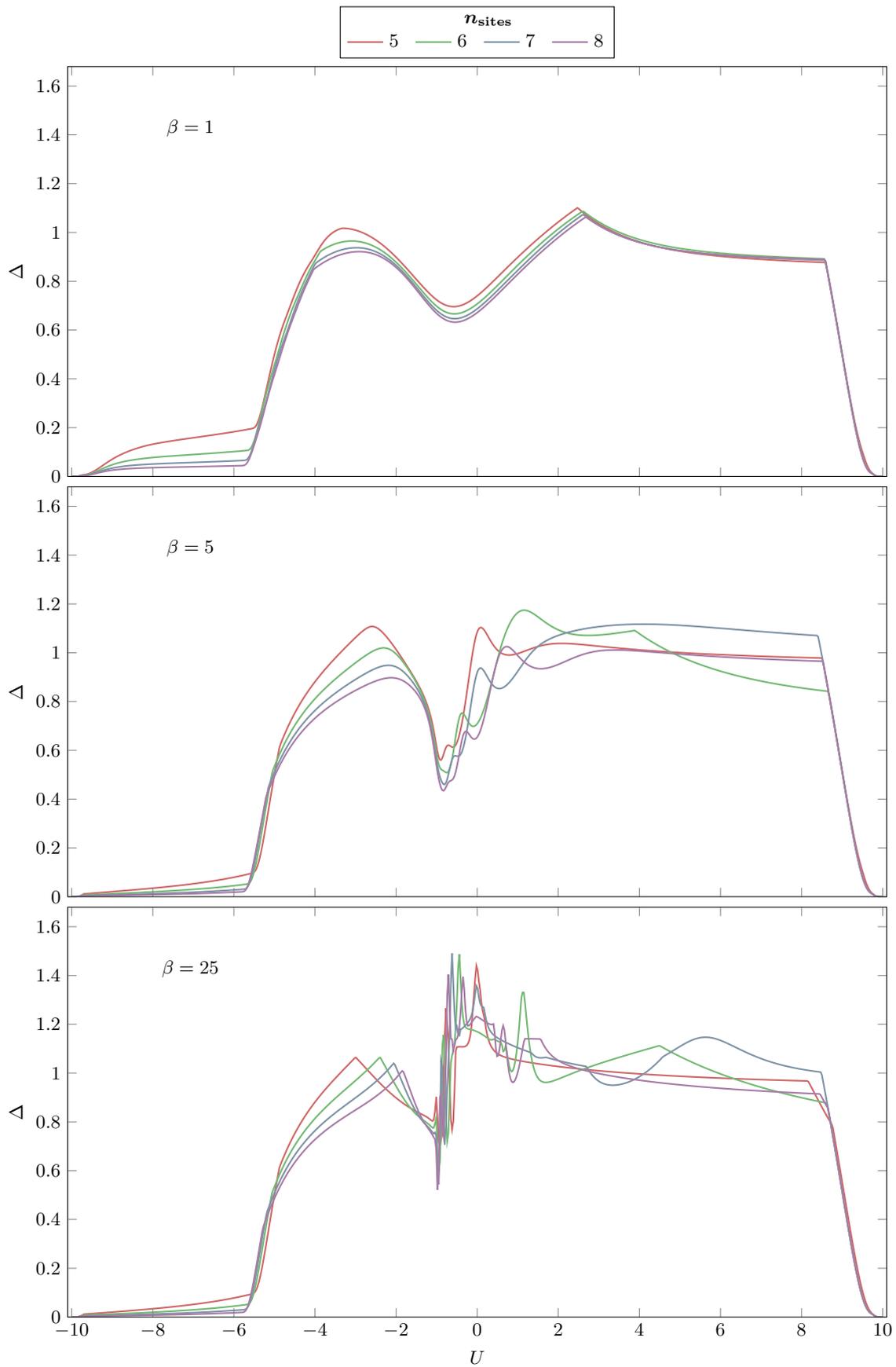
\begin{figure}[p]
    \centering
    \begin{tikzpicture}

\definecolor{zero}{RGB}{180, 180, 180}
\definecolor{one}{RGB}{206, 104, 104}
\definecolor{two}{RGB}{231, 195, 146}
\definecolor{three}{RGB}{231, 217, 146}
\definecolor{four}{RGB}{209, 223, 141}
\definecolor{five}{RGB}{116, 185, 116}
\definecolor{six}{RGB}{126, 147, 165}
\definecolor{seven}{RGB}{146, 136, 176}
\definecolor{eight}{RGB}{166, 124, 166}

\begin{groupplot}[
    height=.5\columnwidth,
    width=.9\columnwidth,
    xmin=-10.1, xmax=10.1,
    ymin=0, ymax=1.68,
    xlabel={$U$}, ylabel={$\Delta$},
    axis on top,
    legend cell align={left},
    legend style={
        at={(0.5,1.02)},
        anchor=south,
        cells={anchor=west},
        /tikz/every even column/.append style={column sep=.6em},
        /tikz/every odd column/.append style={column sep=.05em}},
    legend columns=4,
        group style={
        group size=1 by 3,
        vertical sep=5pt,
      },
    ]
    \nextgroupplot[
      xlabel=\empty,
      xticklabels={}
    ]
    \node at (rel axis cs:.15,.85) [anchor=center] {$\beta=1$};
    \addlegendimage{empty legend}
    \addlegendimage{empty legend}
    \addlegendimage{empty legend}
    \addlegendimage{empty legend}
    \addlegendentry{}
    \addlegendentry{}
    \addlegendentry{\hspace{-3.5em}$\boldsymbol{n_\mathrm{sites}}$}
    \addlegendentry{}
    \addplot[mark=none, color=one, thick] table[x=U, y=Delta_nqb5_beta1] {data/usweep.dat};
    \addlegendentry{$5$}
    \addplot[mark=none, color=five, thick] table[x=U, y=Delta_nqb6_beta1] {data/usweep.dat};
    \addlegendentry{$6$}
    \addplot[mark=none, color=six, thick] table[x=U, y=Delta_nqb7_beta1] {data/usweep.dat};
    \addlegendentry{$7$}
    \addplot[mark=none, color=eight, thick] table[x=U, y=Delta_nqb8_beta1] {data/usweep.dat};
    \addlegendentry{$8$}
    
    \nextgroupplot[
      xlabel=\empty,
      xticklabels={}
    ]
    \node at (rel axis cs:.15,.85) [anchor=center] {$\beta=5$};
    \addplot[mark=none, color=one, thick] table[x=U, y=Delta_nqb5_beta5] {data/usweep.dat};
    \addplot[mark=none, color=five, thick] table[x=U, y=Delta_nqb6_beta5] {data/usweep.dat};
    \addplot[mark=none, color=six, thick] table[x=U, y=Delta_nqb7_beta5] {data/usweep.dat};
    \addplot[mark=none, color=eight, thick] table[x=U, y=Delta_nqb8_beta5] {data/usweep.dat};
    
    \nextgroupplot
    \node at (rel axis cs:.15,.85) [anchor=center] {$\beta=25$};
    \addplot[mark=none, color=one, thick] table[x=U, y=Delta_nqb5_beta25] {data/usweep.dat};
    \addplot[mark=none, color=five, thick] table[x=U, y=Delta_nqb6_beta25] {data/usweep.dat};
    \addplot[mark=none, color=six, thick] table[x=U, y=Delta_nqb7_beta25] {data/usweep.dat};
    \addplot[mark=none, color=eight, thick] table[x=U, y=Delta_nqb8_beta25] {data/usweep.dat};
    
\end{groupplot}

\end{tikzpicture}
    \caption{Lindbladian gap $\Delta$ as a function of the interaction strength $U$ in the spinless Fermi-Hubbard model at different inverse temperatures $\beta$ when using a Metropolis-type filter function with a support of $S=10$ and single-site Pauli operators as jump operators. Observe that the magnitude of the gap does not decay w.r.t. $\beta$, and that it closes completely for $|U| \gtrsim S$.}
    \label{fig:usweep}
\end{figure}


\clearpage

\section{Outlook}
\label{sec:outlook}

We have shown that the Gibbs states of weakly interacting fermions in fixed dimension and at any constant temperature can be efficiently created in polynomial time on quantum computers. It would be great to potentially improve the exact dependencies on all the relevant parameters in our main result \eqref{eq:main-result}, as well as to work out all the hidden constants therein. One might then also fine-tune the choice of jump operators and filter functions for specific systems and parameter regimes, potentially even improving on the current cubic complexity in system size in \eqref{eq:main-result} all the way down to quasi-linear. The rapid mixing of free fermions serves as a good sign for the weakly interacting case to also mix rapidly, which would bring down one factor of $n$. As we have seen, the cubic dependency on the system size is then mainly due to the quadratic dependence on the number of jump operators taken, which needs to be linear so that the quantum Markov semigroup can be irreducible and ergodic. However, we could potentially use different sets of jump operators at different times, consisting of as few as a single jump operator at any given time. The generated dynamics could then not be described by a single quantum Markov semigroup, and any one of the corresponding Lindbladians would not be able to get us from an arbitrary starting position to the Gibbs state; however, each one could get us slightly closer, and the combination of all of them could create a complicated path in the state space eventually getting us to the Gibbs state, with potentially significantly better dependency on the system size. Though such dynamics would be very challenging to analyse.

Furthermore, in order to determine the classical-quantum efficiency boundary for the Fermi-Hubbard model in the absence of reasonable quantum computers, it would be important to perform larger-scale classical simulations of the quantum Gibbs samplers, with varying design parameters to estimate the relevant spectral gap (e.g., based on tensor network methods \cite{biamonte2017tensornutshell,mortier2024tensornetworks}). Especially for translationally invariant systems in $D=1$, for which the parent Hamiltonian would also inherit the invariance, one could use the imaginary-time iTEBD algorithm to simulate the evolution or iDMRG to calculate the spectral gap for infinite-sized systems. Any such findings should then be compared to the state-of-the-art classical results \cite{Arovas_2022,Qin2022hubbardmodel} to make statements about potentially practical quantum advantages. In turn, this will also require to complete the quantum complexity analysis to an end-to-end fashion \cite[Section 1.1]{dalzell2023algorithmsreview}, where the exact properties of interest are sampled from the prepared quantum Gibbs state \cite[Section 5]{baertschi2024usecaseslosalamos} (such as, e.g., correlations functions). This will lead to a significant overhead cost; for example, the approximation error for estimating any observable will scale at the very least as $\mathcal{O}(1/\epsilon)$\,---\,which is on top of the $\mathcal{O}(\text{polylog}(1/\epsilon))$ scaling of the quantum Gibbs state preparation itself. Ultimately and from an application perspective, one can then compare what are the computational costs of resolving the phase diagram of the Fermi-Hubbard model at different interaction strengths.

Finally, our presented proof methods based on eigenvalue perturbation techniques also seem promising to explore other quantum many-body systems in different regimes. As a first step, the extensions presented in Section \ref{sec:stability atomic} are easily shown to hold for any Hamiltonians that are separable in the lattice sites.


\begin{acknowledgments}
During the write-up of our manuscript, we have become aware that Yu Tong and Yongtao Zhan were independently working on the same topic. We thank them for agreeing on a date for uploading to the arXiv together. MB and RM acknowledge support from the EPSRC Grant number EP/W032643/1. MB acknowledges funding by the European Research Council (ERC Grant Agreement No.~948139) and the Excellence Cluster Matter and Light for Quantum Computing (ML4Q). The authors acknowledge the use of the Imperial College Research Computing Service~\cite{rcs} to obtain some of the results in this work.
\end{acknowledgments}



\bibliographystyle{alphaurl}
\bibliography{sample}


\appendix

\section{Useful Lemmas}
\label{app:lemmas}

\begin{lemma}\label{lemma:bound on perturbed evolution}
For any operator $O$, Hermitian operators $H_0, V$ and  $\lambda,\alpha\in\mathbb{C}$, we have
\begin{align}
    \|e^{\alpha (H_0+\lambda V)}Oe^{-\alpha (H_0+\lambda V)} 
    -
    e^{\alpha H_0}Oe^{-\alpha H_0} 
    \|
    \le
    |\lambda|
    \,|\alpha|\,
    \max_{s\in [0,1]}
    \|
    [V, e^{s\alpha H_0}Oe^{-s\alpha H_0}]
    \|
    \,.
\end{align}
\end{lemma}

\begin{proof}
We first recall Duhamel's formula. For any operators $A,B$:
\begin{align}
    e^{(A+B)t}
    =
    e^{At}
    +
    \int_0^t 
    e^{(A+B)(t-s)}Be^{As}\ \dd s
\end{align}
The proof is simple: if we call $C(t)$ the right hand side, we have that 
\begin{align}
C'(t)=
Ae^{At}
+
Be^{At}
+
\int_0^t
(A+B)    e^{(A+B)(t-s)}Be^{As}\ \dd s
=
(A+B)C(t)
\,.
\end{align}
Solving this differential equation together with $C(0)=\id$ yields the formula.
Then denote 
\begin{align}
    \mathcal{A}
    =
    \mathrm{ad}_{H_0}
    \,,\quad 
    \mathcal{B}
    =
    \mathrm{ad}_V
    \,,
\end{align}
where $\mathrm{ad}_HA=[H,A]$.
Let us also denote the formula we want to study by $f(\lambda,\alpha)$, so that we have
\begin{align}
    f(\lambda,\alpha)
    :=
    \|e^{\alpha(H_0+\lambda V)}Oe^{-\alpha(H_0+\lambda V)} 
    -
    e^{\alpha H_0}Oe^{-\alpha H_0} 
    \|
    =
    \|
    (
    e^{\alpha(\mathcal{A}+\lambda \mathcal{B})}
    -
    e^{\alpha\mathcal{A}}
    )O
    \|
    \,,
\end{align}
where we used the Campbell identity
\begin{align}
    e^{\alpha H}Oe^{-\alpha H}
    =
    e^{\alpha \,\mathrm{ad}_H}O\,.
\end{align}
Duhamel's formula with $A=\alpha\mathcal{A}$, $B=
\alpha\lambda \mathcal{B}$, and $ t=1$ gives:
\begin{align}
    f(\lambda,\alpha)
    =
    |\lambda|
    \, |\alpha|\,
    \left\|
    \int_0^1
    e^{\alpha(\mathcal{A}+\lambda \mathcal{B})(1-s)} \mathcal{B}e^{ \alpha \mathcal{A}s}
    O\, \dd s
    \right\|
    \le 
    |\lambda|
    \, |\alpha|\,
    \int_0^1
    \left\|
    e^{\alpha(\mathcal{A}+\lambda \mathcal{B})(1-s)} \mathcal{B}e^{ \alpha \mathcal{A}s}
    O
    \right\|\ \dd s
\end{align}
Now for any operators $H$ and $O'$, and $\alpha\in\mathbb{C}$, we have
\begin{align}
    \|
    e^{\alpha \mathrm{ad}_H}
    O'
    \|
    =
    \|
    e^{\alpha H}
    O'
    e^{-\alpha H}
    \|
    =
    \|
    O'
    \| \,,
\end{align}
because $e^{\alpha H}O'e^{-\alpha H}$ has the same spectrum of $O'$.
Thus
\begin{align}
    f(\lambda,\alpha)
    \le 
    |\lambda|
    \,|\alpha|\,
    \int_0^1
    \left\|
    \mathcal{B}e^{\alpha \mathcal{A}s}
    O
    \right\|\ \dd s
    \le 
    |\lambda|
    \,|\alpha|\,
    \max_{s\in [0,1]}
    \left\|
    \mathcal{B}e^{\alpha \mathcal{A}s}
    O
    \right\|\,.
\end{align}
Plugging in the definitions of $\mathcal{A}$ and $\mathcal{B}$ gives us the result of the lemma.
\end{proof}

\begin{lemma}\label{lemma:contour change}
For the Lindblad operators $L_a$ as defined in \eqref{eq:jump_op_int_def}, we can express their conjugation by the Gibbs state appearing in the parent Hamiltonian as
\begin{equation}
    \sigma_\beta ^{-1/4} L_a \sigma_\beta ^{1/4} = \int_{-\infty}^\infty f^a(t + i\beta/4) \cdot e^{iHt} A^a e^{-iHt}\ \dd t\,.
\end{equation}
Similarly for the (individual parts of the) coherent term defined in \eqref{eqn:Coherent term}: \begin{equation}
     \sigma_\beta ^{-1/4} G_a \sigma_\beta ^{1/4} = \int_{-\infty}^\infty g(t + i\beta/4) \cdot e^{iHt} L_a ^\dagger L_a e^{-iHt}\ \dd t\,.
\end{equation} 
\end{lemma}

\begin{proof}
    We need to change the contour of integration in the complex plane. For that, consider the integral \begin{equation}
        I = \oint_\Gamma f^a(i\beta/4 - iz) \cdot e^{Hz} A^a e^{-Hz}\ \dd z\,,
    \end{equation} where $\Gamma = \Gamma_1 + \Gamma_2 + \Gamma_T + \Gamma_B$ is the rectangular contour of integration, with $\Gamma_1 = \{z = \beta/4 + it\ |\ -R < t < R\}$, $\Gamma_2 = \{z = it\ |\ R > t > -R\}$, $\Gamma_T =\{ z = iR + t\ |\ \beta/4 > t > 0 \}$, and $\Gamma_B = \{ z = -iR + t\ |\ 0 < t < \beta/4 \}$; in the limit $R \to \infty$. As the function of interest is holomorphic, by the functional version of Cauchy's integral theorem we get that $I = 0$. We can also find that \begin{align}
        I_1 &= i \cdot \sigma_\beta ^{-1/4} L_a \sigma_\beta ^{1/4}\,,\\
        I_2 &= -i \int_{-\infty}^\infty f^a(t + i\beta/4) \cdot e^{iHt} A^a e^{-iHt}\ \dd t\,,
    \end{align} in the limit $R \to \infty$.

    Now we want to show that the contributions $I_T$ and $I_B$ vanish. To do that, we can use the fact that conjugation preserves the spectral norm and that $\|A^a\| \leq 1$ to bound \begin{equation}
        \|I_T\| \leq \lim_{R\to \infty} \int^{\beta/4}_0 |f^a(i\beta/4 - it + R)|\ \dd t = 0\,,
    \end{equation} which shows that $I_T$ vanishes; and similarly for $I_B$.
    Hence we get that \begin{equation}
    \sigma_\beta ^{-1/4} L_a \sigma_\beta ^{1/4} = \int_{-\infty}^\infty f^a(t + i\beta/4) \cdot e^{iHt} A^a e^{-iHt}\ \dd t\,.
\end{equation}
    The result for $G$ follows likewise. Alternatively, to avoid any potential issues with the (unspecified) smoothened indicator function $\kappa(\nu)$ appearing in the definition of $\hat g(\nu)$, we present the following functional argument:

    Observe that $\operatorname{ad}_H = -\frac{1}{\beta} \log (\Delta_{\sigma_\beta})$, where $\Delta_\rho [X] = \rho X \rho^{-1}$ is the modular superoperator. Hence the coherent term can be equivalently expressed as \begin{equation}
        G = i \tanh\left(\frac{\beta}{4} \operatorname{ad}_H \right) \left( \frac{1}{2} \sum_{a \in \mathcal{A}} L^\dagger _a L_a\right)\,.
    \end{equation} Now we can equate \begin{align}
        \sigma_\beta ^{-1/4} G \sigma_\beta ^{1/4} = e^{\frac{\beta}{4}\operatorname{ad}_H} G &= i e^{\frac{\beta}{4}\operatorname{ad}_H}\tanh\left(\frac{\beta}{4} \operatorname{ad}_H \right) \left( \frac{1}{2} \sum_{a \in \mathcal{A}} L^\dagger _a L_a\right) \\ 
        &= \hat d(\operatorname{ad}_H)\left( \sum_{a \in \mathcal{A}} L^\dagger _a L_a\right)\,,
    \end{align} with $\hat d (\nu) = \frac{i}{2} e^{\beta \nu /4} \tanh(\beta \nu /4) \kappa(\nu)$, where we've again introduced the smooth indicator function $\kappa(\nu)$, as $\kappa(\nu) = 1$ for any $\nu \in \operatorname{spec}(\operatorname{ad}_H)$. Finally we get that \begin{equation}
     \sigma_\beta ^{-1/4} G\sigma_\beta ^{1/4} = \int_{-\infty}^\infty g(t + i\beta/4) \cdot \sum_{a \in \mathcal{A}} e^{iHt} L_a ^\dagger L_a e^{-iHt}\ \dd t
\end{equation} by using the shifting property of the Fourier transform.
\end{proof}


\section{Details on Bounds of Lindbladian Perturbation}
\label{appendix:detailed bounds}

The full perturbation of the parent Hamiltonian we wish to study is
    \begin{align}
        \mathcal{V}[\rho] &= \mathcal{H}[\rho] - \mathcal{H}_0[\rho]\\
        &= \sigma_{\beta} ^{-1/4} \cdot\mathcal{L}^\dagger [\sigma_{\beta} ^{1/4} \cdot\rho\cdot \sigma_{\beta} ^{1/4}] \cdot \sigma_{\beta} ^{-1/4} - \sigma_{\beta,0} ^{-1/4} \cdot\mathcal{L}_0^\dagger [\sigma_{\beta,0} ^{1/4} \cdot\rho\cdot \sigma_{\beta,0} ^{1/4}] \cdot \sigma_{\beta,0} ^{-1/4}\\
        &=\sigma_{\beta} ^{-1/4}  \left(-i[G,\sigma_{\beta} ^{1/4} \cdot\rho\cdot \sigma_{\beta} ^{1/4}] + \sum_{a \in \mathcal{A}} \left( L_a \sigma_{\beta} ^{1/4} \cdot\rho\cdot \sigma_{\beta} ^{1/4} L_a ^\dagger - \frac{1}{2} \{ L_a ^\dagger L_a, \sigma_{\beta} ^{1/4} \cdot\rho\cdot \sigma_{\beta} ^{1/4} \} \right) \right) \sigma_{\beta} ^{-1/4} \\
        &\quad - \sigma_{\beta,0} ^{-1/4}  \left(-i[G^0,\sigma_{\beta,0} ^{1/4} \cdot\rho\cdot \sigma_{\beta,0} ^{1/4}] + \sum_{a \in \mathcal{A}} \left( L^0_a \sigma_{\beta,0} ^{1/4} \cdot\rho\cdot \sigma_{\beta,0} ^{1/4} L_a ^{0\dagger} - \frac{1}{2} \{ L_a ^{0\dagger} L^0 _a, \sigma_{\beta,0} ^{1/4} \cdot\rho\cdot \sigma_{\beta,0} ^{1/4} \} \right) \right) \sigma_{\beta,0} ^{-1/4}\\
        &= -i \left( \sigma_{\beta} ^{-1/4} G \sigma_{\beta} ^{1/4} -  \sigma_{\beta,0} ^{-1/4} G^0 \sigma_{\beta,0} ^{1/4}\right)\cdot \rho + i \rho \cdot \left( \sigma_{\beta} ^{1/4} G \sigma_{\beta} ^{-1/4} - \sigma_{\beta,0} ^{1/4} G^0 \sigma_{\beta,0} ^{-1/4}\right)\\
        &\quad + \sum_{a \in \mathcal{A}} \left( \sigma_{\beta} ^{-1/4} L_a \sigma_{\beta} ^{1/4} \cdot\rho\cdot \sigma_{\beta} ^{1/4} L_a ^\dagger \sigma_{\beta} ^{-1/4} - \sigma_{\beta,0} ^{-1/4} L^0_a \sigma_{\beta,0} ^{1/4} \cdot\rho\cdot \sigma_{\beta,0} ^{1/4} L_a ^{0\dagger} \sigma_{\beta,0} ^{-1/4}\right)\\
        &\quad - \frac{1}{2} \sum_{a \in \mathcal{A}} \left( \sigma_{\beta} ^{-1/4} L_a ^\dagger L_a \sigma_{\beta} ^{1/4} \cdot\rho + \rho\cdot \sigma_{\beta} ^{1/4} L_a ^\dagger L_a \sigma_{\beta} ^{-1/4} - \sigma_{\beta,0} ^{-1/4} L_a ^{0\dagger} L^0_a \sigma_{\beta,0} ^{1/4} \cdot\rho - \rho\cdot \sigma_{\beta,0} ^{1/4} L_a ^{0\dagger} L^0_a \sigma_{\beta,0} ^{-1/4} \right),
    \end{align} from which we can consider the vectorised operator version of $\mathcal{V}$, obtained via mapping $|\psi\rangle\langle \phi| \to |\psi\rangle|\overline \phi\rangle$ and $O[\rho] = A \rho B \to O \cong A \otimes B^T$:
    \begin{align}
        \mathcal{V} &\cong -i \left( \sigma_{\beta} ^{-1/4} G \sigma_{\beta} ^{1/4} -  \sigma_{\beta,0} ^{-1/4} G^0 \sigma_{\beta,0} ^{1/4}\right)\otimes I + i I \otimes \overline{\left( \sigma_{\beta} ^{1/4} G \sigma_{\beta} ^{-1/4} - \sigma_{\beta,0} ^{1/4} G^0 \sigma_{\beta,0} ^{-1/4}\right)}\\
        &\quad + \sum_{a \in \mathcal{A}} \left( \sigma_{\beta} ^{-1/4} L_a \sigma_{\beta} ^{1/4} \otimes \overline{\sigma_{\beta} ^{-1/4} L_a \sigma_{\beta} ^{1/4}} - \sigma_{\beta,0} ^{-1/4} L^0_a \sigma_{\beta,0} ^{1/4} \otimes \overline{\sigma_{\beta,0} ^{-1/4} L_a ^{0} \sigma_{\beta,0} ^{1/4}}\right)\\
        &\quad - \frac{1}{2} \sum_{a \in \mathcal{A}} \left( \sigma_{\beta} ^{-1/4} L_a ^\dagger L_a \sigma_{\beta} ^{1/4} \otimes I + I\otimes \overline{\sigma_{\beta} ^{1/4} L_a ^\dagger L_a \sigma_{\beta} ^{-1/4}} - \sigma_{\beta,0} ^{-1/4} L_a ^{0\dagger} L^0_a \sigma_{\beta,0} ^{1/4} \otimes I - I \otimes \overline{\sigma_{\beta,0} ^{1/4} L_a ^{0\dagger} L^0_a \sigma_{\beta,0} ^{-1/4}} \right)\,.
    \end{align}

We wish to show that $\mathcal{V}$ has $(c|\lambda|,\mu)$-decay, and we already understand the quasi-locality of this operator given Proposition \ref{prop: locality}; hence we can focus only on upper bounding its strength. We can start by considering $\| \sigma_{\beta} ^{-1/4} L_a \sigma_{\beta} ^{1/4} - \sigma_{\beta,0} ^{-1/4} L^0_a \sigma_{\beta,0} ^{1/4} \|$, and after using the definition of $L_a$'s, we will need to upper bound the following expression to obtain equation \eqref{eqn:bound on L - L}, where we make use of Lemma \ref{lemma:bound on perturbed evolution}: \begin{align}
    \left\| \left( e^{H(\beta/4 + it)} A^a e^{-H(\beta/4 + it)} - e^{H_0(\beta/4 + it)} A^a e^{H_0(\beta/4 + it)} \right) \right\| &\leq |\lambda| |\beta/4 + it| \max_{s\in [0,1]}
    \|[V, e^{s(\beta/4 + it) H_0}A^a e^{-s(\beta/4 + it) H_0}]\|\\
    &= |\lambda| |\beta/4 + it| \max_{s\in [0,1]} \|[V, e^{s(\beta/4 + it) \operatorname{ad}_{H_0}}\omega_a]\|\\
    &= |\lambda| |\beta/4 + it| \max_{s\in [0,1]} \left\|\left[V, \sum_i \left(e^{-4s(\beta/4 + it) h}\right)_{ai}\omega_i\right]\right\|\\
    &\leq |\lambda| |\beta/4 + it| \max_{s\in [0,1]} \left\|e^{-4s(\beta/4 + it) h} \right\|_\infty \max_i \|[V,\omega_i]\|\\
    &\leq  |\lambda| |\beta/4 + it| e^{4|\beta/4 + it| \cdot \|h\|_\infty} \max_i \|[V,\omega_i]\|\\
    &\leq c_2 |\lambda| \cdot |\beta/4 + it|\cdot e^{c_3|\beta/4 + it|}\,,
\end{align} where we then first utilised the exact solution for time evolution of $\omega_a$, then upper bounded a weighted sum by the maximal absolute sum of the weights multiplied by the maximal element, and finally used the submultiplicativity of the $\ell_\infty$ norm, representing the maximal absolute row sum, to obtain a system-size-independent bound due to the assumption $\|h\|_\infty = \mathcal{O}(1)$. This expression is subsequently integrated over the Gaussian filter function $f(t)$, which is convergent, and gives us that $\| \sigma_{\beta} ^{-1/4} L_a \sigma_{\beta} ^{1/4} - \sigma_{\beta,0} ^{-1/4} L^0_a \sigma_{\beta,0} ^{1/4} \| \leq c_1 |\lambda|$. The bounds on the different products of $L_a$'s then follow immediately from this one.

Then we similarly needed to bound $\| \sigma_{\beta} ^{-1/4} G_a \sigma_{\beta} ^{1/4} -  \sigma_{\beta,0} ^{-1/4} G^0_a \sigma_{\beta,0} ^{1/4}\|$:
\begin{align}
    \| \sigma_{\beta} ^{-1/4} &G_a \sigma_{\beta} ^{1/4} - \sigma_{\beta,0} ^{-1/4} G^0_a \sigma_{\beta,0} ^{1/4}\| = \| e^{H\beta/4} G_a e^{-H\beta/4}  - e^{H_0 \beta/4} G^0_a e^{-H_0 \beta/4} \|\\
    &= \left\|\int_{-\infty}^\infty g(t) \cdot \left( e^{H(\beta/4 + it)} L^\dagger_a L_a e^{-H(\beta/4 + it)} -  e^{H_0(\beta/4 + it)} L^{0\dagger}_a L^0_a e^{-H_0(\beta/4 + it)} \right)\ \dd t\right\|\\
    &= \left\|\int_{-\infty}^\infty g(t) \cdot \left( e^{H(\beta/4 + it)} L^\dagger_a L_a e^{-H(\beta/4 + it)} - e^{H(\beta/4 + it)} L^{0\dagger}_a L^0_a e^{-H(\beta/4 + it)}\right.\right. \\ &\qquad \qquad \left. \vphantom{\int_1^2} \left. + e^{H(\beta/4 + it)} L^{0\dagger}_a L^0_a e^{-H(\beta/4 + it)} -  e^{H_0(\beta/4 + it)} L^{0\dagger}_a L^0_a e^{-H_0(\beta/4 + it)} \right)\ \dd t\right\|\\
    &\leq \int_{-\infty}^\infty |g(t)| \left(\left\| L^\dagger_a L_a - L^{0\dagger}_a L^0 _a \right\| + \left\| e^{H(\beta/4 + it)} L^{0\dagger}_a L^0_a e^{-H(\beta/4 + it)} -  e^{H_0(\beta/4 + it)} L^{0\dagger}_a L^0_a e^{-H_0(\beta/4 + it)}  \right\|\right) \dd t\,.
\end{align}
Here we will utilise the exact solution $L^0_a = \sum_i \hat{f}(-4h)_{ai} \omega_i$ to analogously proceed with the following upper bound:
\begin{align}
    \max_{s\in [0,1]}\|[V, e^{s(\beta/4 + it) H_0} L^{0\dagger}_a &L^0_a e^{-s(\beta/4 + it) H_0}]\| \leq \| \hat{f}(-4h) \|_\infty ^2 \max_{i,j} \max_{s\in [0,1]}  \left\|\left[V, e^{s(\beta/4 + it) H_0} \omega_i \omega_j e^{-s(\beta/4 + it) H_0}\right]\right\|\\
    &\leq 2 \| \hat{f}(-4h) \|_\infty ^2 \max_{i} \max_{s\in [0,1]}  \left\|\left[V, e^{s(\beta/4 + it) H_0} \omega_i e^{-s(\beta/4 + it) H_0}\right]\right\|\\
    &= 2 \| \hat{f}(-4h) \|_\infty ^2 \max_{i} \max_{s\in [0,1]} \left\|\sum_{k} \left(e^{-4s(\beta/4+it)h}\right)_{ik} \left[V,\omega_k \right]\right\|\\
    &\leq  2 \| \hat{f}(-4h) \|_\infty ^2  \max_{s\in [0,1]} \left\|e^{-s\beta h} \right\|_\infty \left\|e^{-4isth} \right\|_\infty \max_{k}\left\|\left[V,\omega_k \right]\right\|\\
    &\leq 2  \| \hat{f}(-4h) \|_\infty ^2 \cdot e^{\beta \|h\|_\infty} \cdot w_h(t) \cdot \max_k \left\|\left[V,\omega_k\right]\right\|\\
    &\leq 2c_2 e^{c_3 \beta/4} \cdot w_h(t) \cdot \| \hat{f}(-4h) \|_\infty ^2\,,
    \end{align}
    which is then also system-size-independent due to submultiplicativity of the norm. Here, the $w_h(t) \geq \max_{s \in [0,1]} \left\|e^{-4isth} \right\|_\infty$ represents a function independent of the system size which grows subexponentially in $t$. A priori, we can bound $\max_{s \in [0,1]} \left\|e^{-4isth} \right\|_\infty \leq e^{4|t|\|h\|_\infty}$, which is system size independent due to $\|h\|_\infty = \mathcal{O}(1)$; but it can cause convergence issues in the integral weighted by $g(t)$. However, by Gelfand's spectral radius formula, we get that $1 = \lim_{t \to \infty} \|e^{-4isth}\|_\infty ^{1/t}$ due to the orthogonality, which then implies that the growth has to be subexponential, as otherwise the limit would have to be greater than $1$. One can find that for example in the 1D spinless case we consider in Section \ref{sec:simulations}, the norm grows like $\|e^{ith}\|_\infty \sim \sqrt{\pi t}$, and so these bounds are actually quite loose for the sparse systems we consider. Hence integrating this expression over $g(t)$ is also convergent, and we arrive at $\| \sigma_{\beta} ^{-1/4} G_a \sigma_{\beta} ^{1/4} -  \sigma_{\beta,0} ^{-1/4} G^0_a \sigma_{\beta,0} ^{1/4}\| \leq c_6 |\lambda|$.

    While these results depend on the specifics of the Hamiltonian considered, we may state the following result, which only requires exponentially decaying correlations in the underlying system:

\begin{lemma}[General bound on decay of the perturbation]\label{lemma:general bound decay}
    For a Hamiltonian $H = H_0 + \lambda V$ with exponentially decaying correlations (as per the assumptions of Lieb-Robinson bounds discussed in Section \ref{sec: locality}), there exists a constant $\lambda_\textup{bound}$, such that for any $|\lambda| \leq \lambda_\textup{bound}$ the perturbation of the parent Hamiltonian corresponding to the Lindbladian has $(K,\mu)$-decay, where $K \leq c |\lambda|^\alpha$  for an arbitrary positive constant $\alpha < 1$, with constants $c$ and $\mu$ being independent of the system size and $\lambda$.
\end{lemma}

\begin{proof}
    Define the truncated versions of the operators appearing in the parent Hamiltonian, 
    \begin{align}
        \tilde{L}_a^{(r)}
        &=
        \int_{-\infty} ^{\infty} f^a (t+i\beta/4)  e^{iH_{B_r(a)}t} A^a e^{-iH_{B_r(a)}t}\ \dd t 
        \,,\\
        L_a^{(r)}
        &=
        \int_{-\infty} ^{\infty} f^a (t)  e^{iH_{B_r(a)}t} A^a e^{-iH_{B_r(a)}t}\ \dd t 
        \,,\\
        \tilde{G}_a^{(r)}
        &=
        \int_{-\infty} ^{\infty} g (t+i\beta/4)  e^{iH_{B_r(a)}t} L^{(r)\dagger}_a L^{(r)}_a e^{-iH_{B_r(a)}t}\ \dd t 
        \,,
    \end{align} which are supported only on the balls centred at $a$ with radius $r$; and similarly $\tilde{L}_a^{0(r)}$, $L_a^{0(r)}$, and $\tilde{G}_a^{0(r)}$ for the versions corresponding to the unperturbed Hamiltonian $H_0$. The quasi-locality proved in Proposition \ref{prop: locality} shows immediately that \begin{equation}
        \| \tilde{L}_a^{(r+1)} -  \tilde{L}_a^{(r)}\| \leq c_1 e^{-\mu_1 r}\,,
    \end{equation} where $c_1$ and $\mu_1$ are independent of the system size. Note that these generally depend on the coupling $\lambda$, but due to their independence of system size, they must be continuous and finite for any finite $\lambda$, and so we can say that there exists some $\lambda_\textup{bound}$ below which these bounds hold for constants $c_1$ and $\mu_1$ which are also independent of $\lambda$. The same bound then also holds for $\lambda = 0$ for the case $\| \tilde{L}_a^{0(r+1)} -  \tilde{L}_a^{0(r)}\|$. 

    Now consider bounding $ \| \tilde{L}_a^{(r)} -   \tilde{L}_a^{0(r)}\|$: \begin{align}
         \| \tilde{L}_a^{(r)} -   \tilde{L}_a^{0(r)}\| &= \left\|  \int_{-\infty} ^{\infty} f^a (t+i\beta/4)  (e^{iH_{B_r(a)}t} A^a e^{-iH_{B_r(a)}t} - e^{iH^0_{B_r(a)}t} A^a e^{-iH^0_{B_r(a)}t})\ \dd t  \right\| \\
         &\leq  \int_{-\infty} ^{\infty} |f^a (t+i\beta/4)|  \left\| e^{iH_{B_r(a)}t} A^a e^{-iH_{B_r(a)}t} - e^{iH^0_{B_r(a)}t} A^a e^{-iH^0_{B_r(a)}t})\right\| \ \dd t  \\
         &\leq  \int_{-\infty} ^{\infty} |f^a (t+i\beta/4)|  |\lambda| |t| \max_{s\in [0,1]}
    \left\|[V_{B_r(a)}, e^{ist H^0_{B_r(a)}}A^a e^{-ist H^0_{B_r(a)}}]\right\|\\
    &\leq  \int_{-\infty} ^{\infty} |f^a (t+i\beta/4)|  \cdot |\lambda| |t| \cdot 2 \|V_{B_r(a)}\|\\
    &\leq c_2 r^D |\lambda|\,,
    \end{align} where we've used Lemma \ref{lemma:bound on perturbed evolution}, $\|A^a\| \leq 1$ and $\|V_{B_r(a)}\| = \mathcal{O}(r^D)$. Hence it follows that \begin{align}
        \|\tilde{L}_a ^{(r+1)} \otimes \tilde{L}_a ^{(r+1)} - \tilde{L}_a ^{(r)}\otimes \tilde{L}_a ^{(r)} \| &\leq 2 c_1 e^{-\mu_1 r}\,,\\
        \|\tilde{L}_a ^{0(r+1)} \otimes \tilde{L}_a ^{0(r+1)} - \tilde{L}_a ^{0(r)}\otimes \tilde{L}_a ^{0(r)} \| &\leq 2 c_1 e^{-\mu_1 r}\,,\\
        \| \tilde{L}_a^{(r)} \otimes \tilde{L}_a^{(r)} -  \tilde{L}_a^{0(r)} \otimes  \tilde{L}_a^{0(r)}\| &\leq 2 c_2 r^D |\lambda|\,,\\
        \| \tilde{L}_a^{(r+1)} \otimes \tilde{L}_a^{(r+1)} -  \tilde{L}_a^{0(r+1)} \otimes  \tilde{L}_a^{0(r+1)}\| &\leq 2 c_2 (r+1)^D |\lambda| \leq 2^{D+1}c_2 r^D |\lambda|\,.\\        
    \end{align} We can then combine these bounds to say \begin{align}
        \|\tilde{L}_a ^{(r+1)} \otimes \tilde{L}_a ^{(r+1)}  - \tilde{L}_a ^{0(r+1)} \otimes \tilde{L}_a ^{0(r+1)}  - \tilde{L}_a ^{(r)}\otimes \tilde{L}_a ^{(r)}+ \tilde{L}_a ^{0(r)}\otimes \tilde{L}_a ^{0(r)} \| &\leq 4 c_1 e^{-\mu_1 r}\,,\\
         \|\tilde{L}_a ^{(r+1)} \otimes \tilde{L}_a ^{(r+1)}  - \tilde{L}_a ^{0(r+1)} \otimes \tilde{L}_a ^{0(r+1)} - \tilde{L}_a ^{(r)}\otimes \tilde{L}_a ^{(r)} + \tilde{L}_a ^{0(r)}\otimes \tilde{L}_a ^{0(r)} \| &\leq (2+2^{D+1})c_2 r^D |\lambda|\,. 
    \end{align} Finally, as we obtained two different bounds for the same operator, we can unify them by taking their weighted geometric mean to get 
    \begin{align}
        \|\tilde{L}_a ^{(r+1)} \otimes \tilde{L}_a ^{(r+1)}  - \tilde{L}_a ^{0(r+1)} \otimes \tilde{L}_a ^{0(r+1)}  - \tilde{L}_a ^{(r)}\otimes \tilde{L}_a ^{(r)}+ \tilde{L}_a ^{0(r)}\otimes \tilde{L}_a ^{0(r)} \| &\leq (4c_1 e^{-\mu_1 r})^{1-\alpha} \left((2+2^{D+1})c_2r^D|\lambda|\right)^{\alpha}\\ 
        &\leq c_3(\alpha) |\lambda|^\alpha e^{-\mu_2(\alpha) r} \,,
    \end{align} for an arbitrary positive constant $\alpha <1$. Denoting this operator by $\epsilon_a ^{1(r+1)}$, hence writing \begin{equation}
        \mathcal{V}^1= \sum_{a \in \mathcal{A}}  \left( \mathcal{V}_a^{1(0)} + \sum_{r \geq 1} \mathcal{V}_a^{1(r+1)} - \mathcal{V}_a^{1(r)} \right) = \sum_{a \in \mathcal{A}} \sum_{r \geq 0} \epsilon_a^{1(r)} \,,
    \end{equation} we see that the first part of the perturbation $\mathcal{V}$ obeys the result of the lemma (where we made a shortcut by writing $\epsilon_a^{1(0)} = \tilde{L}_a ^{(0)} \otimes \tilde{L}_a ^{(0)}  - \tilde{L}_a ^{0(0)} \otimes \tilde{L}_a ^{0(0)}$); here we have separated the perturbation into three parts like    \begin{align}
        \mathcal{V}^1 &=  \sum_{a \in \mathcal{A}} \left( \sigma_{\beta} ^{-1/4} L_a \sigma_{\beta} ^{1/4} \otimes \overline{\sigma_{\beta} ^{-1/4} L_a \sigma_{\beta} ^{1/4}} - \sigma_{\beta,0} ^{-1/4} L^0_a \sigma_{\beta,0} ^{1/4} \otimes \overline{\sigma_{\beta,0} ^{-1/4} L_a ^{0} \sigma_{\beta,0} ^{1/4}}\right)\\
        \mathcal{V}^2 &= - \frac{1}{2} \sum_{a \in \mathcal{A}} \left( \sigma_{\beta} ^{-1/4} L_a ^\dagger L_a \sigma_{\beta} ^{1/4} \otimes I + I\otimes \overline{\sigma_{\beta} ^{1/4} L_a ^\dagger L_a \sigma_{\beta} ^{-1/4}} - \sigma_{\beta,0} ^{-1/4} L_a ^{0\dagger} L^0_a \sigma_{\beta,0} ^{1/4} \otimes I - I \otimes \overline{\sigma_{\beta,0} ^{1/4} L_a ^{0\dagger} L^0_a \sigma_{\beta,0} ^{-1/4}} \right)\\
        \mathcal{V}^3 &= -i \left( \sigma_{\beta} ^{-1/4} G \sigma_{\beta} ^{1/4} -  \sigma_{\beta,0} ^{-1/4} G^0 \sigma_{\beta,0} ^{1/4}\right)\otimes I + i I \otimes \overline{\left( \sigma_{\beta} ^{1/4} G \sigma_{\beta} ^{-1/4} - \sigma_{\beta,0} ^{1/4} G^0 \sigma_{\beta,0} ^{-1/4}\right)}\,.
    \end{align}

    For the second part, first note that the bound $\| \tilde{L}_a^{(r)} -   \tilde{L}_a^{0(r)}\|$ will also work if we conjugate $L_a^{(r)}$ by the Gibbs state with opposite exponents, amounting to changing $+i\beta/4$ to $-i\beta/4$ within the filter function in the integral. We may denote these two directions of conjugation by $\tilde{L}_a^{(r,+)}$ and $\tilde{L}_a^{(r,-)}$ respectively. Hence we arrive at \begin{align}
        \| \tilde{L}_a^{(r,-) \dagger} \tilde{L}_a^{(r,+) } -   \tilde{L}_a^{0(r,-) \dagger} \tilde{L}_a^{0(r,+) }\| &\leq c_4 r^D |\lambda|\,,\\
        \| \tilde{L}_a^{(r+1,-) \dagger} \tilde{L}_a^{(r+1,+) } -   \tilde{L}_a^{0(r+1,-) \dagger} \tilde{L}_a^{0(r+1,+) }\| &\leq 2 c_4 r^D |\lambda|\,,\\
    \end{align} and the quasi-locality also gives us \begin{align}
        \| \tilde{L}_a^{0(r+1,-) \dagger} \tilde{L}_a^{0(r+1,+) } -  \tilde{L}_a^{0(r,-) \dagger} \tilde{L}_a^{0(r,+) }\| &\leq c_5 e^{-\mu_3 r}\,,\\
        \| \tilde{L}_a^{(r+1,-) \dagger} \tilde{L}_a^{(r+1,+) } -  \tilde{L}_a^{(r,-) \dagger} \tilde{L}_a^{(r,+) } \| &\leq c_5 e^{-\mu_3 r}\,,\\
    \end{align} which finally leads to \begin{equation}
        \| \tilde{L}_a^{(r+1,-) \dagger} \tilde{L}_a^{(r+1,+) } -   \tilde{L}_a^{0(r+1,-) \dagger} \tilde{L}_a^{0(r+1,+) } - \tilde{L}_a^{(r,-) \dagger} \tilde{L}_a^{(r,+) } +  \tilde{L}_a^{0(r,-) \dagger} \tilde{L}_a^{0(r,+) }\| \leq c_6(\alpha) |\lambda|^\alpha e^{-\mu_4(\alpha) r}\,,
    \end{equation} by the same argument as previously. This then shows that the second part of the perturbation $\mathcal{V}$ obeys the result of the lemma.

    Lastly, we look at the coherent part. Here we start by bounding \begin{align}
        \|\tilde{G}_a^{(r)} - \tilde{G}_a^{0(r)}\| &= \left\| \int_{-\infty} ^{\infty} g (t+i\beta/4) ( e^{iH_{B_r(a)}t} L^{(r)\dagger}_a L^{(r)}_a e^{-iH_{B_r(a)}t} - e^{iH^0_{B_r(a)}t} L^{0(r)\dagger}_a L^{0(r)}_a e^{-iH^0_{B_r(a)}t})\ \dd t  \right\|\\
        &\leq \int_{-\infty} ^{\infty} |g (t+i\beta/4)| \left\| e^{iH_{B_r(a)}t} L^{(r)\dagger}_a L^{(r)}_a e^{-iH_{B_r(a)}t} - e^{iH^0_{B_r(a)}t} L^{0(r)\dagger}_a L^{0(r)}_a e^{-iH^0_{B_r(a)}t} \right\|\ \dd t \\
        &\leq \int_{-\infty} ^{\infty} |g (t+i\beta/4)| \left( \| L^{(r)\dagger}_a L^{(r)}_a - L^{0(r)\dagger}_a L^{0(r)}_a \| \right. \\ 
        &\quad\quad\quad + \left. \left\| e^{iH_{B_r(a)}t} L^{0(r)\dagger}_a L^{0(r)}_a e^{-iH_{B_r(a)}t} - e^{iH^0_{B_r(a)}t} L^{0(r)\dagger}_a L^{0(r)}_a e^{-iH^0_{B_r(a)}t} \right\| \right)\ \dd t\\
        &\leq  c_7 r^D |\lambda| +  \int_{-\infty} ^{\infty} |g (t+i\beta/4)| |t| |\lambda| \max_{s \in [0,1]} \left\| \left[V_{B_r(a)},e^{iH^0_{B_r(a)}st} L^{0(r)\dagger}_a L^{0(r)}_a e^{-iH^0_{B_r(a)}st}\right] \right\| \dd t\\
        &\leq c_8 r^D |\lambda|\,,
    \end{align} where we've used that $\|L_a\| \leq 1$ due to the normalisation of $f^a(t)$, the fact that the previous bound on $\| \tilde{L}_a^{0(r+1,-) \dagger} \tilde{L}_a^{0(r+1,+) } -  \tilde{L}_a^{0(r,-) \dagger} \tilde{L}_a^{0(r,+) }\|$ also immediately works without the conjugation, and again that $\|V_{B_r(a)}\| = \mathcal{O}(r^D)$. Together with the quasi-locality of $\tilde{G}_a^{(r)}$, the same argument as before shows \begin{equation}
       \| \tilde{G}_a^{(r+1)} - \tilde{G}_a^{0(r+1)} - \tilde{G}_a^{(r)} - \tilde{G}_a^{0(r)}\| \leq c_9(\alpha) |\lambda|^\alpha e^{-\mu_5(\alpha) r}\,,
    \end{equation} which leads to the third and final part of $\mathcal{V}$ to obey the result of the lemma, meaning it holds for the full $\mathcal{V}$, as \begin{equation}
        \mathcal{V} = \sum_{a \in \mathcal{A}} \sum_{t=1}^3 \sum_{r\geq 0} \epsilon_a^{t(r)},
    \end{equation} where $\|\epsilon_a^{t(r)}\| \leq c(\alpha)|\lambda|^\alpha e^{-\mu(\alpha) r}$; hence finishing the proof.
\end{proof}


\end{document}